\documentclass[letterpaper,10pt]{article}
\usepackage[tt=true]{libertine}
\usepackage[T1]{fontenc}
\usepackage[margin=1in]{geometry}
\usepackage{multicol}
\usepackage{authblk}
\usepackage{abstract}
\usepackage{appendix}
\usepackage{subcaption}
\makeatletter
\renewcommand\AB@affilsepx{, \protect\Affilfont}
\makeatother
\usepackage{cancel}
\usepackage{amsmath}
\usepackage{amssymb}
\usepackage{amsthm}
\usepackage{mathtools}
\usepackage{bm}
\usepackage{dsfont}
\usepackage{bbding}
\usepackage{slantsc}
\usepackage{thmtools, thm-restate}
\usepackage{tikz}
\usepackage{capt-of}
\usepackage{caption}
\usepackage{algorithm}
\usepackage[noend]{algpseudocode}
\usepackage{paralist}
\usepackage{enumitem}
\usepackage{url}
\usepackage{xspace}
\usepackage{soul}
\usepackage{booktabs}
\usepackage{multirow}
\usepackage{pifont}
\usepackage{setspace}

\usetikzlibrary{shapes}
\usetikzlibrary{arrows}
\usetikzlibrary{calc}
\usetikzlibrary{shadows}

\makeatletter
\algnewcommand{\LineComment}[1]{\Statex \hskip\ALG@thistlm {\color{gray}\textrm{// #1}}}
\makeatother
\algnewcommand{\SectionComment}[2]{\Statex {\color{#2}\(\triangleright\) \textrm{#1}}}

\algnewcommand{\InlineComment}[1]{{\hspace{0.5em}\color{gray}\textrm{// #1}}}
\algnewcommand{\Phase}[1]{\SectionComment{#1}{myorange}}
\MakeRobust{\Call}
\DeclareCaptionFormat{algor}{%
\hrulefill\par\offinterlineskip\vskip1pt%
\textbf{#1#2}#3\offinterlineskip\hrulefill}
\DeclareCaptionStyle{algori}{singlelinecheck=off,format=algor,labelsep=space}
\makeatletter 
\renewcommand\paragraph{\@startsection{paragraph}{4}{\z@}%
{1.25ex \@plus1ex \@minus.2ex}%
{-1em}%
{\normalfont\normalsize\bfseries}}
\makeatother

\newcommand{\sharedbase}{./podc2019-final}
\newcommand{\graphbase}{./graphs}

\algnewcommand\algorithmicswitch{\textbf{switch}}
\algnewcommand\algorithmiccase{\textbf{case}}
\algnewcommand\algorithmicassert{\texttt{assert}}
\algnewcommand\Assert[1]{\State \algorithmicassert(#1)}%
\algdef{SE}[SWITCH]{Switch}{EndSwitch}[1]{\algorithmicswitch\ #1\ \algorithmicdo}{\algorithmicend\ \algorithmicswitch}%
\algdef{SE}[CASE]{Case}{EndCase}[1]{\algorithmiccase\ #1}{\algorithmicend\ \algorithmiccase}%
\algtext*{EndSwitch}%
\algtext*{EndCase}
\algnewcommand\algorithmiccontinue{\textbf{continue}}
\algnewcommand\algorithmicbreak{\textbf{break}}
\algnewcommand\Continue{\algorithmiccontinue}
\algnewcommand\Break{\algorithmicbreak}

\newtheorem{theorem}{Theorem}
\newtheorem{lemma}[theorem]{Lemma}

\theoremstyle{definition}

\newtheorem{notation}{Notation}
\algnewcommand{\IIf}[1]{\State\algorithmicif\ #1\ \algorithmicthen}
\algnewcommand{\EndIIf}{\unskip\ \algorithmicend\ \algorithmicif}

\DeclareMathOperator*{\argmax}{\arg\max} 
\DeclareMathOperator*{\argmin}{\arg\min} 

\usetikzlibrary{calc}

\definecolor{myblue}{HTML}{0072bd}
\definecolor{myorange}{HTML}{d95319}
\definecolor{mybrown}{HTML}{ad2c0a}

\newcommand{\acset}{V}
\newcommand{\acsetof}[1]{\acset[#1]}

\newcommand{\Lvl}{b}
\newcommand{\lvl}{w}
\newcommand{\Lvlx}{b}
\newcommand{\Lvlxx}{b'}
\newcommand{\Lvlxxx}{b''}
\newcommand{\lvly}{w}
\newcommand{\lvlyy}{w'}
\newcommand{\lvlyyy}{w''}
\newcommand{\Lqc}{\Lvl_\mathit{lock}}
\newcommand{\Lnew}{\Lvl_\mathit{new}}
\newcommand{\Ltail}{\Lvl_\mathit{leaf}}

\newcommand{\QCedge}{(\Leftarrow \land \gets)}
\newcommand{\qcedge}{\Leftarrow}
\newcommand{\pathfrom}{\stackrel{*}{\gets}}
\newcommand{\nopathfrom}{\stackrel{*}{\cancel{\gets}}}
\newcommand{\vheight}{\mathit{vheight}}

\newcommand{\Lexec}{\Lvl_\mathit{exec}}
\newcommand{\Cmd}{\mathit{cmd}}

\ifdefined\submission
   \newcommand{\remark}[1]{}
\else
    \newcommand{\remark}[1]{#1}
\fi

\newcommand{\HotStuffBasic}{Basic \HotStuff}
\newcommand{\HotStuffPro}{Chained \HotStuff}

\newcommand{\Gadget}{Pacemaker\xspace}

\newcommand{\sigSign}[1]{\ensuremath{\mathit{sign}_{#1}}\xspace}
\newcommand{\sigVerify}[1]{\ensuremath{\mathit{verify}_{#1}}\xspace}
\newcommand{\tsigSign}[1]{\ensuremath{\mathit{tsign}_{#1}}\xspace}
\newcommand{\tsigCombine}{\ensuremath{\mathit{tcombine}}\xspace}
\newcommand{\tsigVerify}{\ensuremath{\mathit{tverify}}\xspace}
\newcommand{\sig}{\ensuremath{\sigma}\xspace}
\newcommand{\partialResult}[1]{\ensuremath{\rho}_{#1}\xspace}
\newcommand{\cryptoHashFn}{\ensuremath{h}\xspace}

\newcommand{\gViewNumber}{\mathit{curView}}
\newcommand{\gBranch}{\mathit{curProposal}}
\newcommand{\fViewNumber}[1]{{#1}\mathit{.viewNumber}}
\newcommand{\fBranch}[1]{{#1}\mathit{.node}}
\newcommand{\fType}[1]{{#1}\mathit{.type}}
\newcommand{\fCert}[1]{{#1}\mathit{.sig}}
\newcommand{\fHeight}[1]{{#1}\mathit{.height}}
\newcommand{\fJustify}[1]{{#1}\mathit{.justify}}
\newcommand{\fParent}[1]{{#1}\mathit{.parent}}
\newcommand{\fCmd}[1]{{#1}\mathit{.cmd}}

\newcommand{\msgViewNumber}[1]{\fViewNumber{#1}}
\newcommand{\msgType}[1]{\fType{#1}}
\newcommand{\msgBranch}[1]{\fBranch{#1}}
\newcommand{\msgQC}[1]{#1\mathit{.justify}}
\newcommand{\msgSig}[1]{#1\mathit{.partialSig}}

\newcommand{\msgQCViewNumber}[1]{\fViewNumber{\msgQC{#1}}}
\newcommand{\msgQCBranch}[1]{\fBranch{\msgQC{#1}}}

\newcommand{\cmd}{\mathit{cmd}}
\newcommand{\type}{\mathit{type}}
\newcommand{\branch}{\mathit{node}}
\newcommand{\parent}{\mathit{parent}}
\newcommand{\height}{\mathit{height}}
\newcommand{\viewNumber}{\mathit{viewNumber}}
\newcommand{\qcVar}[1]{\mathit{#1}\mathit{QC}}
\newcommand{\qc}{\mathit{qc}}
\newcommand{\Hqc}{\qc_{\mathit{high}}}

\newcommand{\hqc}{\qcVar{high}}
\newcommand{\lqc}{\qcVar{locked}}
\newcommand{\prepareQC}{\qcVar{prepare}}
\newcommand{\keyQC}{\qcVar{precommit}}
\newcommand{\commitQC}{\qcVar{commit}}

\newcommand{\genericQC}{\qcVar{generic}}

\newcommand{\PREPARE}{\textsc{prepare}\xspace}
\newcommand{\COMMIT}{\textsc{commit}\xspace}
\newcommand{\KEY}{\textsc{pre-commit}\xspace}
\newcommand{\DECIDE}{\textsc{decide}\xspace}
\newcommand{\NEWVIEW}{\textsc{new-view}\xspace}
\newcommand{\GENERIC}{\textsc{generic}\xspace}

\newcommand{\QC}{QC\xspace}
\newcommand{\QCfull}{quorum certificate\xspace}
\newcommand{\QCFull}{Quorum certificate\xspace}
\newcommand{\QCFULL}{Quorum Certificate\xspace}

\algnewcommand{\leader}{\Call{leader}{\gViewNumber}}
\algnewcommand{\nextLeader}{\Call{leader}{\gViewNumber+1}}
\algnewcommand{\leaderOf}[1]{\Call{leader}{#1}}

\algblock{As}{EndAs}
\algnewcommand\algorithmicas{\textbf{as}}
\algrenewtext{As}[1]{\algorithmicas\ #1}
\algtext*{EndAs}

\newcommand{\HotStuff}{HotStuff\xspace}
\newcommand{\multiversions}[2]{#2}

\begin{document}
\title{\HotStuff: BFT Consensus in the Lens of Blockchain}
\author[1,2]{Maofan Yin}
\author[2]{Dahlia Malkhi}
\author[2,3]{Michael K. Reiter}
\author[2]{Guy Golan Gueta}
\author[2]{Ittai Abraham}
\affil[1]{Cornell University}
\affil[2]{VMware Research}
\affil[3]{UNC-Chapel Hill}
\date{}
\maketitle

\begin{abstract}
We present \HotStuff, a leader-based Byzantine fault-tolerant
replication protocol for the partially synchronous model.  Once
network communication becomes synchronous, \HotStuff enables a correct
leader to drive the protocol to consensus at the pace of actual (vs.\
maximum) network delay---a property
called \textit{responsiveness}---and with communication complexity
that is linear in the number of replicas.  To our
knowledge, \HotStuff\ is the first partially synchronous BFT
replication protocol exhibiting these combined properties.
\multiversions{%
Its simplicity enables it to be further pipelined and simplified into a
practical, concise protocol for building large-scale replication services.
}{%
\HotStuff is built around a novel framework
that forms a bridge between classical BFT foundations and
blockchains. It allows the expression of other known protocols
(DLS, PBFT, Tendermint, Casper), and ours, in a
common framework.

Our deployment of \HotStuff over a
network with over 100 replicas achieves throughput and latency
comparable to that of BFT-SMaRt, while enjoying linear communication
footprint during leader failover (vs.\ cubic with BFT-SMaRt).}

\end{abstract}

\section{Introduction}
\label{sec:intro}

Byzantine fault tolerance (BFT) refers to the ability of a computing
system to endure arbitrary (i.e., Byzantine) failures of its
components while taking actions critical to the system's operation.
In the context of state machine replication
(SMR)~\cite{Lamport78,SMR}, the system as a whole provides a
replicated service whose state is mirrored across $n$ deterministic
replicas.  A BFT SMR protocol is used to ensure that non-faulty
replicas agree on an order of execution for client-initiated service
commands, despite the efforts of $f$ Byzantine replicas.  This, in
turn, ensures that the $n - f$ non-faulty replicas will run commands
identically and so produce the same response for each command.  As is
common, we are concerned here with the partially synchronous
communication model~\cite{DLS88}, whereby a known bound $\Delta$
on message transmission holds after some unknown \textit{global
  stabilization time} (GST).  In this model, $n \ge 3f+1$ is required
for non-faulty replicas to agree on the same commands in the same
order (e.g.,~\cite{Ben-Or83}) and progress can be
ensured deterministically only after GST~\cite{FLP85}.

When BFT SMR protocols were originally conceived, a typical target
system size was $n = 4$ or $n = 7$, deployed on a local-area network.
However, the renewed interest in Byzantine fault-tolerance brought
about by its application to blockchains now demands solutions that can
scale to much larger $n$. 
In contrast to
\textit{permissionless} blockchains such as the one that supports
Bitcoin, for example, so-called \textit{permissioned} blockchains
involve a fixed set of replicas that collectively maintain an ordered
ledger of commands or, in other words, that support SMR.  Despite
their permissioned nature, numbers of replicas in the hundreds or even
thousands are envisioned (e.g.,~\cite{Thunder,SBFT}).
Additionally, their deployment to wide-area networks
requires setting $\Delta$ to accommodate higher variability in 
communication delays. 

\paragraph{The scaling challenge.}
Since the introduction of PBFT~\cite{CL99}, the first practical BFT replication solution in
the partial synchrony model,
numerous BFT solutions were built around its core two-phase paradigm. 
The practical aspect is that a stable leader can drive a consensus decision in just two rounds of message exchanges.
The first phase guarantees proposal uniqueness through the formation of a
\QCfull (\QC) consisting of $(n-f)$ votes. 
The second phase guarantees that the next leader can convince replicas to vote for a safe proposal. 

The algorithm for a new leader to collect information and propose it to replicas---called a \emph{view-change}---is the epicenter of
replication.
Unfortunately, view-change based on the two-phase paradigm is far from simple~\cite{SaddestMoment}, 
is bug-prone~\cite{RevisitBFT},
and incurs a significant communication penalty for even moderate system sizes.  
It requires the new leader to relay information from $(n-f)$ replicas, each reporting
its own highest known \QC. 
Even counting just authenticators (digital signatures or message authentication
codes), conveying a new proposal has a communication
footprint of $O(n^3)$ authenticators in PBFT, and variants that
combine multiple authenticators into one via threshold digital
signatures (e.g.,~\cite{CKS05,SBFT}) still send $O(n^2)$
authenticators.  The total number of authenticators transmitted if
$O(n)$ view-changes occur before a single consensus decision is
reached is $O(n^4)$ in PBFT, and even with threshold signatures is
$O(n^3)$. This scaling challenge plagues not only PBFT, but many other
protocols developed since then, e.g., 
Prime~\cite{Prime11},
Zyzzyva~\cite{KAD09},
Upright~\cite{Upright},
BFT-SMaRt~\cite{BFTSMaRt}, 
700BFT~\cite{SevenBFT},
and SBFT~\cite{SBFT}.

\HotStuff revolves around a three-phase core, allowing a new leader to simply
pick the highest \QC it knows of. It introduces a second phase that allows
replicas to ``change their mind'' after voting in the phase, without
requiring a leader proof at all.
This alleviates the above complexity, and at the same time considerably simplifies the leader replacement protocol.
Last, having (almost) canonized all the phases, it is very easy to pipeline \HotStuff, and
to frequently rotate leaders. 

To our knowledge, only BFT protocols in the blockchain arena like
Tendermint~\cite{TendermintThesis,TendermintGossip} and
Casper~\cite{Casper} follow such a simple leader regime.  However, these
systems are built around a synchronous core, wherein proposals are
made in pre-determined intervals that must accommodate the worst-case
time it takes to propagate messages over a wide-area peer-to-peer
gossip network.  In doing so, they forego a hallmark of most practical
BFT SMR solutions (including those listed above), namely
\textit{optimistic responsiveness}~\cite{Thunder}.  Informally,
\textit{responsiveness} requires that a non-faulty leader, once
designated, can drive the protocol to consensus in time depending only
on the \textit{actual} message delays, independent of any known upper
bound on message transmission delays~\cite{ADLS94}.  More
appropriate for our model is \textit{optimistic responsiveness}, which
requires responsiveness only in beneficial (and hopefully common)
circumstances---here, after GST is reached.  Optimistic or not,
responsiveness is precluded with designs such as Tendermint/Casper. The crux of
the difficulty is that there may exist an honest replica that has the highest
\QC, but the leader does not know about it. One can build scenarios where this
prevents progress ad infinitum (see Section~\ref{hs:correct} for a detailed
liveless scenario).
Indeed, failing to incorporate necessary delays at crucial protocol
steps can result in losing liveness outright, as has been reported in
several existing deployments, e.g., see~\cite{TMLivenessBug,IBFTLivenessBug,CV17}.

\paragraph{Our contributions.}
To our knowledge, we present the first BFT SMR protocol, called
\HotStuff, to achieve the following two properties:
\begin{itemize}
\item \textbf{Linear View Change}: After GST, any correct leader, once
  designated, sends only $O(n)$ authenticators to drive a consensus
  decision.  This includes the case where a leader is replaced.
  Consequently, communication costs to reach consensus after GST is
  $O(n^2)$ authenticators in the worst case of cascading leader
  failures.

\item \textbf{Optimistic Responsiveness}: After GST, any correct
  leader, once designated, needs to wait just for the first $n - f$
  responses to guarantee that it can create a proposal that will make
  progress.  This includes the case where a leader is replaced.
\end{itemize}

Another feature of \HotStuff is that the costs for a new leader to
drive the protocol to consensus is no greater than that for the
current leader.  As such, \HotStuff supports frequent succession of
leaders, which has been argued is useful in blockchain contexts for
ensuring chain quality~\cite{GKL15}.

\HotStuff achieves these properties by adding another phase to each view, 
a small price to latency in return for considerably simplifying the leader
replacement protocol.
This exchange incurs only the actual network
delays, which are typically far
smaller than $\Delta$ in practice.
As such, we expect this added
latency to be much smaller than that incurred by previous protocols
that forgo responsiveness to achieve linear view-change. 
Furthermore, throughput is not affected due to the efficient pipeline we introduce in Section~\ref{sec:hspro}.

\multiversions{%
\HotStuff has the additional benefit of being remarkably simple. Safety is
specified via voting and commit rules over graphs of nodes.
The mechanisms needed to achieve liveness are
encapsulated within a \emph{\Gadget}, cleanly separated from the
mechanisms needed for safety (Section~\ref{sec:implementation}).
}{%
In addition to the theoretical contribution, \HotStuff also provides insights
in understanding BFT replication in general and instantiating the protocol in
practice (see Section~\ref{sec:implementation}):
\begin{itemize}
\item A framework for BFT replication over graphs of nodes. Safety is
  specified via voting and commit graph rules. Liveness is specified
  separately via a \Gadget that extends the graph with new nodes.

\item A casting of several known protocols (DLS, PBFT, Tendermint, and Casper)
  and one new (ours, \HotStuff), in this framework.
\end{itemize}
\HotStuff has the additional benefit of being remarkably simple,
owing in part to its economy of mechanism: There are only two message
types and simple rules to determine how a
replica treats each.  Safety is
specified via voting and commit rules over graphs of nodes.
The mechanisms needed to achieve liveness are
encapsulated within a \emph{\Gadget}, cleanly separated from the
mechanisms needed for safety.  At the same time, it is expressive in
that it allows the representation of several known protocols (DLS, PBFT, Tendermint, and Casper) as minor variations.  In part this flexibility
derives from its operation over a graph of nodes, in a way that forms
a bridge between classical BFT foundations and modern blockchains.

We describe a prototype implementation and a preliminary evaluation of
\HotStuff.  Deployed over a network with over a hundred
replicas, \HotStuff achieves throughput and latency comparable to, and
sometimes exceeding, those of
mature systems such as BFT-SMaRt, whose code complexity far exceeds
that of \HotStuff.
We further demonstrate that the communication footprint of \HotStuff remains
constant in face of frequent leader replacements, whereas BFT-SMaRt
grows quadratically with the number of replicas.
}

\begin{table*}[h]
  \begin{center}
  \multiversions{\footnotesize}{\small}
  \begin{tabular}{lcccc}
    \toprule
     \multirow{2}{*}{Protocol} & \multicolumn{3}{c}{Authenticator complexity} & \multirow{2}{*}{Responsiveness} \\
     & \textit{Correct leader} & \textit{Leader failure (view-change)} & \textit{$f$ leader failures} & \\
    \midrule
    DLS~\cite{DLS88}
    & $O(n^4)$ & $O(n^4)$ & $O(n^4)$ &  \\
    PBFT~\cite{CL99}
    & $O(n^2)$ & $O(n^3)$ & $O(fn^3)$ & \checkmark \\
    SBFT~\cite{SBFT}
    & $O(n)$ & $O(n^2)$ & $O(fn^2)$ & \checkmark \\
    Tendermint~\cite{TendermintThesis} / Casper~\cite{Casper}
    & $O(n^2)$ & $O(n^2)$ & $O(fn^2)$ \\
    Tendermint\textsuperscript{*} / Casper\textsuperscript{*}
    & $O(n)$ & $O(n)$ & $O(fn)$ \\
    \textbf{\HotStuff}
      & $\bm{O(n)}$ & $\bm{O(n)}$ & $\bm{O(fn)}$ & \Checkmark \\
    \bottomrule
  \end{tabular}

  \textsuperscript{*}\footnotesize{Signatures can be combined using threshold signatures, though this optimization is not mentioned in their original works.}
  \end{center}
  \caption{Performance of selected protocols after GST.}
  \label{tbl:comparison}
  \vspace{-2em}
\end{table*}

\section{Related work}
\label{sec:related}

Reaching consensus in face of Byzantine failures was formulated as the
Byzantine Generals Problem by Lamport et al.~\cite{LSP82}, who also
coined the term ``Byzantine failures''. The first synchronous
solution was given by Pease et al.~\cite{PSL80}, and later improved by
Dolev and Strong~\cite{DS82}. The improved
protocol has $O(n^3)$ communication complexity, which was shown
optimal by Dolev and Reischuk~\cite{DR85}. A
leader-based synchronous protocol that uses randomness
was given by Katz and
Koo~\cite{KK09}, showing an expected constant-round
solution with $(n-1)/2$ resilience.

Meanwhile, in the asynchronous settings, Fischer et al.~\cite{FLP85}
showed that the problem is unsolvable deterministically in
asynchronous setting in face of a single failure.  Furthermore, an
$(n-1)/3$ resilience bound for any asynchronous solution was proven by
Ben-Or~\cite{Ben-Or83}.  Two approaches were
devised to circumvent the impossibility. One relies on randomness,
initially shown by Ben-Or~\cite{Ben-Or83}, using
independently random coin flips by processes until they happen to
converge to consensus. Later works used cryptographic methods to share
an unpredictable coin and drive complexities down to constant expected
rounds, and $O(n^3)$ communication~\cite{CKS05}.

The second approach relies on partial synchrony, first shown by Dwork,
Lynch, and Stockmeyer (DLS)~\cite{DLS88}.  This protocol
preserves safety during asynchronous periods, and after the system
becomes synchronous, DLS guarantees termination. Once synchrony is
maintained, DLS incurs $O(n^4)$ total communication and $O(n)$ rounds per
decision.

State machine replication relies on
consensus at its core to order client requests so that correct
replicas execute them in this order.  The recurring need for consensus
in SMR led Lamport to devise Paxos~\cite{Paxos}, a protocol that
operates an efficient pipeline in which a stable leader drives
decisions with linear communication and one round-trip. A similar
emphasis led Castro and Liskov~\cite{CL99,CL02} to develop an efficient
leader-based Byzantine SMR protocol named PBFT, whose stable leader
requires $O(n^2)$ communication and two round-trips per decision, and
the leader replacement protocol incurs $O(n^3)$ communication. PBFT
has been deployed in several systems, including
BFT-SMaRt~\cite{BFTSMaRt}. Kotla et al.\ introduced an optimistic
linear path into PBFT in a protocol named Zyzzyva~\cite{KAD09}, which
was utilized in several systems, e.g., Upright~\cite{Upright} and
Byzcoin~\cite{KJ16}.  The optimistic path has linear complexity, while
the leader replacement protocol remains $O(n^3)$.  Abraham et
al.~\cite{RevisitBFT} later exposed a safety violation in Zyzzyva,
and presented fixes~\cite{Zelma18,SBFT}. On the other hand,
to also reduce the complexity of the protocol itself, Song et al.\ proposed
Bosco~\cite{Bosco}, a simple one-step protocol with low latency on the optimistic path, requiring $5f+1$ replicas.
SBFT~\cite{SBFT} introduces an $O(n^2)$ communication view-change protocol that
supports a stable leader protocol with optimistically linear, one round-trip
decisions. It reduces the communication complexity 
by harnessing two methods: a collector-based communication paradigm by
Reiter~\cite{R95}, and signature combining via threshold cryptography on
protocol votes by Cachin et al.~\cite{CKS05}.

A leader-based Byzantine SMR protocol that employs randomization was
presented by Ramasamy et al.~\cite{RC05},
and a leaderless variant named HoneyBadgerBFT was developed by Miller
et al.~\cite{HoneyBadger}. At their core, these randomized Byzantine
solutions employ randomized asynchronous Byzantine consensus, whose
best known communication complexity was $O(n^3)$ (see above),
amortizing the cost via batching.
However, most recently, based on the idea in this \HotStuff paper, a parallel submission to
PODC'19~\cite{AMS19} further improves the communication complexity to $O(n^2)$.


Bitcoin's core is a protocol known as Nakamoto Consensus~\cite{Bitcoin}, a
synchronous protocol with only probabilistic safety guarantee and
no finality (see analysis
in~\cite{GKL15,PSS17,AM17}). It operates in a
\emph{permissionless} model where participants are unknown, and
resilience is kept via Proof-of-Work.  As described above, recent
blockchain solutions hybridize Proof-of-Work solutions with classical
BFT solutions in various
ways~\cite{BitcoinNG,KJ16,Solida,Casper,Algorand,Dfinity,Thunder}.  The
need to address rotating leaders in these hybrid solutions and others
provide the motivation behind \HotStuff.

\section{Model} \label{sec:model}

We consider a system consisting of a fixed set of $n=3f+1$
\textit{replicas}, indexed by $i \in [n]$ where $[n] = \{1, \ldots,
n\}$.  A set $F \subset [n]$ of up to $f = |F|$ replicas are Byzantine
faulty, and the remaining ones are correct.  We will often refer to
the Byzantine replicas as being coordinated by an \textit{adversary},
which learns all internal state held by these replicas (including
their cryptographic keys, see below).

Network communication is point-to-point, authenticated and reliable: one correct
replica receives a message from another correct replica if and only if the
latter sent that message to the former. When we refer to a ``broadcast'', it involves
the broadcaster, if correct, sending the same point-to-point messages to all replicas, including itself.
We adopt the \textit{partial synchrony model} of Dwork et
al.~\cite{DLS88}, where there is a known bound $\Delta$ and an
unknown Global Stabilization Time (GST), such that after GST, all
transmissions between two correct replicas arrive within time
$\Delta$.  Our protocol will ensure safety always, and will guarantee
progress within a bounded duration after GST.  (Guaranteeing progress
before GST is impossible~\cite{FLP85}.)  In practice, our protocol will
guarantee progress if the system remains stable (i.e., if messages
arrive within $\Delta$ time) for sufficiently long after GST, though
assuming that it does so forever simplifies discussion.

\paragraph{Cryptographic primitives.}

\HotStuff makes use of threshold signatures~\cite{S00,CKS05,BLS04}.
In a $(k, n)$-threshold signature scheme, there is a single public key
held by all replicas, and each of the $n$ replicas holds a distinct
private key.  The $i$-th replica can use its private key to contribute
a \textit{partial signature} $\partialResult{i} \gets \tsigSign{i}(m)$
on message $m$.  Partial signatures $\{\partialResult{i}\}_{i \in I}$,
where $|I| = k$ and each $\partialResult{i} \gets \tsigSign{i}(m)$,
can be used to produce a digital signature $\sig \gets \tsigCombine(m,
\{\partialResult{i}\}_{i \in I})$ on $m$.  Any other replica can
verify the signature using the public key and the function
\tsigVerify.  We require that if $\partialResult{i} \gets
\tsigSign{i}(m)$ for each $i \in I$, $|I| = k$, and if $\sig \gets
\tsigCombine(m, \{\partialResult{i}\}_{i \in I})$, then
$\tsigVerify(m, \sig)$ returns true.  However, given oracle access to
oracles $\{\tsigSign{i}(\cdot)\}_{i \in [n]\setminus F}$, an adversary
who queries $\tsigSign{i}(m)$ on strictly fewer than $k - f$ of these
oracles has negligible probability of producing a signature $\sig$ for
the message $m$ (i.e., such that $\tsigVerify(m, \sig)$ returns true).
Throughout this paper, we use a threshold of $k = 2f+1$.  Again, we
will typically leave invocations of \tsigVerify implicit in our
protocol descriptions.

We also require a cryptographic hash function \cryptoHashFn (also
called a \textit{message digest} function), which maps an
arbitrary-length input to a fixed-length output.  The hash function
must be \textit{collision resistant}~\cite{RS04}, which
informally requires that the probability of an adversary producing
inputs $m$ and $m'$ such that $\cryptoHashFn(m) = \cryptoHashFn(m')$
is negligible.  As such, $\cryptoHashFn(m)$ can serve as an identifier
for a unique input $m$ in the protocol.

\paragraph{Complexity measure.}
The complexity measure we care about is \textit{authenticator
  complexity}, which specifically is the sum, over all replicas $i \in
[n]$, of the number of authenticators received by replica $i$ in the
protocol to reach a consensus decision after GST.  (Again, before GST,
a consensus decision might not be reached at all in the worst
case~\cite{FLP85}.)  Here, an \textit{authenticator} is either a
partial signature or a signature.  Authenticator complexity is a
useful measure of communication complexity for several reasons.
First, like bit complexity and unlike message complexity, it hides
unnecessary details about the transmission topology.  For example, $n$
messages carrying one authenticator count the same as one message
carrying $n$ authenticators.  Second, authenticator complexity is
better suited than bit complexity for capturing costs in protocols
like ours that reach consensus repeatedly, where each consensus
decision (or each view proposed on the way to that consensus decision)
is identified by a monotonically increasing counter.  That is, because
such a counter increases indefinitely, the bit complexity of a
protocol that sends such a counter cannot be bounded.  Third, since in
practice, cryptographic operations to produce or verify digital
signatures and to produce or combine partial signatures are typically
the most computationally intensive operations in protocols that use
them, the authenticator complexity provides insight into the
computational burden of the protocol, as well.

\section{\HotStuffBasic}

HotStuff solves the State Machine Replication (SMR) problem.
At the core of SMR is a protocol for deciding on a growing log of
\emph{command} requests by clients. A group of state-machine replicas apply commands in
sequence order consistently.
A client sends a command request to all replicas, and waits for
responses from $(f+1)$ of them. For the most part, we omit the client from
the discussion, and defer to the standard literature for issues regarding
numbering and de-duplication of client requests.

The \HotStuffBasic solution is presented in Algorithm~\ref{alg:hotstuff-basic}. 
The protocol works in a succession of \emph{views}
numbered with monotonically increasing view numbers. Each
$\viewNumber$ has a unique dedicated leader known to all.
Each replica stores a \emph{tree} of pending commands as its local data structure. Each tree \emph{node} contains a proposed
command (or a batch of them), metadata associated with the protocol,
and a \emph{parent} link. The \emph{branch} led by a given node is the path from the node all the way to the tree root by visiting parent links.
During the protocol, a monotonically growing branch becomes \emph{committed}. To become committed, the
leader of a particular view proposing the branch must collect
votes from a quorum of $(n-f)$ replicas in three phases, \PREPARE, \KEY,
and \COMMIT.

A key ingredient in the protocol is
a collection of $(n-f)$ votes over a leader proposal, referred to as a \emph{\QCfull} (or ``\QC'' in short).
The \QC is associated with a particular node and a view number.
The $\tsigCombine$ utility employs a threshold signature scheme to generate a
representation of $(n-f)$ signed votes as a single authenticator.

Below we give an operational description of the protocol logic by phases,
followed by a precise specification in Algorithm~\ref{alg:hotstuff-basic},
and conclude the section with safety, liveness, and complexity arguments.
\subsection{Phases}
\paragraph{\PREPARE phase.}

The protocol for a new leader starts by collecting \NEWVIEW messages from $(n-f)$
replicas. 
The \NEWVIEW message is sent by a replica as it transitions into $\viewNumber$
(including the first view)
and carries the highest $\prepareQC$ that the replica received ($\bot$ if none),
as described below.

The leader processes these messages in order to select a branch that
has the highest preceding view in which a $\prepareQC$ was formed.
The leader selects the $\prepareQC$ with the highest view, denoted $\hqc$,
among the \NEWVIEW messages.
Because $\hqc$ is the highest among $(n-f)$ replicas,
no higher view could have reached a commit decision.
The branch led by $\fBranch{\hqc}$ is therefore safe.

Collecting \NEWVIEW messages to select a safe branch may be omitted by an
incumbent leader, who may simply select its own highest $\prepareQC$ as $\hqc$.
We defer this optimization to Section~\ref{sec:implementation} and only
describe a single, unified leader protocol in this section.
Note that, different from PBFT-like protocols, including
this step in the leader protocol is straightforward, and it incurs the same,
linear overhead as all the other phases of the protocol, regardless of the situation.

The leader uses the \Call{createLeaf}{} method to extend the tail of $\fBranch{\hqc}$ with
a new proposal. The method creates a new leaf node as a child and 
embeds a digest of the parent in the child node. 
The leader then sends the new node in a \PREPARE message
to all other replicas.
The proposal carries $\hqc$ for safety justification.

Upon receiving the \PREPARE message for the current view from the
leader, replica $r$ uses the \Call{safeNode}{} predicate to determine
whether to accept it.  If it is accepted, the replica sends a
\PREPARE vote with a partial signature (produced by \tsigSign{r})
for the proposal to the leader. 

\paragraph{\Call{safeNode}{} predicate.}

The \Call{safeNode}{} predicate is a core ingredient of the protocol. It examines
a proposal message $m$ carrying a \QC justification $\msgQC{m}$,
and determines whether $\msgBranch{m}$ is safe to accept.
The safety rule to accept a proposal is
the branch of $\msgBranch{m}$ extends from the currently locked node $\fBranch{\lqc}$.
On the other hand, the liveness rule is the replica will accept $m$ if $\msgQC{m}$ has a higher view than the current $\lqc$. The predicate is true as long as either one of two rules holds.

\paragraph{\KEY phase.}

When the leader receives $(n-f)$ \PREPARE votes for the
current proposal $\gBranch$,
it combines them into a $\prepareQC$.
The leader broadcasts $\prepareQC$ in \KEY messages.
A replica responds to the leader with \KEY vote
having a signed digest of the proposal.

\paragraph{\COMMIT phase.}

The \COMMIT phase is similar to \KEY phase.
When the leader receives $(n-f)$ \KEY votes,
it combines them into a $\keyQC$ and broadcasts it
in \COMMIT messages; replicas respond to it with a \COMMIT vote.
Importantly, a replica becomes \emph{locked} on the $\keyQC$ at
this point by setting its $\lqc$ to $\keyQC$
(Line \ref{line:hotstuff-basic:lock} of Algorithm~\ref{alg:hotstuff-basic}). This is crucial to guard the safety of the proposal in case it
becomes a consensus decision.

\paragraph{\DECIDE phase.}

When the leader receives $(n-f)$ \COMMIT votes,
it combines them into a $\commitQC$.
Once the leader has assembled a $\commitQC$, it sends it in a \DECIDE message to all other replicas.
Upon receiving a \DECIDE message,
a replica considers the proposal embodied in the $\commitQC$ a committed
decision, and executes the commands in the committed branch.
The replica increments $\viewNumber$ and starts the next view.

\paragraph{\Call{nextView}{} interrupt.}

In all phases, a replica waits for a message
at view $\viewNumber$ for a timeout
period, determined by an auxiliary $\textsc{nextView}(\viewNumber)$ utility.
If $\textsc{nextView}(\viewNumber)$ interrupts waiting,
the replica also increments $\viewNumber$ and starts the next view.


\subsection{Data Structures}
\paragraph{Messages.}
A message $m$ in the protocol has a fixed set of fields that are
populated using the $\Call{Msg}{{}}$ utility shown in
Algorithm~\ref{alg:util}.  $m$ is automatically stamped with
$\gViewNumber$, the sender's current view number.  Each message
has a type
$\msgType{m} \in \{ \NEWVIEW, \PREPARE, \allowbreak\KEY, \COMMIT, \DECIDE\}$.
$\msgBranch{m}$ contains a proposed node (the leaf node of a
proposed branch).  There is an optional field $\msgQC{m}$. The leader always
uses this field to carry the \QC for the different phases. Replicas
use it in \NEWVIEW messages to carry the highest $\prepareQC$. 
Each message sent in a replica role contains a partial signature
$\msgSig{m}$ by the sender over the tuple $\langle \msgType{m}$,
$\msgViewNumber{m}$, $\msgBranch{m} \rangle$, which is added in the
$\Call{voteMsg}{{}}$ utility.

\paragraph{\QCFull{}s.}

A \QCFULL (\QC) over a tuple $\langle \type, \viewNumber, \branch \rangle$ is a
data type that combines a collection of signatures for the same tuple
signed by $(n-f)$ replicas.
Given a \QC $\qc$, we use $\fType{\qc}$, $\fViewNumber{\qc}$, $\fBranch{\qc}$ to refer to the matching fields of the original tuple.

\paragraph{Tree and branches.}

Each command is wrapped in a node that additionally contains a parent link which could be a hash digest of the parent node. 
We omit the implementation details from the pseudocode.
During the protocol, a replica delivers a message only
after the branch led by the node is already in its local tree. In
practice, a recipient who falls behind can catch up by
fetching missing nodes from other replicas. 
For brevity, these details are also omitted from the pseudocode.
Two branches are \emph{conflicting} if neither one is an extension of the other.
Two nodes are conflicting if the branches led by them are conflicting.

\paragraph{Bookkeeping variables.}
A replica uses additional local variables for bookkeeping the protocol state:
(i) a $\viewNumber$, initially $1$ and incremented either by finishing a decision
or by a $\Call{nextView}{}$ interrupt;
(ii) a locked \QCfull $\lqc$, initially $\bot$, storing the highest \QC for which a replica voted \COMMIT; and
(iii) a $\prepareQC$, initially $\bot$, storing the highest \QC
for which a replica voted \KEY.
Additionally, in order to incrementally execute a committed log of commands, the replica
maintains the highest node whose branch has been executed. This is omitted below for brevity.

\subsection{Protocol Specification}
The protocol given in Algorithm~\ref{alg:hotstuff-basic} is described as an iterated
view-by-view loop.
In each view, 
a replica performs phases in succession based on its role, 
described as a succession of 
``\textbf{as}'' blocks. 
A replica can have more than one role. For example, a leader is also a (normal)
replica. Execution of \textbf{as} blocks across roles can be proceeded
concurrently. The execution of each \textbf{as} block is atomic.
A \Call{nextView}{} interrupt aborts all operations in any \textbf{as} block,
and jumps to the ``Finally'' block.

\begingroup
\captionof{algorithm}{Utilities (for replica $r$).}\label{alg:util}
\multiversions{\setstretch{0.93}}{}
\small\begin{algorithmic}[1]
\Function{Msg}{$\type$, $\branch$, $\qc$}
    \State $\msgType{m} \gets \type$
    \State $\msgViewNumber{m} \gets \gViewNumber$
    \State $\msgBranch{m} \gets \branch$
    \State $\msgQC{m}\gets \qc$
    \State \Return $m$
\EndFunction

\Function{voteMsg}{$\type$, $\branch$, $\qc$}
    \State $m \gets \Call{Msg}{\type, \branch, \qc}$
    \State $\msgSig{m} \gets \tsigSign{r}(\langle\msgType{m}, \msgViewNumber{m}, \msgBranch{m}\rangle)$
    \State \Return $m$
\EndFunction

\Procedure{createLeaf}{$\parent, \cmd$}
    \State $\fParent{b} \gets \parent$
    \State $\fCmd{b} \gets \cmd$
    \State \Return $b$
\EndProcedure

\Function{QC}{$V$}
    \State $\fType{\qc} \gets \msgType{m}: m \in V$
    \State $\fViewNumber{\qc} \gets \msgViewNumber{m} : m \in V$
    \State $\fBranch{\qc} \gets \msgBranch{m} : m \in V$
\multiversions{%
    \State $\fCert{\qc} \gets \tsigCombine(\langle\fType{\qc}, \fViewNumber{\qc}, \fBranch{\qc}\rangle,$
    \Statex\hspace{10em}$\{\msgSig{m} \mid m \in V\})$
}{%
    \State $\fCert{\qc} \gets \tsigCombine(\langle\fType{\qc}, \fViewNumber{\qc}, \fBranch{\qc}\rangle, \{\msgSig{m} \mid m \in V\})$
}
    \State \Return $\qc$
\EndFunction

\Function{matchingMsg}{$m, t, v$}
    \State\Return $(\msgType{m} = t) \land (\msgViewNumber{m} = v)$
\EndFunction

\Function{matchingQC}{$\qc, t, v$}
    \State\Return $(\fType{\qc} = t) \land (\fViewNumber{\qc} = v)$
\EndFunction

\Function{safeNode}{$\branch$, $\qc$}
    \State \Return
{\color{myblue}
    $(\branch\textrm{ extends from }\fBranch{\lqc})\,\,\lor$ \InlineComment{safety rule} \label{line:safeNode:conflict}}\\
    \hspace{1em}
{\color{myblue}
    $(\fViewNumber{\qc} > \fViewNumber{\lqc})$ \InlineComment{liveness rule} \label{line:safeNode:viewNumber}
}
        
\EndFunction
\end{algorithmic}
\endgroup
\begingroup
\captionof{algorithm}{\HotStuffBasic protocol (for replica $r$).}\label{alg:hotstuff-basic}
\multiversions{\setstretch{0.93}}{}
\small\begin{algorithmic}[1]
\For{$\gViewNumber \gets 1, 2, 3, \ldots$}

\Phase{\PREPARE phase}
\As{a leader} \InlineComment{$r = \leader$}
    \LineComment{we assume special \NEWVIEW messages from view $0$}
\multiversions{%
    \State wait for $(n-f)$ \NEWVIEW messages:
    \Statex\hspace{5em}$M \gets \{ m \mid \Call{matchingMsg}{m, \NEWVIEW, \gViewNumber-1 } \}$
}{%
    \State wait for $(n-f)$ \NEWVIEW messages: $M \gets \{ m \mid \Call{matchingMsg}{m, \NEWVIEW, \gViewNumber-1 } \}$
}
    \State $\hqc \gets \displaystyle\msgQC{\left(\argmax_{m \in M}\{\msgQCViewNumber{m}\}\right)}$
\multiversions{%
    \State $\gBranch \gets \textsc{createLeaf}(\fBranch{\hqc}, $
    \Statex\hspace{15.5em}$\textrm{client's command})$
}{%
    \State $\gBranch \gets \textsc{createLeaf}(\fBranch{\hqc}, \textrm{client's command})$
}
    \State broadcast \Call{Msg}{$\PREPARE, \gBranch, \hqc$}
\EndAs
\As{a replica}
\multiversions{%
    \State wait for message $m$ from $\leader$
    \Statex\hspace{5em}$m: \Call{matchingMsg}{m, \PREPARE, \gViewNumber}$
}{%
    \State wait for message $m: \Call{matchingMsg}{m, \PREPARE, \gViewNumber}$ from $\leader$
}
    \If{$\msgBranch{m}\textrm{ extends from }\fBranch{\msgQC{m}}\land{}$\Statex
        \hspace{8em}$\Call{safeNode}{\msgBranch{m}, \msgQC{m}}$}
        \State send $\Call{voteMsg}{\PREPARE, \msgBranch{m}, \bot}$ to $\leader$ \label{line:hotstuff-basic:vote-prepare}
    \EndIf
\EndAs

\Phase{\KEY phase}
\As{a leader} 
\multiversions{%
    \State wait for $(n-f)$ votes:
    \Statex\hspace{5em}$V \gets \{ v \mid \Call{matchingMsg}{v, \PREPARE, \gViewNumber} \}$
}{%
    \State wait for $(n-f)$ votes: $V \gets \{ v \mid \Call{matchingMsg}{v, \PREPARE, \gViewNumber} \}$
}
    \State $\prepareQC \gets \Call{QC}{V}$
   \State broadcast $\Call{Msg}{\KEY, \bot, \prepareQC}$
\EndAs
\As{a replica}
\multiversions{%
    \State wait for message $m$ from $\leader$
    \Statex\hspace{5em}$m: \Call{matchingQC}{\msgQC{m}, \PREPARE, \gViewNumber}$
    \State {\color{myblue} $\prepareQC \gets \msgQC{m}$}\label{line:hotstuff-basic:prepare}
    \State send to $\leader$
    \Statex\hspace{5em}$\Call{voteMsg}{\KEY, \msgQCBranch{m}, \bot}$
}{%
    \State wait for message $m: \Call{matchingQC}{\msgQC{m}, \PREPARE, \gViewNumber}$ from $\leader$
    \State {\color{myblue} $\prepareQC \gets \msgQC{m}$}\label{line:hotstuff-basic:prepare}
    \State send $\Call{voteMsg}{\KEY, \msgQCBranch{m}, \bot}$ to $\leader$
}
\EndAs

\Phase{\COMMIT phase }
\As{a leader} 
\multiversions{%
   \State wait for $(n-f)$ votes:
   \Statex\hspace{5em}$V \gets \{ v \mid \Call{matchingMsg}{v, \KEY, \gViewNumber} \}$
}{%
   \State wait for $(n-f)$ votes: $V \gets \{ v \mid \Call{matchingMsg}{v, \KEY, \gViewNumber} \}$
}
   \State $\keyQC \gets \Call{QC}{V}$
   \State broadcast $\Call{Msg}{\COMMIT, \bot, \keyQC}$
\EndAs
\As{a replica}
\multiversions{%
    \State wait for message $m$ from $\leader$
    \Statex\hspace{5em}$m: \Call{matchingQC}{\msgQC{m}, \KEY, \gViewNumber}$
    \State {\color{myblue} $\lqc \gets \msgQC{m}$}\label{line:hotstuff-basic:lock}
    \State send to $\leader$
    \Statex\hspace{5em}$\Call{voteMsg}{\COMMIT, \msgQCBranch{m}, \bot}$
}{%
    \State wait for message $m: \Call{matchingQC}{\msgQC{m}, \KEY, \gViewNumber}$ from $\leader$
    \State {\color{myblue} $\lqc \gets \msgQC{m}$}\label{line:hotstuff-basic:lock}
    \State send $\Call{voteMsg}{\COMMIT, \msgQCBranch{m}, \bot}$ to $\leader$
}
\EndAs

\Phase{\DECIDE phase }
\As{a leader} 
\multiversions{%
    \State wait for $(n-f)$ votes:
    \Statex\hspace{5em}$V \gets \{ v \mid \Call{matchingMsg}{v, \COMMIT, \gViewNumber} \}$
}{%
    \State wait for $(n-f)$ votes: $V \gets \{ v \mid \Call{matchingMsg}{v, \COMMIT, \gViewNumber} \}$
}
    \State $\commitQC \gets \Call{QC}{V}$
    \State broadcast $\Call{Msg}{\DECIDE, \bot, \commitQC}$
\EndAs
\As{a replica}
    \State wait for message $m$ from $\leader$
\multiversions{%
    \Statex\hspace{5em}$m: \Call{matchingQC}{\msgQC{m}, \COMMIT, \gViewNumber}$
    \State execute new commands through $\fBranch{\msgQC{m}}$,
    \Statex\hspace{3em}respond to clients
}{%
    \State wait for message $m: \Call{matchingQC}{\msgQC{m}, \COMMIT, \gViewNumber}$ from $\leader$
    \State execute new commands through $\fBranch{\msgQC{m}}$, respond to clients
}
\EndAs

\Phase{Finally}
\multiversions{%
\State \Call{nextView}{} interrupt: goto this line if $\Call{nextView}{\gViewNumber}$ is
\Statex\hspace{1.5em}called during ``wait for'' in any phase
}{%
\State \Call{nextView}{} interrupt: goto this line if
$\Call{nextView}{\gViewNumber}$ is called during ``wait for'' in any phase
}
\State {\color{myblue} send $\Call{Msg}{\NEWVIEW, \bot, \prepareQC}$ to
$\nextLeader$}
\label{line:hotstuff-basic:newview}

\EndFor
\end{algorithmic}
\endgroup

\subsection{Safety, Liveness, and Complexity}\label{hs:correct}
\paragraph{Safety.}
We first define a quorum certificate $\qc$ to be \textit{valid} if
$\tsigVerify(\langle\fType{\qc}$, $\fViewNumber{\qc}$,
$\fBranch{\qc}\rangle$, $\fCert{\qc})$ is true.

\begin{lemma}\label{lm:basic}
    For any valid $\qc_1, \qc_2$ in which $\fType{\qc_1} = \fType{\qc_2}$ and $\fBranch{\qc_1}$ conflicts with $\fBranch{\qc_2}$, we have $\fViewNumber{\qc_1} \neq \fViewNumber{\qc_2}$.
\end{lemma}
\begin{proof}
To show a contradiction, suppose $\fViewNumber{\qc_1}$ $=$
$\fViewNumber{\qc_2}$ $=$ $v$.  Because a valid QC can be formed only with
$n - f = 2f + 1$ votes (i.e., partial signatures) for it, there must
be a correct replica who voted twice in the same phase of $v$.  This
is impossible because the pseudocode allows voting only once for each
phase in each view.
\end{proof}

\begin{theorem}
  If $w$ and $b$ are conflicting nodes, then they cannot be both
  committed, each by a correct replica.
\end{theorem}


\begin{proof}
We prove this important theorem by contradiction.  Let $\qc_1$ denote a
valid $\commitQC$ (i.e., $\fType{\qc_1} = \COMMIT$) such that
$\fBranch{\qc_1} = w$, and $\qc_2$ denote a valid $\commitQC$ such that
$\fBranch{\qc_2} = b$.  Denote $v_1 = \fViewNumber{\qc_1}$ and $v_2 =
\fViewNumber{\qc_2}$.  By Lemma~\ref{lm:basic}, $v_1 \neq
v_2$. W.l.o.g.\ assume $v_1 < v_2$.

We will now denote by $v_s$ the lowest view higher than $v_1$ for
which there is a valid $\prepareQC{}$, $\qc_s$ (i.e., $\fType{\qc_s} =
\PREPARE$) where $\fViewNumber{\qc_s} = v_s$, and $\fBranch{\qc_s}$
conflicts with $w$.  Formally, we define the following predicate for
any $\prepareQC$:
\begin{align*}
\multiversions{%
    E(\prepareQC) \coloneqq &(v_1 < \fViewNumber{\prepareQC} \le v_2) \\
          &\land~(\fBranch{\prepareQC}\textrm{ conflicts with }w).
}{%
    E(\prepareQC) \coloneqq &(v_1 < \fViewNumber{\prepareQC} \le v_2) \land (\fBranch{\prepareQC}\textrm{ conflicts with }w).
}
\end{align*}
We can now set the \emph{first} switching point $\qc_s$:
\[
\multiversions{%
    \qc_s \coloneqq \argmin_{\prepareQC}
    \left\{
    \begin{array}{c}
        \fViewNumber{\prepareQC} \mid \\
        \prepareQC\textrm{ is valid} \land E(\prepareQC)
    \end{array}
    \right\}.
}{%
    \qc_s \coloneqq \argmin_{\prepareQC}
    \left\{
        \fViewNumber{\prepareQC} \mid \prepareQC\textrm{ is valid} \land E(\prepareQC)
    \right\}.
}
\]
Note that, by assumption such a $\qc_s$ must exist; for example, $\qc_s$ could be the
$\prepareQC{}$ formed in view $v_2$.

Of the correct replicas that sent a partial result
$\tsigSign{r}(\langle\fType{\qc_1}$, $\fViewNumber{\qc_1}$,
$\fBranch{\qc_1}\rangle)$, let $r$ be the first that contributed
$\tsigSign{r}(\langle\fType{\qc_s}$, $\fViewNumber{\qc_s}$,
$\fBranch{\qc_s}\rangle)$; such an $r$ must exist since otherwise, one of
$\fCert{\qc_1}$ and $\fCert{\qc_s}$ could not have been created.  During
view $v_1$, replica $r$ updates its lock $\lqc$ to a $\keyQC$ on $w$
at Line~\ref{line:hotstuff-basic:lock} of Algorithm~\ref{alg:hotstuff-basic}.
Due to the minimality of $v_s$, the lock that replica $r$ has on the branch led by $w$
is not changed before $\qc_s$ is formed. Otherwise $r$ must have seen some
other $\prepareQC$ with lower view because Line~\ref{line:hotstuff-basic:prepare} comes before Line~\ref{line:hotstuff-basic:lock}, contradicting to the minimality.  Now consider the invocation of
$\Call{safeNode}{}$ in the \PREPARE{} phase of view $v_s$ by replica
$r$, with a message $m$ carrying $\msgBranch{m} = \fBranch{\qc_s}$.  By
assumption, $\msgBranch{m}$ conflicts with $\fBranch{\lqc}$, and so
the disjunct at Line~\ref{line:safeNode:conflict} of
Algorithm~\ref{alg:util} is false.  Moreover, $\fViewNumber{\msgQC{m}}
> v_1$ would violate the minimality of $v_s$, and so the disjunct in
Line~\ref{line:safeNode:viewNumber} of Algorithm~\ref{alg:util} is
also false.
Thus, $\Call{safeNode}{}$ must return false and $r$ cannot
cast a \PREPARE{} vote on the conflicting branch in view $v_s$, a
contradiction.
\end{proof}

\paragraph{Liveness.}
There are two functions left undefined in the previous section:
$\Call{leader}{}$ and $\Call{nextView}{}$.  Their definition
will \emph{not} affect safety of the protocol, though they do matter
to liveness.  Before giving candidate definitions for them, we first
show that after GST, there is a bounded duration $T_f$ such that if
all correct replicas remain in view $v$ during $T_f$ and the leader
for view $v$ is correct, then a decision is reached.  Below, we say
that $\qc_1$ and $\qc_2$ \textit{match} if $\qc_1$ and $\qc_2$ are
valid, $\fBranch{\qc_1} = \fBranch{\qc_2}$, and $\fViewNumber{\qc_1}
= \fViewNumber{\qc_2}$.

\begin{lemma}\label{lm:liveness}
If a correct replica is locked such that $\lqc$ $=$ $\keyQC$, then at least $f
+ 1$ correct replicas voted for some $\prepareQC$ matching $\lqc$.

\end{lemma}
\begin{proof}
Suppose replica $r$ is locked on $\keyQC$.  Then, $(n-f)$ votes were
cast for the matching $\prepareQC$ in the $\PREPARE{}$ phase
(Line~\ref{line:hotstuff-basic:vote-prepare} of
Algorithm~\ref{alg:hotstuff-basic}), out of which at least $f + 1$ were
from correct replicas.
\end{proof}

\begin{theorem}
After GST, there exists a bounded time period $T_f$ such that if all
correct replicas remain in view $v$ during $T_f$ and the leader for
view $v$ is correct, then a decision is reached.
\end{theorem}
\begin{proof}
    Starting in a new view, the leader collects $(n-f)$ \NEWVIEW{}
    messages and calculates its $\hqc$ before broadcasting a \PREPARE{}
    messsage.  Suppose among all replicas (including the leader
    itself), the highest kept lock is $\lqc = \keyQC^{*}$.  By
    Lemma~\ref{lm:liveness}, we know there are at least $f + 1$
    correct replicas that voted for a $\prepareQC^{*}$ matching
    $\keyQC^{*}$, and have already sent them to the leader in
    their \NEWVIEW{} messages. Thus, the leader must learn a matching
    $\prepareQC^{*}$ in at least one of these \NEWVIEW{} messages and
    use it as $\hqc$ in its \PREPARE{} message. By the assumption, all
    correct replicas are synchronized in their view and the leader is
    non-faulty.  Therefore, all correct replicas will vote in
    the \PREPARE{} phase, since in \Call{safeNode}{}, the condition on
    Line~\ref{line:safeNode:viewNumber} of Algorithm~\ref{alg:util} is
    satisfied (even if the $\branch$ in the message conflicts with a
    replica's stale $\fBranch{\lqc}$, and so
    Line~\ref{line:safeNode:conflict} is not).
    Then, after the leader
    assembles a valid $\prepareQC$ for this view, all replicas will
    vote in all the following phases, leading to a new decision.
    After GST, the duration $T_f$ for these phases to complete is of
    bounded length.

    The protocol is Optimistically Responsive because there is no explicit
    ``wait-for-$\Delta$'' step, and the logical disjunction in \textsc{safeNode}
    is used to override a stale lock with the help of the three-phase paradigm.
\end{proof}

We now provide simple constructions for $\Call{leader}{}$ and
$\Call{nextView}{}$ that suffice to ensure that after GST, eventually
a view will be reached in which the leader is correct and all correct
replicas remain in this view for $T_f$ time.  It suffices for
$\Call{leader}{}$ to return some deterministic mapping from view
number to a replica, eventually rotating through all replicas.  A
possible solution for $\Call{nextView}{}$ is to utilize an exponential
back-off mechanism that maintains a timeout interval. Then a timer is
set upon entering each view. When the timer goes off without making
any decision, the replica doubles the interval and calls
$\Call{nextView}{}$ to advance the view.  Since the interval is
doubled at each time, the waiting intervals of all correct replicas
will eventually have at least $T_f$ overlap in common, during which
the leader could drive a decision.

\paragraph{Livelessness with two-phases.}
We now briefly demonstrate an infinite non-deciding scenario for a
``two-phase'' HotStuff. This explains the necessity for introducing a
synchronous delay in Casper and Tendermint, and hence for abandoning
(Optimistic) Responsiveness.

In the two-phase HotStuff variant, we omit the \KEY phase and proceed
directly to \COMMIT. A replica becomes locked when it votes on a $\prepareQC$.
Suppose, in view $v$, a leader proposes $b$. It completes the \PREPARE phase, and
some replica $r_v$ votes for the $\prepareQC$, say $\qc$, such that $\fBranch{\qc} = b$. Hence,
$r_v$ becomes locked on $\qc$. An asynchronous network scheduling causes the rest
of the replicas to move to view $v+1$ without receiving $\qc$.

We now repeat ad infinitum the following single-view transcript. We start view
$v+1$ with only $r_v$ holding the highest $\prepareQC$ (i.e. $\qc$) in
the system. The new leader $l$ collects new-view messages
from $2f+1$ replicas excluding $r_v$. The highest $\prepareQC$ among these, $\qc'$, has
view $v-1$ and $b' = \fBranch{\qc'}$ conflicts with $b$. $l$ then proposes $b''$
which extends $b'$, to which $2f$ honest replicas respond with a vote, but $r_v$ rejects it
because it is locked on $\qc$, $b''$ conflicts with $b$ and $\qc'$ is lower than $\qc$. Eventaully,
$2f$ replicas give up and move to the next view. Just then, a faulty replica
responds to $l$'s proposal, $l$ then puts together a $\prepareQC(v + 1, b'')$
and one replica, say $r_{v+1}$ votes for it and becomes locked on it.

\paragraph{Complexity.}
In each phase of \HotStuff, only the leader broadcasts to all replicas
while the replicas respond to the sender once with a partial
signature to certify the vote. In the leader's message, the \QC{}
consists of a proof of $(n-f)$ votes collected previously, which can
be encoded by a single threshold signature. In a replica's response,
the partial signature from that replica is the only authenticator.
Therefore, in each phase, there are $O(n)$ authenticators received in
total.  As there is a constant number of phases, the overall complexity
per view is $O(n)$.

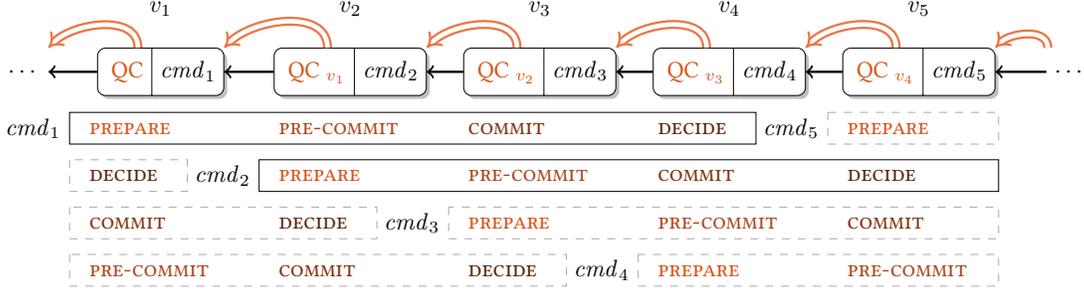
\begin{figure*}[htb]
\centering
\begin{tikzpicture}[x=1.12cm, scale=0.9, every node/.append style={transform shape}]
    \tikzstyle{data}=[rectangle split,rectangle split horizontal,rectangle split parts=2,rectangle split part align=base,draw,text centered, drop shadow={shadow xshift=0.3ex,shadow yshift=-0.3ex}, rounded corners=0.7ex]
    \definecolor{darkorange}{HTML}{d95319}
    \definecolor{orange}{HTML}{eb7645}
    \definecolor{mygray}{HTML}{aaaaaa}
    \begin{scope}[all/.style={draw, minimum height=0.7cm, minimum width=0.7cm},line width=0.08ex]
        \node[all, draw=none,fill=none] (b0) at (0.2, 0) {$\cdots$};
        \node[all, draw=none,fill=none] (b6) at (14, 0) {$\cdots$};
        \foreach \x in {1,...,5}
        {
            \node[all, data, draw,fill=white] (b\x) at ({2 + (\x - 1) * 2.5}, 0) {
                \pgfmathtruncatemacro{\xx}{\x - 1}
                {\color{darkorange}\ifthenelse{\x = 1}{\QC}{\QC\textsubscript{$v_\xx$}}}\nodepart{two}$\cmd_\x$
            };
            \node[anchor=south, yshift={1em}] (b\x{}l) at (b\x.north) {$v_\x$};
        }
        \begin{scope}[line width=0.2ex]
        \foreach \x in {1,...,6}
        {
            \pgfmathtruncatemacro{\xx}{\x - 1}
            \path[<-] (b\xx) edge[out=0, in=180] node[sloped,above] {} (b\x);
        }
        \end{scope}
        \foreach \x in {1,...,6}
        {
            \pgfmathtruncatemacro{\xx}{\x - 1}
            \begin{scope}[color=orange,line width=0.2ex]
                \path[-implies] (b\x.north) ++ (-0.3, 0) edge[double distance=0.4ex, out=90, in=north east] node[sloped,above] {} (b\xx);
            \end{scope}
        }
        \newcommand{\phases}{{"prepare", "pre-commit", "commit", "decide"}}
        \newcommand{\phasecolors}{{"d95319", "ac471d", "823210", "5b240c"}}
        \foreach \x in {1,...,5}
        {
            \foreach \y in {0,...,3}
            {
                \pgfmathsetmacro{\cphase}{\phases[mod(4 + (\x - 1) - \y, 4)]}
                \pgfmathsetmacro{\colorcode}{\phasecolors[mod(4 + (\x - 1) - \y, 4)]}
                \pgfmathsetmacro{\yy}{-(0.5 + \y * 0.7)}
                \definecolor{currentcolor}{HTML}{\colorcode}
                \node[anchor=north, align=left, text width={6em}] (p\x\y) at ({2 + (\x - 1) * 2.5}, \yy) {\strut\textsc{\color{currentcolor}{\cphase}}};
            }
        }
        \draw[]($(p10.north west)+(-0.15,-0.1)$) rectangle ($(p40.south east)+(-0.7,0.1)$);
        \draw[]($(p21.north west)+(-0.15,-0.1)$) rectangle ($(p51.south east)+(0,0.1)$);
        \begin{scope}[dashed, color=mygray]
            \draw[]($(p50.north west)+(-0.15,-0.1)$) rectangle ($(p50.south east)+(0,0.1)$);
            \draw[]($(p11.north west)+(-0.15,-0.1)$) rectangle ($(p11.south east)+(-0.7,0.1)$);
            \draw[]($(p12.north west)+(-0.15,-0.1)$) rectangle ($(p22.south east)+(-0.7,0.1)$);
            \draw[]($(p32.north west)+(-0.15,-0.1)$) rectangle ($(p52.south east)+(0,0.1)$);
            \draw[]($(p13.north west)+(-0.15,-0.1)$) rectangle ($(p33.south east)+(-0.7,0.1)$);
            \draw[]($(p43.north west)+(-0.15,-0.1)$) rectangle ($(p53.south east)+(0,0.1)$);
        \end{scope}
        \node[anchor=east, xshift={-0.5em}] (cmd1) at (p10.west) {$\cmd_1$};
        \node[anchor=east, xshift={-0.5em}] (cmd2) at (p21.west) {$\cmd_2$};
        \node[anchor=east, xshift={-0.5em}] (cmd3) at (p32.west) {$\cmd_3$};
        \node[anchor=east, xshift={-0.5em}] (cmd4) at (p43.west) {$\cmd_4$};
        \node[anchor=east, xshift={-0.5em}] (cmd5) at (p50.west) {$\cmd_5$};
    \end{scope}
\end{tikzpicture}
\caption{\HotStuffPro{} is a pipelined \HotStuffBasic{} where a QC
  can serve in different phases simultaneously.}
\label{fig:pipelining}
\end{figure*}
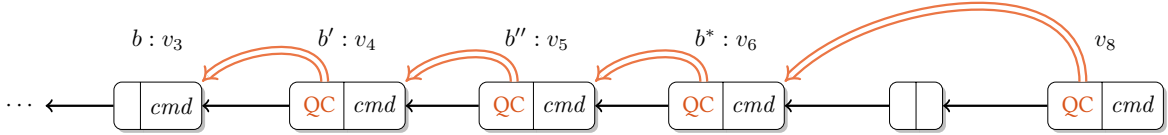
\begin{figure*}[htb]
\centering
\vspace{-2em}
\begin{tikzpicture}[x=1.12cm, scale=0.9, every node/.append style={transform shape}]
    \tikzstyle{data}=[rectangle split,rectangle split horizontal,rectangle split parts=2,rectangle split part align=base,draw,text centered, drop shadow={shadow xshift=0.3ex,shadow yshift=-0.3ex}, rounded corners=0.7ex]
    \definecolor{darkorange}{HTML}{d95319}
    \definecolor{orange}{HTML}{eb7645}
    \definecolor{mygray}{HTML}{aaaaaa}
    \begin{scope}[all/.style={draw, minimum height=0.7cm, minimum width=0.7cm},line width=0.08ex]
        \node[all, draw=none,fill=none] (b0) at (0.2, 0) {$\cdots$};
        \node[all, draw=none,fill=none] (b6) at (14, 0) {$\cdots$};
        \foreach \x in {1}
        {
            \node[all, data, draw,fill=white] (b\x) at ({2 + (\x - 1) * 2.5}, 0) {
                \pgfmathtruncatemacro{\xx}{\x - 1}
                {\hspace{3ex}}\nodepart{two}$\cmd$
            };
        }
        \foreach \x in {5}
        {
            \node[all, data, draw,fill=white] (b\x) at ({2 + (\x - 1) * 2.5}, 0) {
                \pgfmathtruncatemacro{\xx}{\x - 1}
            };
        }
        \newcommand{\nodes}{{"$b: v_3$", "$b': v_4$", "$b'': v_5$", "$b^{*}: v_6$", "", "$v_8$"}}
        \foreach \x in {2,3,4,6}
        {
            \node[all, data, draw,fill=white] (b\x) at ({2 + (\x - 1) * 2.5}, 0) {
                \pgfmathtruncatemacro{\xx}{\x - 1}
                {\color{darkorange}\QC}\nodepart{two}$\cmd$
            };
            \pgfmathtruncatemacro{\xt}{\x + 2}
        }
        \foreach \x in {1,2,3,4,6}
        {
            \pgfmathsetmacro{\bv}{\nodes[\x - 1]}
            \node[anchor=south, yshift={1em}] (b\x{}l) at (b\x.north) {\bv};
        }
        \begin{scope}[line width=0.2ex]
        \foreach \x in {1,...,6}
        {
            \pgfmathtruncatemacro{\xx}{\x - 1}
            \path[<-] (b\xx) edge[out=0, in=180] node[sloped,above] {} (b\x);
        }
        \end{scope}
        \foreach \x in {2,3,4}
        {
            \pgfmathtruncatemacro{\xx}{\x - 1}
            \begin{scope}[color=orange,line width=0.2ex]
                \path[-implies] (b\x.north) ++ (-0.3, 0) edge[double distance=0.4ex, out=90, in=north east] node[sloped,above] {} (b\xx);
            \end{scope}
        }
        \foreach \x in {6}
        {
            \pgfmathtruncatemacro{\xx}{\x - 2}
            \begin{scope}[color=orange,line width=0.2ex]
                \path[-implies] (b\x.north) ++ (-0.3, 0) edge[double distance=0.4ex, out=90, in=north east] node[sloped,above] {} (b\xx);
            \end{scope}
        }
    \end{scope}
\end{tikzpicture}
\caption{The nodes at views $v_4, v_5, v_6$ form a Three-Chain. The node at
view $v_8$ does not make a valid One-Chain in \HotStuffPro{} (but it is a valid
One-Chain after relaxation in the algorithm of
Section~\ref{sec:implementation}).}
\label{fig:three-chain}
\end{figure*}

\section{\HotStuffPro{}}
\label{sec:hspro}

It takes three phases for a \HotStuffBasic{} leader to commit a
proposal. These phases are not doing ``useful'' work except collecting votes
from replicas, and they are all very similar.
In \HotStuffPro{}, we improve the \HotStuffBasic{} protocol utility while at
the same time considerably simplifying it. The idea is to change the view on
\emph{every \PREPARE{} phase}, so each proposal has its own view.
This reduces the number of message types and allows for pipelining of decisions.
A similar approach for message type reduction was suggested in Casper~\cite{CasperOneMessage}.

More specifically, in \HotStuffPro{} the votes over a \PREPARE{} phase are collected in a view by the leader into a
$\genericQC$. Then the $\genericQC$ is  relayed to the leader of the next view,
essentially delegating responsibility for the next phase, which would have been
\KEY, to the next leader.
However, the next leader does not actually carry a \KEY{} phase, but instead initiates a
new \PREPARE{} phase and adds its own proposal.
This \PREPARE{} phase for view $v+1$
simultaneously serves as the \KEY{} phase for view $v$.
The \PREPARE{} phase for view $v+2$
simultaneously serves as the \KEY{} phase for view $v+1$ and
as the \COMMIT{} phase for view $v$.
This is possible because all the phases have identical structure.

The pipeline of \HotStuffBasic{} protocol phases embedded in a chain
of \HotStuffPro{} proposals is depicted in Figure~\ref{fig:pipelining}.
Views $v_1, v_2, v_3$ of \HotStuffPro{} serve as the \PREPARE, \KEY, and \COMMIT{} \HotStuffBasic{} phases for $cmd_1$ proposed in $v_1$.
This command becomes committed by the end of $v_4$. 
Views $v_2, v_3, v_4$ serve as the three \HotStuffBasic{} phases
for $cmd_2$ proposed in $v_2$, and it becomes committed by the end of $v_5$.
Additional proposals generated in these phases continue the pipeline similarly,
and are denoted by dashed boxes.  In Figure~\ref{fig:pipelining}, a single arrow
denotes the $\fParent{b}$ field for a node $b$, and a double arrow
denotes $\fBranch{\msgQC{b}}$. 

Hence, there are only two types of messages in \HotStuffPro{}, a \NEWVIEW{} message and
generic-phase \GENERIC{} message. 
The \GENERIC{} QC functions in all logically pipelined phases.
We next explain the mechanisms in the pipeline
to take care of locking and committing, which occur only in the \COMMIT{} and
\DECIDE{} phases of \HotStuffBasic.

\paragraph{Dummy nodes.}

The $\genericQC$ used by a leader in some view $\viewNumber$ may not
directly reference the proposal of the preceding view $(\viewNumber-1)$.
The reason is that the leader of a preceding view fails to obtain a \QC{},
either because there are conflicting proposals, or due to a benign crash.
To simplify the tree structure, 
\Call{createLeaf}{} extends $\fBranch{\genericQC}$ with blank nodes up to the height (the number of parent links on a node's branch) of
the proposing view, so view-numbers are equated with node heights.
As a result, the \QC\ embedded in a node $b$ may not refer to its
parent, i.e., $\fBranch{\fJustify{b}}$ may not equal $\fParent{b}$ (the last node in Figure~\ref{fig:three-chain}).

\paragraph{One-Chain, Two-Chain, and Three-Chain.}

When a node $b^{*}$ carries a \QC\ that refers to a direct parent,
i.e., $\fBranch{\fJustify{b^{*}}} = \fParent{b^{*}}$,
we say that it forms a \emph{One-Chain}.
Denote by $b'' = \fBranch{\fJustify{b^{*}}}$.
Node $b^{*}$ forms a \emph{Two-Chain},
if in addition to forming a One-Chain,
$\fBranch{\fJustify{b''}} = \fParent{b''}$.
It forms a \emph{Three-Chain}, if $b''$ forms a Two-Chain.

Looking at chain $b = \fBranch{\fJustify{b'}}$, $b' = \fBranch{\fJustify{b''}}$, $b'' =
\fBranch{\fJustify{b^{*}}}$, ancestry gaps might occur at any one of the nodes. 
These situations are similar to a leader of \HotStuffBasic{} failing to complete
any one of three phases, and getting interrupted to the next view by \Call{nextView}{}.

If $b^{*}$ forms a One-Chain, the \PREPARE{} phase of $b''$ has succeeded.
Hence, when a replica votes for $b^{*}$, it should remember
$\genericQC \gets \fJustify{b^{*}}$.
We remark that it is safe to update $\genericQC$ even when a One-Chain is not direct,
so long as it is higher than the current $\genericQC$.
In the implementation code described in Section~\ref{sec:implementation}, we indeed
update $\genericQC$ in this case.

If $b^{*}$ forms a Two-Chain, then the \KEY{} phase of $b'$ has succeeded.
The replica should therefore update $\lqc$ $\gets$ $\fJustify{b''}$. 
Again, we remark that the lock can be
updated even when a Two-Chain is not direct---safety will not break---and
indeed, this is given in the
implementation code in Section~\ref{sec:implementation}.

Finally, if $b^{*}$ forms a Three-Chain, the \COMMIT{} phase of $b$ has succeeded,
and $b$ becomes a committed decision.


\multiversions{}{\medskip}

Algorithm~\ref{alg:hotstuff-pipelined} shows the pseudocode for \HotStuffPro{}.
\multiversions{%
The proof of safety given by \cite{HotStuff} is similar to the one for \HotStuffBasic.
}{%
The proof of safety given by Theorem~\ref{thm:hotstuff-safety} in Appendix~\ref{app:hs-proof} is similar to the one for \HotStuffBasic.
}
We require the \QC{} in a valid node refers to
its ancestor. For brevity, we assume the constraint always holds and omit checking in the code.
\begingroup
\captionof{algorithm}{\HotStuffPro{} protocol.}\label{alg:hotstuff-pipelined}
\multiversions{\setstretch{0.93}}{}
\small\begin{algorithmic}[1]

\Procedure{createLeaf}{$\parent, \cmd, \qc$}
\multiversions{%
    \State $\fParent{b} \gets$ branch extending with blanks from $\parent$ to
    height $\gViewNumber$; $\fCmd{b} \gets \cmd$; $\fJustify{b} \gets \qc$; \Return $b$
}{%
    \State $\fParent{b} \gets$ branch extending with blanks from $\parent$ to
    height $\gViewNumber$
    \State $\fCmd{b} \gets \cmd$
    \State $\fJustify{b} \gets \qc$
    \State \Return $b$
}
\EndProcedure
\Statex
\For{$\gViewNumber \gets 1, 2, 3, \ldots$}

\Phase{\GENERIC{} phase}
\As{a leader} \InlineComment{$r = \leader{}$}
\multiversions{%
    \State wait for $(n-f)$ \NEWVIEW{} messages:
    \Statex\hspace{5em}$M \gets \{ m \mid \Call{matchingMsg}{m, \NEWVIEW, \gViewNumber-1} \}$\label{line:hs-generic-collect}
}{%
}
    \LineComment{$M$ is the set of messages collected at the end of previous view by the leader of this view}
    \State $\hqc \gets \displaystyle\msgQC{\left(\argmax_{m \in M}\{\msgQCViewNumber{m}\}\right)}$
    \If{$\fViewNumber{\hqc} > \fViewNumber{\genericQC}$} $\genericQC \gets \hqc$\EndIf
\multiversions{%
    \State $\gBranch \gets \textsc{createLeaf}(\fBranch{\genericQC},$
    \Statex\hspace{15em}$\textrm{client's command}, \genericQC)$
}{%
    \State $\gBranch \gets \textsc{createLeaf}(\fBranch{\genericQC}, \textrm{client's command}, \genericQC)$
}
    \LineComment{\PREPARE{} phase \multiversions{(leader-half)}{}}
    \State broadcast \Call{Msg}{$\GENERIC, \gBranch, \bot$}
\EndAs

\As{a replica}
\multiversions{%
    \State wait for message $m$ from $\leader{}$
    \Statex\hspace{5em}$m: \Call{matchingMsg}{m, \GENERIC{}, \gViewNumber}$
    \State $b^{*} \gets \msgBranch{m}$; $b'' \gets \fBranch{\fJustify{b^{*}}}$;
    \Statex\hspace{3em}$b' \gets \fBranch{\fJustify{b''}}$; $b \gets \fBranch{\fJustify{b'}}$
}{%
    \State wait for message $m: \Call{matchingMsg}{m, \GENERIC{}, \gViewNumber}$ from $\leader{}$
    \State $b^{*} \gets \msgBranch{m}$; $b'' \gets \fBranch{\fJustify{b^{*}}}$; $b' \gets \fBranch{\fJustify{b''}}$; $b \gets \fBranch{\fJustify{b'}}$
}
    \If {\Call{safeNode}{$b^{*}, \fJustify{b^{*}}$}}
        \State send $\Call{voteMsg}{\GENERIC, b^{*}, \bot}$ to
\multiversions{%
$\leader{}$
}{%
$\nextLeader{}$
}
    \EndIf
    \LineComment{start \KEY{} phase on $b^{*}$'s parent}
    \If {$\fParent{b^{*}} = b''$}
        \State {\color{myblue}$\genericQC \gets \fJustify{b^{*}}$}
    \EndIf
    \LineComment{start \COMMIT{} phase on $b^{*}$'s grandparent}
    \If {$(\fParent{b^{*}} = b'') \land (\fParent{b''} = b')$}
        \State {\color{myblue} $\lqc \gets \fJustify{b''}$}\label{line:hs-update-lock2}
    \EndIf
    \LineComment{start \DECIDE{} phase on $b^{*}$'s great-grandparent}
\multiversions{%
    \If {$(\fParent{b^{*}} = b'') \land (\fParent{b''} = b') \land{}$
    \Statex\hspace{5em}$(\fParent{b'} = b)$}\label{line:hs-check-commit}
}{%
    \If {$(\fParent{b^{*}} = b'') \land (\fParent{b''} = b') \land (\fParent{b'} = b)$ }\label{line:hs-check-commit}
}
        \State execute new commands through $b$, respond to clients
    \EndIf
\EndAs

\multiversions{%
\As{a leader} \InlineComment{\KEY{} phase (leader-half)}
    \State wait for $(n-f)$ votes:
    \Statex\hspace{5em}$V \gets \{ v \mid \Call{matchingMsg}{v, \GENERIC, \gViewNumber} \}$
}{%
\As{the next leader}
    \State wait for all messages: $M \gets \{ m \mid \Call{matchingMsg}{m, \GENERIC, \gViewNumber} \}$
    \Statex\hspace{5em}until there are $(n-f)$ votes: $V \gets \{ v \mid \msgSig{v} \neq \bot \land v \in M \}$
}
    \State $\genericQC \gets \Call{QC}{V}$
\EndAs

\multiversions{%
\As{the next leader}
    \State wait for message $m$ from $\leader{}$
    \Statex\hspace{5em}$m: \Call{matchingMsg}{m, \NEWVIEW{}, \gViewNumber}$\label{line:hs-new-view-wait}
\EndAs
}{%
}

\Phase{Finally}
\multiversions{%
\State \Call{nextView}{} interrupt: goto this line if $\Call{nextView}{\gViewNumber}$ is
\Statex\hspace{1.5em}called during ``wait for'' in any phase
}{%
\State \Call{nextView}{} interrupt: goto this line if
$\Call{nextView}{\gViewNumber}$ is called during ``wait for'' in any phase
}
\State\color{myblue}send $\Call{Msg}{\GENERIC, \bot, \genericQC}$ to $\nextLeader$
\EndFor
\end{algorithmic}

\endgroup

\section{Implementation}\label{sec:implementation}
\HotStuff is a practical protocol for building efficient SMR systems. 
Because of its simplicity, we can easily
turn Algorithm~\ref{alg:hotstuff-pipelined} into an event-driven-style
specification that is almost like the code skeleton for a prototype implementation. 

As shown in Algorithm~\ref{alg:hotstuff-incumbent-code}, the code is further simplified and
generalized by extracting the liveness mechanism from the body into a module
named \emph{\Gadget{}}. Instead of the next leader always waiting for a $\genericQC{}$ at the end of the
\GENERIC{} phase before starting its reign, this logic is delegated to the \Gadget. A
stable leader can skip this step and streamline proposals across multiple heights. 
Additionally, we relax the direct parent constraint
for maintaining the highest $\genericQC$ and $\lqc$, while still preserving the
requirement that the \QC{} in a valid node always refers to its ancestor.
The proof of correctness is similar to \HotStuffPro and we also defer it to the
appendix of \cite{HotStuff}.

\paragraph{Data structures.}
Each replica $u$ keeps track of the following main state variables:

\vspace{0.7em}
\begin{tabular}{rp{21em}}
    $\acsetof{\cdot}$ & mapping from a node to its votes. \\
    $\vheight$ & height of last voted node. \\
    $\Lqc$ & locked node (similar to $\lqc$). \\
    $\Lexec$ & last executed node. \\
    $\Hqc$ & highest known QC (similar to $\genericQC$) kept by a \Gadget. \\
    $\Ltail$ & leaf node kept by a \Gadget.
\end{tabular}
\multiversions{\vspace{0.2em}}{\vspace{0.7em}}\\
It also keeps a constant $\Lvl_0$, the same genesis node known by all correct
replicas. To bootstrap, $\Lvl_0$ contains a hard-coded QC for itself, $\Lqc,
\Lexec, \Ltail$ are all initialized to $\Lvl_0$, and $\Hqc$ contains the QC for $\Lvl_0$.

\paragraph{\Gadget{}.}

A \Gadget{} is a mechanism that guarantees progress after GST. It achieves this
through two ingredients.

The first one is ``synchronization'', bringing all correct replicas, and a
unique leader, into a common height for a sufficiently long period. The
usual synchronization mechanism in the
literature~\cite{DLS88,CL99,TendermintThesis} is for replicas to increase
the count of $\Delta$'s they spend at larger heights, until progress is being
made. A common way to deterministically elect a leader is to use a rotating
leader scheme in which all correct replicas keep a predefined leader
schedule and rotate to the next one when the leader is demoted.

Second, a \Gadget{} needs to provide the leader with a way to choose a proposal
that will be supported by correct replicas. As shown in
Algorithm~\ref{alg:hotstuff-incumbent-pmaker}, after a view change, in
\Call{onReceiveNewView}{}, the new leader collects \NEWVIEW{} messages sent by
replicas through \Call{onNextSyncView}{} to discover the highest QC to satisfy
the second part of the condition in \Call{onReceiveProposal}{} for liveness
(Line~\ref{line:hs-incumbent-liveness} of Algorithm~\ref{alg:hotstuff-incumbent-code}). During the same view, however, the incumbent leader will chain
the new node to the end of the leaf last proposed by itself, where no
\NEWVIEW{} message is needed.  Based on some
application-specific heuristics (to wait until the previously proposed node
gets a QC, for example), the current leader invokes \Call{onBeat}{} to propose
a new node carrying the command to be executed.

It is worth noting that even if a bad \Gadget{} invokes $\textsc{onPropose}$
arbitrarily, or selects a parent and a QC capriciously, and against any
scheduling delays, safety is always guaranteed. Therefore, safety guaranteed by
Algorithm~\ref{alg:hotstuff-incumbent-code} alone is entirely decoupled from
liveness by any potential instantiation of
Algorithm~\ref{alg:hotstuff-incumbent-pmaker}.

\begingroup
\captionof{algorithm}{Event-driven \HotStuff (for replica $u$).}\label{alg:hotstuff-incumbent-code}
\begin{center}
\small%
\multiversions{%
    \begin{algorithmic}[1]
    \Procedure{createLeaf}{$\parent, \cmd, \qc, \height$}
        \State $\fParent{b} \gets \parent$; $\fCmd{b} \gets \cmd$;
        \State $\fJustify{b} \gets \qc$; $\fHeight{b} \gets \height$; \Return $b$
    \EndProcedure
    \Procedure{update}{$b^{*}$}
       \State $b'' \gets {\fBranch{\fJustify{b^{*}}}}$; $b' \gets \fBranch{\fJustify{{b''}}}$
       \State $b \gets \fBranch{\fJustify{{b'}}}$
       \LineComment{\KEY{} phase on $b''$}
       \State\Call{updateQCHigh}{$\fJustify{b^{*}}$}
       \If{$\fHeight{b'} > \fHeight{\Lqc}$} 
        \State $\Lqc \gets b'$\InlineComment{\COMMIT{} phase on $b'$}
       \EndIf
       \If{$(\fParent{b''} = b') \land (\fParent{b'} = b)$}
            \State\Call{onCommit}{$b$}
            \State$\Lexec \gets b$\InlineComment{\DECIDE{} phase on $b$}
       \EndIf
    \EndProcedure

    \Procedure{onCommit}{$b$}
        \If{$\fHeight{\Lexec} < \fHeight{b}$}
            \State $\Call{onCommit}{\fParent{b}}$; $\Call{execute}{\fCmd{b}}$
        \EndIf
    \EndProcedure
 
     \Procedure{onReceiveProposal}{\Call{$\textsc{Msg}_v$}{$\GENERIC, \Lnew, \bot$}}
         \If{$\fHeight{\Lnew} > \vheight \land (\Lnew\textrm{ extends }\Lqc \lor$\label{line:hs-incumbent-conflict}\\
            \hspace{5ex}$\fHeight{\fBranch{\fJustify{\Lnew}}} > \fHeight{\Lqc})$}\label{line:hs-incumbent-liveness}
             \State$\vheight \gets \fHeight{\Lnew}$
             \State\Call{send}{\Call{getLeader}{{}},
\Call{$\textsc{voteMsg}_u$}{$\GENERIC, \Lnew, \bot$}}
         \EndIf
         \State\Call{update}{$\Lnew$}
     \EndProcedure
 
    \Procedure{onReceiveVote}{\Call{$m = \textsc{voteMsg}_v$}{$\GENERIC, b, \bot$}}
         \If{$\exists \langle v, \sigma'\rangle \in \acsetof{b}$}
             \Return \InlineComment{avoid duplicates}\EndIf
         \State$\acsetof{b} \gets \acsetof{b} \cup \{\langle v, \msgSig{m}\rangle\}$\InlineComment{collect votes}
         \If{$|\acsetof{b}| \ge n - f$}
            \State$\qc \gets \Call{QC}{\{\sigma \mid \langle v', \sigma \rangle \in \acsetof{b}\}}$
            \State$\Call{updateQCHigh}{\qc}$
         \EndIf
     \EndProcedure
 
    \Function{onPropose}{$\Ltail, \cmd, \Hqc$}
         \State$\begin{aligned}
             \Lnew \gets
             \Call{createLeaf}{\Ltail, \cmd, \Hqc, \fHeight{\Ltail} + 1}
         \end{aligned}$
         \LineComment{send to all replicas, including $u$ itself}
         \State$\textsc{broadcast}(\Call{$\textsc{Msg}_u$}{\GENERIC, \Lnew, \bot})$
         \State\Return$\Lnew$
     \EndFunction
\end{algorithmic}

}{%
    \begin{minipage}{\textwidth}
        \begin{multicols}{2}
            
        \end{multicols}
    \end{minipage}
}
\end{center}
\endgroup
\begingroup
\captionof{algorithm}{Code skeleton for a \Gadget{} (for replica $u$).}\label{alg:hotstuff-incumbent-pmaker}
\begin{center}
\small%
\multiversions{%
    \begin{algorithmic}[1]
    \LineComment{We assume \Gadget{} in all correct replicas will have
    synchronized leadership after GST.}
    \Function{getLeader}{}
        \InlineComment{\ldots specified by the application}
    \EndFunction
    \Procedure{updateQCHigh}{$\Hqc'$}
        \If{$\fHeight{\fBranch{\Hqc'}} > \fHeight{\fBranch{\Hqc}}$}
            \State$\Hqc \gets \Hqc'$
            \State$\Ltail \gets \fBranch{\Hqc}$
        \EndIf
    \EndProcedure
    \Procedure{onBeat}{$\cmd$}
        \If{$u = \Call{getLeader}{{}}$}
            \State$\Ltail \gets \Call{onPropose}{\Ltail, \cmd, \Hqc}$
        \EndIf
    \EndProcedure
    \Procedure{onNextSyncView}{}
        \State send $\Call{Msg}{\NEWVIEW, \bot, \Hqc}$ to $\Call{getLeader}{{}}$
    \EndProcedure
    \Procedure{onReceiveNewView}{\Call{$\textsc{Msg}$}{$\NEWVIEW, \bot, \Hqc'$}}
        \State\Call{updateQCHigh}{$\Hqc'$}
    \EndProcedure
\end{algorithmic}

}{%
    \begin{minipage}{\textwidth}
        \begin{multicols}{2}
            
        \end{multicols}
    \end{minipage}
}
\end{center}
\endgroup
\begingroup
\captionof{algorithm}{\Call{update}{} replacement for two-phase \HotStuff{}.}\label{alg:hotstuff-two-step-code}
\begin{center}
\small%
\begin{algorithmic}[1]
    \Procedure{update}{$b^{*}$}
       \State $b' \gets \fBranch{\fJustify{b^{*}}}$ ;
       $b \gets \fBranch{\fJustify{{b'}}}$
       \State\Call{updateQCHigh}{$\fJustify{b^{*}}$}
       \If{$\fHeight{b'} > \fHeight{\Lqc}$}
            $\Lqc \gets b'$
       \EndIf
       \If{$(\fParent{b'} = b)$}
            \Call{onCommit}{$b$}; $\Lexec \gets b$
       \EndIf
    \EndProcedure
\end{algorithmic}

\end{center}
\endgroup

\paragraph{Two-phase \HotStuff variant.}
To further demonstrate the flexibility of the \HotStuff framework, 
Algorithm~\ref{alg:hotstuff-two-step-code} shows the two-phase variant of
\HotStuff. Only the \Call{update}{} procedure is affected, a Two-Chain is
required for reaching a commit decision, and a One-Chain determines the lock. 
As discussed above (Section~\ref{hs:correct}),
this two-phase variant loses Optimistic Responsiveness, 
and is similar to Tendermint/Casper. 
The benefit is fewer phases, 
while liveness may be addressed by incorporating
in \Gadget{} a wait based on maximum network
delay. \multiversions{%

\paragraph{Evaluation.}
Due to the space limitation, we defer our evaluation results to the longer
paper~\cite{HotStuff}.  There, we compare our implementation to
BFT-SMaRt~\cite{BFTSMaRt}, a state-of-the-art implementation based on a two-phase
PBFT variant.  We show that even though three-phase \HotStuff{} has an additional
phase for its responsiveness and uses digital signatures universally (where
BFT-SMaRt only uses MACs for votes), it still
achieves similar latency, while being able to outperform BFT-SMaRt in
throughput. It also scales better than BFT-SMaRt.
}{%
See Section~\ref{sec:comparison-tendermint-casper} for further discussion.
}

\begingroup
\begin{center}
%
%
\begin{tabular}{ll}
    \begin{tikzpicture}[x=1.12cm, scale=0.9, every node/.append style={transform shape}]
    \tikzstyle{block}=[drop shadow={shadow xshift=0.3ex,shadow yshift=-0.3ex}, rounded corners=0.7ex]
    \tikzstyle{data}=[rectangle split,rectangle split horizontal,rectangle split parts=2,rectangle split part align=base,draw,text centered,block]
    \definecolor{darkorange}{HTML}{d95319}
    \definecolor{orange}{HTML}{eb7645}
    \begin{scope}[all/.style={draw, minimum height=0.7cm, minimum width=0.7cm},line width=0.08ex]
        \node[all, data, draw,fill=white] (b0) at (0, 0) {{\color{darkorange}\QC}\nodepart{two}};
        \node[all, data, draw,fill=white] (b1) at (2, 0) {{\color{darkorange}\QC}\nodepart{two}};
        \node[all, draw=none,fill=none] (c0) at (-2, 0) {$\cdots$};
        \node[anchor=south] (b0l) at (b0.north) {$b$};
        \node[anchor=south] (b1l) at (b1.north) {$b^{*}$};
        \node[anchor=north] (h0) at (b0.south) {height $k$};
        \node[anchor=north] (h1) at (b1.south) {height $k + 1$};
        \begin{scope}[dotted, line width=0.2ex, color=gray]
        \draw ($(b0.north west) + (-0.5, 0.3)$)  -- ++(0, -1.5);
        \draw ($(b1.north west) + (-0.5, 0.3)$)  -- ++(0, -1.5);
        \draw ($(b1.north east) + (0.5, 0.3)$)  -- ++(0, -1.5);
        \end{scope}
        \begin{scope}[line width=0.2ex]
        \path[<-] (b0) edge[out=0, in=180] node[sloped,above] {} (b1) ;
        \path[<-] (c0) edge[out=0, in=180] node[sloped,above] {} (b0) ;
        \end{scope}
        \begin{scope}[color=orange,line width=0.2ex]
            \path[-implies] (b1.north) ++ (-0.3, 0) edge[double distance=0.4ex, out=90, in=north east] node[sloped,above] {} (b0);
        \end{scope}
    \end{scope}
\end{tikzpicture}
 &
    \begin{tikzpicture}[x=1.12cm, scale=0.9, every node/.append style={transform shape}]
    \tikzstyle{block}=[drop shadow={shadow xshift=0.3ex,shadow yshift=-0.3ex}, rounded corners=0.7ex]
    \tikzstyle{data}=[rectangle split,rectangle split horizontal,rectangle split parts=2,rectangle split part align=base,draw,text centered,block]
    \definecolor{darkorange}{HTML}{d95319}
    \definecolor{orange}{HTML}{eb7645}
    \begin{scope}[all/.style={draw, minimum height=0.7cm, minimum width=0.7cm},line width=0.08ex]
        \node[all, data, draw,fill=white] (b0) at (0, 0) {{\color{darkorange}\QC}\nodepart{two}};
        \node[all, data, draw,fill=white] (b1) at (2, 0) {{\color{darkorange}\QC}\nodepart{two}};
        \node[all, data, draw,fill=white] (b2) at (4, 0) {{\color{darkorange}\QC}\nodepart{two}};
        \node[all, draw=none,fill=none] (c0) at (-2, 0) {$\cdots$};
        \node[anchor=south] (b0l) at (b0.north) {$b$};
        \node[anchor=south] (b1l) at (b1.north) {$b'$};
        \node[anchor=south] (b2l) at (b2.north) {$b^{*}$};
        \node[anchor=north] (h0) at (b0.south) {height $k$};
        \node[anchor=north] (h1) at (b1.south) {height $k + 1$};
        \node[anchor=north] (h2) at (b2.south) {height $k + 2$};
        \begin{scope}[dotted, line width=0.2ex, color=gray]
        \draw ($(b0.north west) + (-0.5, 0.3)$)  -- ++(0, -1.5);
        \draw ($(b1.north west) + (-0.5, 0.3)$)  -- ++(0, -1.5);
        \draw ($(b2.north west) + (-0.5, 0.3)$)  -- ++(0, -1.5);
        \draw ($(b2.north east) + (0.5, 0.3)$)  -- ++(0, -1.5);
        \end{scope}
        \begin{scope}[line width=0.2ex]
        \path[<-] (b1) edge[out=0, in=180] node[sloped,above] {} (b2) ;
        \path[<-] (b0) edge[out=0, in=180] node[sloped,above] {} (b1) ;
        \path[<-] (c0) edge[out=0, in=180] node[sloped,above] {} (b0) ;
        \end{scope}
        \begin{scope}[color=orange,line width=0.2ex]
            \path[-implies] (b1.north) ++ (-0.3, 0) edge[double distance=0.4ex, out=90, in=north east] node[sloped,above] {} (b0);
            \path[-implies] (b2.north) ++ (-0.3, 0) edge[double distance=0.4ex, out=90, in=north east] node[sloped,above] {} (b1);
        \end{scope}
    \end{scope}
\end{tikzpicture}

    \\
    (a) One-Chain (DLS, 1988) &
    (b) Two-Chain (PBFT, 1999)
    \\

    \begin{tikzpicture}[x=1.12cm, scale=0.9, every node/.append style={transform shape}]
    \tikzstyle{block}=[drop shadow={shadow xshift=0.3ex,shadow yshift=-0.3ex}, rounded corners=0.7ex]
    \tikzstyle{data}=[rectangle split,rectangle split horizontal,rectangle split parts=2,rectangle split part align=base,draw,text centered,block]
    \definecolor{darkorange}{HTML}{d95319}
    \definecolor{orange}{HTML}{eb7645}
    \begin{scope}[all/.style={draw, minimum height=0.7cm, minimum width=0.7cm},line width=0.08ex]
        \node[all, data, draw,fill=white] (b0) at (0, 0) {{\color{darkorange}\QC}\nodepart{two}};
        \node[all, data, draw,fill=white] (b1) at (2, 0) {{\color{darkorange}\QC}\nodepart{two}};
        \node[all, data, draw,fill=white] (b2) at (4, 0) {{\color{darkorange}\QC}\nodepart{two}};
        \node[all, draw=none,fill=none] (c0) at (-2, 0) {$\cdots$};
        \node[anchor=south] (b0l) at (b0.north) {$b$};
        \node[anchor=south] (b1l) at (b1.north) {$b'$};
        \node[anchor=south] (b2l) at (b2.north) {$b^{*}$};
        \node[anchor=north, yshift=-3ex] (h0) at (b0.south) {height $k$};
        \node[anchor=north, yshift=-3ex] (h1) at (b1.south) {height $k + 1$};
        \node[anchor=north, yshift=-3ex] (h2) at (b2.south) {height $k + 2$};
        \begin{scope}[dotted, line width=0.2ex, color=gray]
            \draw ($(c0.north east) + (0.4, 0.3)$)  -- ++(0, -1.5) node[inner sep=0] (c0e){};
            \draw ($(b0.north west) + (-0.2, 0.3)$)  -- ++(0, -1.5) node[inner sep=0] (v0s){};
            \draw ($(b0.north east) + (0.2, 0.3)$)  -- ++(0, -1.5) node[inner sep=0] (v0e) {};
            \draw ($(b1.north west) + (-0.2, 0.3)$)  -- ++(0, -1.5) node[inner sep=0] (v1s) {};
            \draw ($(b1.north east) + (0.2, 0.3)$)  -- ++(0, -1.5) node[inner sep=0] (v1e) {};
            \draw ($(b2.north west) + (-0.2, 0.3)$)  -- ++(0, -1.5) node[inner sep=0] (v2s) {};
            \draw ($(b2.north east) + (0.2, 0.3)$)  -- ++(0, -1.5) node[inner sep=0] (v2e) {};
        \end{scope}
        \begin{scope}[line width=0.15ex, color=gray]
            \draw[<->] ($(c0e) + (0, 0.5)$) -- ($(v0s) + (0, 0.5)$) node[below,midway] {$\Delta$};
            \draw[<->] ($(v0e) + (0, 0.5)$) -- ($(v1s) + (0, 0.5)$) node[below,midway] {$\Delta$};
            \draw[<->] ($(v1e) + (0, 0.5)$) -- ($(v2s) + (0, 0.5)$) node[below,midway] {$\Delta$};
        \end{scope}
        \begin{scope}[line width=0.2ex]
        \path[<-] (b1) edge[out=0, in=180] node[sloped,above] {} (b2) ;
        \path[<-] (b0) edge[out=0, in=180] node[sloped,above] {} (b1) ;
        \path[<-] (c0) edge[out=0, in=180] node[sloped,above] {} (b0) ;
        \end{scope}
        \begin{scope}[color=orange,line width=0.2ex]
            \path[-implies] (b1.north) ++ (-0.3, 0) edge[double distance=0.4ex, out=90, in=north east] node[sloped,above] {} (b0);
            \path[-implies] (b2.north) ++ (-0.3, 0) edge[double distance=0.4ex, out=90, in=north east] node[sloped,above] {} (b1);
        \end{scope}
    \end{scope}
\end{tikzpicture}
 &
    \begin{tikzpicture}[x=1.12cm, scale=0.9, every node/.append style={transform shape}]
    \tikzstyle{block}=[drop shadow={shadow xshift=0.3ex,shadow yshift=-0.3ex}, rounded corners=0.7ex]
    \tikzstyle{data}=[rectangle split,rectangle split horizontal,rectangle split parts=2,rectangle split part align=base,draw,text centered,block]
    \definecolor{darkorange}{HTML}{d95319}
    \definecolor{orange}{HTML}{eb7645}
    \begin{scope}[all/.style={draw, minimum height=0.7cm, minimum width=0.7cm},line width=0.08ex]
        \node[all, data, draw,fill=white] (b0) at (0, 0) {{\color{darkorange}\QC}\nodepart{two}};
        \node[all, data, draw,fill=white] (b1) at (2, 0) {{\color{darkorange}\QC}\nodepart{two}};
        \node[all, data, draw,fill=white] (b2) at (4.9, 0) {{\color{darkorange}\QC}\nodepart{two}};
        \node[all, draw=none,fill=none] (c0) at (-2, 0) {$\cdots$};
        \node[all, draw=none,fill=none] (c1) at (3.7, 0) {$\cdots$};
        \node[anchor=south] (b0l) at (b0.north) {$b$};
        \node[anchor=south] (b1l) at (b1.north) {$b'$};
        \node[anchor=south] (b2l) at (b2.north) {$b^{*}$};
        \node[anchor=north, yshift=-3ex] (h0) at (b0.south) {height $k$};
        \node[anchor=north, yshift=-3ex] (h1) at (b1.south) {height $k + 1$};
        \node[anchor=north, yshift=-3ex] (h2) at (b2.south) {height $k'$};
        \begin{scope}[dotted, line width=0.2ex, color=gray]
            \draw ($(c0.north east) + (0.4, 0.3)$)  -- ++(0, -1.5) node[inner sep=0] (c0e){};
            \draw ($(b0.north west) + (-0.2, 0.3)$)  -- ++(0, -1.5) node[inner sep=0] (v0s){};
            \draw ($(b0.north east) + (0.2, 0.3)$)  -- ++(0, -1.5) node[inner sep=0] (v0e) {};
            \draw ($(b1.north west) + (-0.2, 0.3)$)  -- ++(0, -1.5) node[inner sep=0] (v1s) {};
            \draw ($(b1.north east) + (0.2, 0.3)$)  -- ++(0, -1.5) node[inner sep=0] (v1e) {};
            \draw ($(c1.north west) + (-0.1, 0.3)$)  -- ++(0, -1.5) node[inner sep=0] (c1s){};
            \draw ($(b2.north west) + (-0.2, 0.3)$)  -- ++(0, -1.5) node[inner sep=0] (v2s) {};
            \draw ($(b2.north east) + (0.2, 0.3)$)  -- ++(0, -1.5) node[inner sep=0] (v2e) {};
        \end{scope}
        \begin{scope}[line width=0.15ex, color=gray]
            \draw[<->] ($(c0e) + (0, 0.5)$) -- ($(v0s) + (0, 0.5)$) node[below,midway] {$\Delta$};
            \draw[<->] ($(v0e) + (0, 0.5)$) -- ($(v1s) + (0, 0.5)$) node[below,midway] {$\Delta$};
            \draw[<->] ($(v1e) + (0, 0.5)$) -- ($(c1s) + (0, 0.5)$) node[below,midway] {$\Delta$};
        \end{scope}
        \begin{scope}[line width=0.2ex]
        \path[<-] (c1) edge[out=0, in=180] node[sloped,above] {} (b2) ;
        \path[<-] (b1) edge[out=0, in=180] node[sloped,above] {} (c1) ;
        \path[<-] (b0) edge[out=0, in=180] node[sloped,above] {} (b1) ;
        \path[<-] (c0) edge[out=0, in=180] node[sloped,above] {} (b0) ;
        \end{scope}
        \begin{scope}[color=orange,line width=0.2ex]
            \path[-implies] (b1.north) ++ (-0.3, 0) edge[double distance=0.4ex, out=90, in=north east] node[sloped,above] {} (b0);
            \path[-implies] (b2.north) ++ (-0.3, 0) edge[double distance=0.4ex, out=90, in=north east] node[sloped,above] {} (b1);
        \end{scope}
    \end{scope}
\end{tikzpicture}

    \\
    (c) Two-Chain w/ delay (Tendermint, 2016) &
    (d) Two-Chain w/ delay (Casper, 2017)
    \\

    \multicolumn{2}{c}{
        \begin{tabular}{c}
        \begin{tikzpicture}[x=1.12cm, scale=0.9, every node/.append style={transform shape}]
    \tikzstyle{block}=[drop shadow={shadow xshift=0.3ex,shadow yshift=-0.3ex}, rounded corners=0.7ex]
    \tikzstyle{data}=[rectangle split,rectangle split horizontal,rectangle split parts=2,rectangle split part align=base,draw,text centered,block]
    \definecolor{darkorange}{HTML}{d95319}
    \definecolor{orange}{HTML}{eb7645}
    \begin{scope}[all/.style={draw, minimum height=0.7cm, minimum width=0.7cm},line width=0.08ex]
        \node[all, data, draw,fill=white] (b0) at (0, 0) {{\color{darkorange}\QC}\nodepart{two}};
        \node[all, data, draw,fill=white] (b1) at (2, 0) {{\color{darkorange}\QC}\nodepart{two}};
        \node[all, data, draw,fill=white] (b2) at (4, 0) {{\color{darkorange}\QC}\nodepart{two}};
        \node[all, data, draw,fill=white] (b3) at (6.9, 0) {{\color{darkorange}\QC}\nodepart{two}};
        \node[all, draw=none,fill=none] (c0) at (-2, 0) {$\cdots$};
        \node[all, draw=none,fill=none] (c1) at (5.7, 0) {$\cdots$};
        \node[anchor=south] (b0l) at (b0.north) {$b$};
        \node[anchor=south] (b1l) at (b1.north) {$b'$};
        \node[anchor=south] (b2l) at (b2.north) {$b''$};
        \node[anchor=south] (b3l) at (b3.north) {$b^{*}$};
        \node[anchor=north, yshift=-3ex] (h0) at (b0.south) {height $k$};
        \node[anchor=north, yshift=-3ex] (h1) at (b1.south) {height $k + 1$};
        \node[anchor=north, yshift=-3ex] (h2) at (b2.south) {height $k + 2$};
        \node[anchor=north, yshift=-3ex] (h3) at (b3.south) {height $k'$};
        \begin{scope}[dotted, line width=0.2ex, color=gray]
            \draw ($(b0.north west) + (-0.5, 0.3)$)  -- ++(0, -1.5) node[inner sep=0] (v0s){};
            \draw ($(b1.north west) + (-0.5, 0.3)$)  -- ++(0, -1.5) node[inner sep=0] (v1s) {};
            \draw ($(b2.north west) + (-0.5, 0.3)$)  -- ++(0, -1.5) node[inner sep=0] (v2s) {};
            \draw ($(b2.north east) + (0.5, 0.3)$)  -- ++(0, -1.5) node[inner sep=0] (v3e) {};
        \end{scope}
        \begin{scope}[line width=0.2ex]
        \path[<-] (c1) edge[out=0, in=180] node[sloped,above] {} (b3) ;
        \path[<-] (b2) edge[out=0, in=180] node[sloped,above] {} (c1) ;
        \path[<-] (b1) edge[out=0, in=180] node[sloped,above] {} (b2) ;
        \path[<-] (b0) edge[out=0, in=180] node[sloped,above] {} (b1) ;
        \path[<-] (c0) edge[out=0, in=180] node[sloped,above] {} (b0) ;
        \end{scope}
        \begin{scope}[color=orange,line width=0.2ex]
            \path[-implies] (b1.north) ++ (-0.3, 0) edge[double distance=0.4ex, out=90, in=north east] node[sloped,above] {} (b0);
            \path[-implies] (b2.north) ++ (-0.3, 0) edge[double distance=0.4ex, out=90, in=north east] node[sloped,above] {} (b1);
            \path[-implies] (b3.north) ++ (-0.3, 0) edge[double distance=0.4ex, out=90, in=north east] node[sloped,above] {} (b2);
        \end{scope}
    \end{scope}
\end{tikzpicture}
 \\
        (e) Three-Chain (\HotStuff, 2018) \\
        \end{tabular}}
\end{tabular}


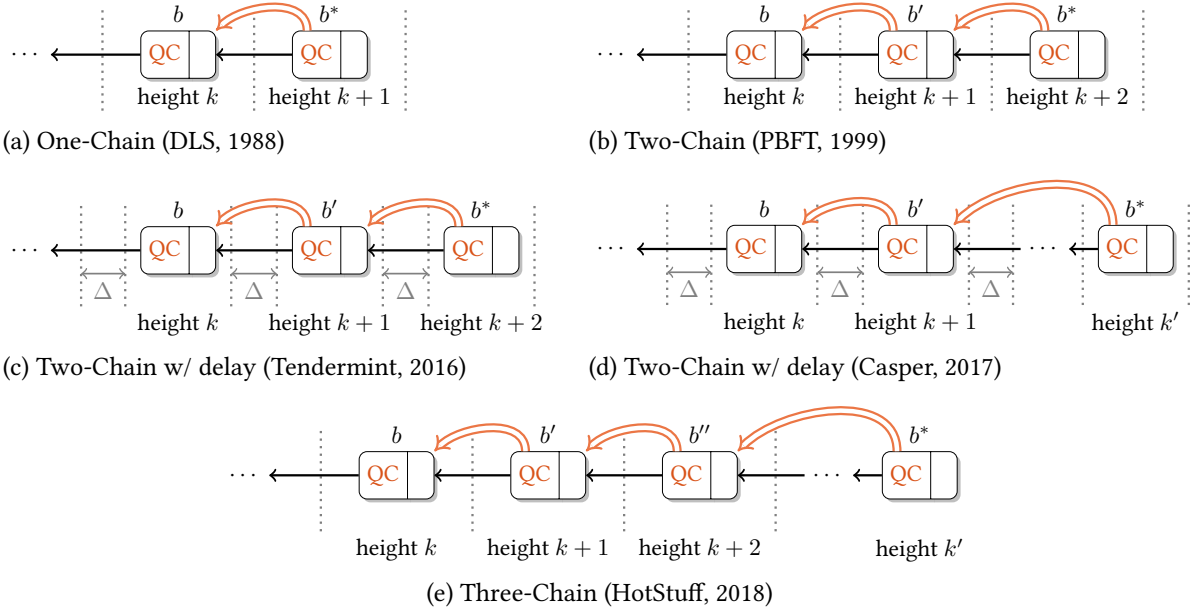
\captionof{figure}{Commit rules for different BFT protocols.}\label{fig:comparison}
\end{center}
\endgroup

\section{One-Chain and Two-Chain BFT Protocols}

In this section, we examine four BFT replication protocols
spanning four decades of research in Byzantine fault tolernace, casting them
into a chained framework similar to \HotStuffPro.

Figure~\ref{fig:comparison} provides a birds-eye view of the commit rules of
five protocols we consider, including \HotStuff.

In a nutshell,
the commit rule in DLS~\cite{DLS88} is One-Chain, allowing a node to be committed
only by its own leader.
The commit rules in PBFT~\cite{CL99}, Tendermint~\cite{TendermintThesis,TendermintGossip} and
Casper~\cite{Casper} are almost identical, and consist of Two-Chains. They differ
in the mechanisms they introduce for liveness, PBFT has leader ``proofs'' of
quadratic size (no Linearity), Tendermint and Casper introduce a mandatory $\Delta$ delay
before each leader proposal (no Optimistic Responsiveness).
\HotStuff{} uses a Three-Chain rule, and has a linear leader protocol without
delay.

\subsection{DLS}

The simplest commit rule is a One-Chain.
Modeled after Dwork, Lynch, and Stockmeyer (DLS),
the first known asynchronous Byzantine Consensus solution, 
this rule is depicted in Figure~\ref{fig:comparison}(a).
A replica becomes locked in DLS on the highest node it voted for.

Unfortunately, this rule would easily lead to a deadlock if at some height,
a leader equivocates, and two correct replicas became locked on the conflicting
proposals at that height.
Relinquishing either lock is unsafe unless there are $2f+1$ that indicate they
did not vote for the locked value.

Indeed, in DLS only the leader of each height can itself reach a commit
decision by the One-Chain commit rule. Thus, only the leader itself is harmed if it
has equivocated. Replicas can relinquish a lock either if $2f+1$ replicas did
not vote for it, or if there are conflicting proposals (signed by the leader).
The unlocking protocol occurring at the end of each height in DLS turns out to be fairly complex
and expensive. Together with the fact that only the leader for a height can
decide, in the best scenario where no fault occurs and the network
is timely, DLS requires $n$ leader rotations, and $O(n^4)$ message
transmissions, per single decision. While it broke new ground in demonstrating a safe asynchronous
protocol, DLS was not designed as a practical solution.

\subsection{PBFT}

Modeled after PBFT, a more practical appraoch uses a Two-Chain
commit rule, see Figure~\ref{fig:comparison}(b).
When a replica votes for a node that forms a One-Chain, it becomes locked on
it.
Conflicting One-Chains at the same height are simply not possible, as each has a
\QC, hence the deadlock situation of DLS is avoided.

However, if one replica holds a higher lock than others, a leader may not
know about it even if it collects information from $n-f$ replicas.
This could prevent leaders from reaching decisions ad infinitum, purely due to
scheduling.
To get ``unstuck'', the PBFT unlocks all replicas by
carrying a \emph{proof} consisting of the highest One-Chain's by
$2f+1$ replicas.
This proof is quite involved, as explained below.

The original PBFT, which has been
open-sourced~\cite{CL99} and adopted in several follow up
works~\cite{BFTSMaRt,KAD09}, a leader proof contains a set of messages collected from
$n-f$ replicas reporting the highest One-Chain each member voted for. Each
One-Chain contains a \QC, hence the total communication cost is $O(n^3)$.
Harnessing signature combining methods from~\cite{R95,CKS05},
SBFT~\cite{SBFT} reduces this cost to $O(n^2)$ by turning each \QC{} to a single value.

In the PBFT variant in \cite{CL02}, a leader proof contains
the highest One-Chain the leader collected from the quorum only once. It also includes one signed
value from each member of the quorum, proving that it did not vote for a higher
One-Chain.
Broadcasting this proof incurs communication complexity $O(n^2)$.
Note that whereas the signatures on a \QC{} may be combined into a single value, the proof
as a whole cannot be reduced to constant size
because messages from different members of the quorum may have different values.

In both variants, a correct replica unlocks
even it has a higher One-Chain than the leader's proof.
Thus, a correct leader can force its proposal to
be accepted during period of synchrony, and liveness is guaranteed.
The cost is quadratic communication per leader replacement.

\subsection{Tendermint and Casper}\label{sec:comparison-tendermint-casper}

Tendermint has a Two-Chain commit rule identical to PBFT,
and Casper has a Two-Chain rule in which the leaf does not need to
have a \QC\ to direct parent.  That is, in Casper,
Figure~\ref{fig:comparison}(c,d) depicts the commit rules for Tendermint and
Casper, respectively.

In both methods, a leader simply sends the highest One-Chain it knows along with its
proposal.
A replica unlocks a One-Chain if it receives from the leader a higher one.

However, because correct replicas may not vote for a leader's node,
to guarantee progress a new leader must obtain the highest
One-Chain by waiting the maximal network delay.
Otherwise, if leaders only wait for the first $n-f$ messages to start a new
height, there is no progress guarantee.
Leader delays are inherent both in Tendermint and in Casper, in order to provide
liveness.

This simple leader protocol embodies a linear leap in the communication complexity of
the leader protocol, which \HotStuff{} borrows from.
As already mentioned above, a \QC{} could be captured in a single value using
threshold-signatures, hence a leader can collect and disseminate the highest
One-Chain with linear communication complexity.
However, crucially, due to the extra \QC{} step, \HotStuff{} does not require
the leader to wait the maximal network delay.

\section{Evaluation}\label{app:eval}

We have implemented \HotStuff as a library in roughly 4K
lines of C++ code. Most noticeably, the core consensus logic specified
in the pseudocode consumes only around 200 lines.
In this section, we will first examine baseline
throughput and latency by comparing to a state-of-art system,
BFT-SMaRt~\cite{BFTSMaRt}.  We then focus on the message cost for view
changes to see our advantages in this scenario.
 
\subsection{Setup}
We conducted our experiments on Amazon EC2 using \texttt{c5.4xlarge}
instances.  Each instance had 16 vCPUs supported by Intel Xeon
Platinum 8000 processors.  All cores sustained a Turbo CPU clock speed
up to 3.4GHz. We ran each replica on a single VM instance, and so
BFT-SMaRt, which makes heavy use of threads, was allowed to utilize 16
cores per replica, as in their original evaluation~\cite{BFTSMaRt}.
The maximum TCP bandwidth measured by \texttt{iperf} was around 1.2 Gigabytes
per second. We did not throttle the bandwidth in any run. The network latency
between two machines was less than 1~ms.

Our prototype implementation of \HotStuff uses secp256k1 for all
digital signatures in both votes and quorum certificates. BFT-SMaRt
uses hmac-sha1 for MACs (Message Authentication Codes) in the messages
during normal operation and uses digital signatures in addition to
MACs during a view change.
 
All results for \HotStuff reflect end-to-end measurement from the
clients. For BFT-SMaRt, we used the micro-benchmark programs
\texttt{Throughput\-Latency\-Server} and
\texttt{Throughput\-Latency\-Client} from the BFT-SMaRt website
(\url{https://github.com/bft-smart/library}). The client program
measures end-to-end latency but not throughput, while the server-side
program measures both throughput and latency. We used the throughput
results from servers and the latency results from clients.
 
\subsection{Base Performance}
We first measured throughput and latency in a setting commonly seen in
the evaluation of other BFT replication systems. We ran 4 replicas in
a configuration that tolerates a single failure, i.e., $f = 1$, while
varying the operation request rate until the system saturated.  This
benchmark used empty (zero-sized) operation requests and responses and
triggered no view changes; we expand to other settings below.
Although our responsive \HotStuff is three-phase, we also run its two-phase
variant as an additional baseline, because the BFT-SMaRt baseline has only two phases.

\begin{center}
    \begin{minipage}{0.47\linewidth}
        \includegraphics[width=\linewidth]{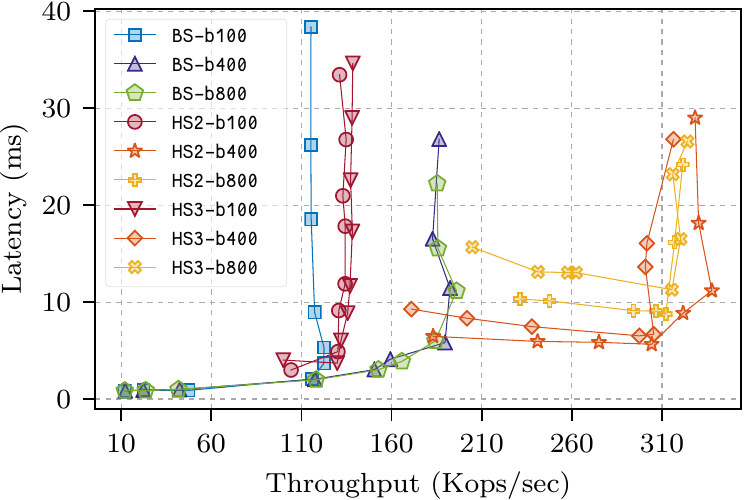}
        \captionof{figure}{Throughput vs.\@ latency with different choices of batch size, 4 replicas, 0/0 payload.}\label{fig:thr-lat1}
    \end{minipage}
    \hfill
    \begin{minipage}{0.47\linewidth}
        \includegraphics[width=\linewidth]{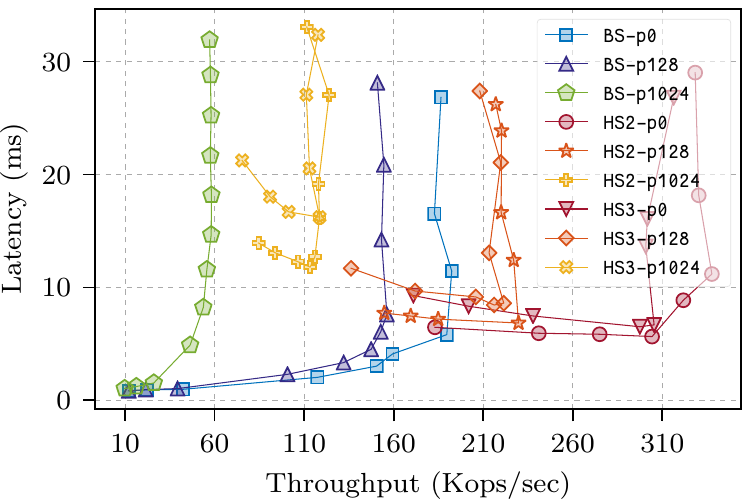}
        \captionof{figure}{Throughput vs.\@ latency with different choices of payload size, 4 replicas, batch size of 400.}\label{fig:thr-lat2}
    \end{minipage}
\end{center}

Figure~\ref{fig:thr-lat1} depicts three batch sizes for both systems,
100, 400, and 800, though because these systems have different
batching schemes, these numbers mean slightly different things for
each system.  BFT-SMaRt drives a separate consensus decision for each
operation, and batches the messages from multiple consensus protocols.
Therefore, it has a typical L-shaped latency/throughput performance
curve.  \HotStuff batches multiple operations in each node, and in
this way, mitigates the cost of digital signatures per
decision. However, above $400$ operations per batch, the latency
incurred by batching becomes higher than the cost of the crypto.  Despite these
differences, both three-phase (``\verb|HS3-|'') and two-phase (``\verb|HS2-|'')
\HotStuff achieves comparable latency performance to BFT-SMaRt (``\verb|BS-|'')
for all three batch sizes, while their maximum throughput noticeably
outperformed BFT-SMaRt.

For batch sizes of 100 and 400, the lowest-latency \HotStuff point
provides latency and throughput that are better than the latency and
throughput simultaneously achievable by BFT-SMaRT at its highest
throughput, while incurring a small increase in latency.  
This increase is partly due to the batching strategy employed by \HotStuff: It
needs three additional full batches (two in the two-phase variant) to arrive at a decision on a batch.
Our experiments kept the number of outstanding requests high, but the higher the
batch size, the longer it takes to fill the batching pipeline.
Practical deployments could be further optimized to adapt the
batch size to the number of outstanding operations.  

Figure~\ref{fig:thr-lat2} depicts three client request/reply payload
sizes (in bytes) of 0/0, 128/128, and 1024/1024, denoted
``\verb|p0|'', ``\verb|p128|'', and ``\verb|p1024|'' respectively.  At all payload sizes,
both three-phase and two-phase \HotStuff outperformed BFT-SMaRt in throughput,
with similar or comparable latency.

Notice BFT-SMaRt uses MACs based on symmetric crypto that is orders of
magnitude faster than the asymmetric crypto in digital signatures used by
\HotStuff, and also three-phase HotStuff has more round trips compared to
two-phase PBFT variant used by BFT-SMaRt. Yet \HotStuff is still able to
achieve comparable latency and much higher throughput.
Below we evaluate both systems in more challenging situations, where the
performance advantages of \HotStuff will become more pronounced.

\subsection{Scalability}
To evaluate the scalability of \HotStuff in various dimensions, we
performed three experiments. For the baseline, we used zero-size
request/response payloads while varying the number of replicas.  The
second evaluation repeated the baseline experiment with 128-byte and 1024-byte
request/response payloads.  The third test repeated the baseline (with
empty payloads) while introducing network delays between replicas
that were uniformly distributed in 5ms $\pm$ 0.5ms or in 10ms $\pm$
1.0ms, implemented using \verb|NetEm| (see
\url{https://www.linux.org/docs/man8/tc-netem.html}).
For each data point, we repeated five runs with the same setting and show error
bars to indicate the standard deviation for all runs.

\begin{figure}[b]
    \begin{subfigure}{0.47\linewidth}
        \includegraphics[width=\linewidth]{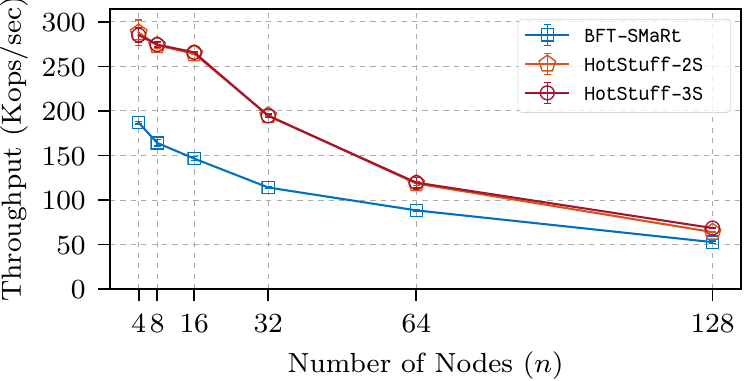}
        \caption{Throughput}\label{fig:n-thr-0}
    \end{subfigure}
    \hfill
    \begin{subfigure}{0.47\linewidth}
        \includegraphics[width=\linewidth]{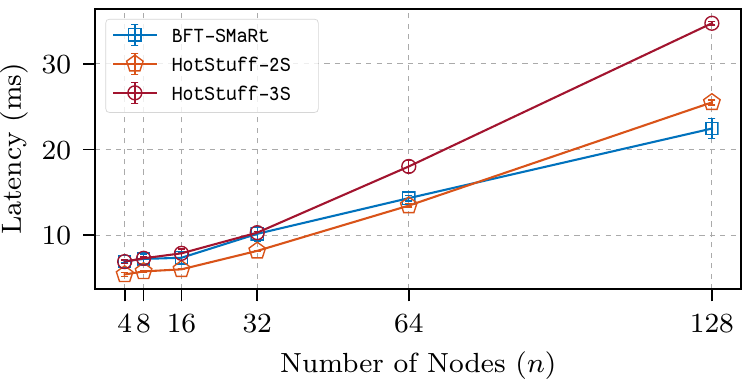}
        \caption{Latency}\label{fig:n-lat-0}
    \end{subfigure}
    \caption{Scalability with 0/0 payload, batch size of 400.}
\end{figure}

The first setting is depicted in
Figure~\ref{fig:n-thr-0} (throughput) and Figure~\ref{fig:n-lat-0}
(latency). Both three-phase and two-phase \HotStuff show consistently better
throughput than BFT-SMaRt, while their latencies are still comparable to BFT-SMaRt
with graceful degradation. The performance scales better than
BFT-SMaRt when $n < 32$. This is because we currently still use a list of
secp256k1 signatures for a QC\@. In the future, we plan to reduce the
cryptographic computation overhead in \HotStuff by using a fast threshold signature scheme.

The second setting with payload size 128 or 1024 bytes is denoted by
``\verb|p128|'' or ``\verb|p1024|'' in Figure~\ref{fig:n-thr-1} (throughput)
and Figure~\ref{fig:n-lat-1} (latency).  Due to its quadratic bandwidth cost,
the throughput of BFT-SMaRt scales worse than \HotStuff for reasonably large (1024-byte) payload size.

The third setting is shown in Figure~\ref{fig:n-thr-2} (throughput)
and Figure~\ref{fig:n-lat-2} (latency) as ``\verb|5ms|'' or
``\verb|10ms|''.  Again, due to the larger use of communication in
BFT-SMaRt, \HotStuff consistently outperformed BFT-SMaRt in both
cases.


\begin{figure}
    \begin{subfigure}{0.47\linewidth}
        \includegraphics[width=\linewidth]{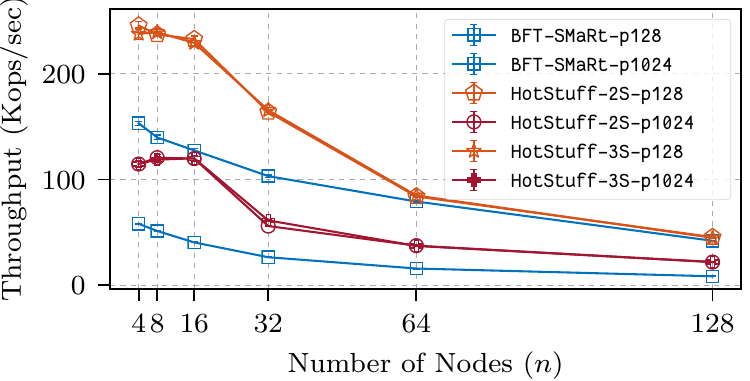}
        \caption{Throughput}\label{fig:n-thr-1}
    \end{subfigure}
    \hfill
    \begin{subfigure}{0.47\linewidth}
        \includegraphics[width=\linewidth]{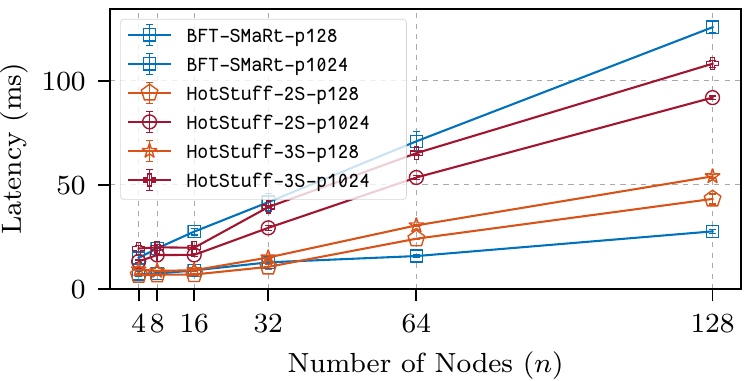}
        \caption{Latency}\label{fig:n-lat-1}
    \end{subfigure}
    \caption{Scalability for 128/128 payload or 1024/1024 payload, with batch size of 400.}
\end{figure}

\begin{figure}
    \begin{subfigure}{0.47\linewidth}
        \includegraphics[width=\linewidth]{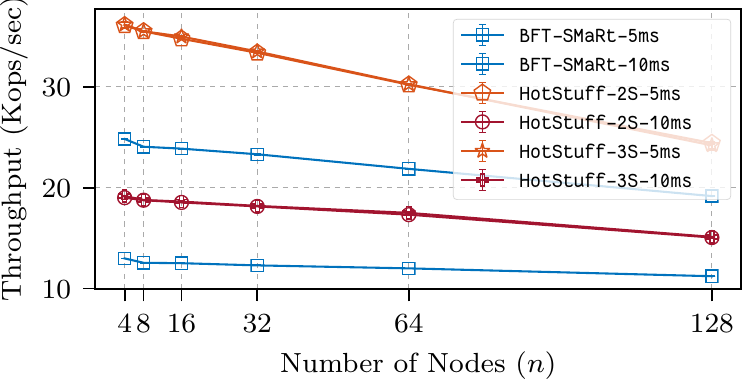}
        \caption{Throughput}\label{fig:n-thr-2}
    \end{subfigure}
    \hfill
    \begin{subfigure}{0.47\linewidth}
        \includegraphics[width=\linewidth]{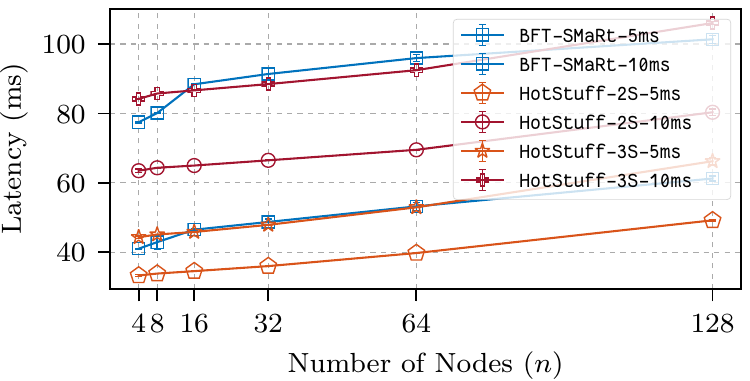}
        \caption{Latency}\label{fig:n-lat-2}
    \end{subfigure}
    \caption{Scalability for inter-replica latency 5ms $\pm$ 0.5ms or 10ms $\pm$ 1.0ms, with 0/0 payload, batch size of 400.}
\end{figure}

\subsection{View Change}

To evaluate the communication complexity of
leader replacement, we counted the number of \emph{MAC or signature
verifications} performed within BFT-SMaRt's view-change protocol. 
Our evaluation strategy was as follows. We injected a view change into BFT-SMaRt every one
thousand decisions. We instrumented the BFT-SMaRt source code to count the number of verifications
upon receiving and processing messages within the view-change protocol. 
Beyond communication complexity, 
this measurement underscores the cryptographic computation load
associated with transferring these authenticated values. 

Figure~\ref{fig:n-nauth} and Figure~\ref{fig:n-nauth-sig} show the number of
extra authenticators (MACs and signatures, respectively) processed for each view change,
where ``extra'' is defined to be those authenticators that would not be sent if the leader remained stable.  Note that \HotStuff has no ``extra'' authenticators by this definition, since the number of authenticators remains the same regardless of whether the leader stays the same or not.
The two figures show that BFT-SMaRt uses cubic numbers of MACs and quadratic
numbers of signatures. \HotStuff does not require extra authenticators for
view changes and so is omitted from the graph.

Evaluating the real-time performance of leader replacement is tricky. 
First, BFT-SMaRt got stuck when triggering frequent view changes; our
authenticator-counting benchmark had to average over as many successful
view changes as possible before the system got stuck, repeating the experiment
many times. Second, the actual elapsed time for leader replacement depends
highly on timeout parameters
and the leader-election
mechanism. It is therefore impossible to provide a meaningful comparison.

\begin{figure}
    \begin{subfigure}{0.49\textwidth}
    \includegraphics[width=\linewidth]{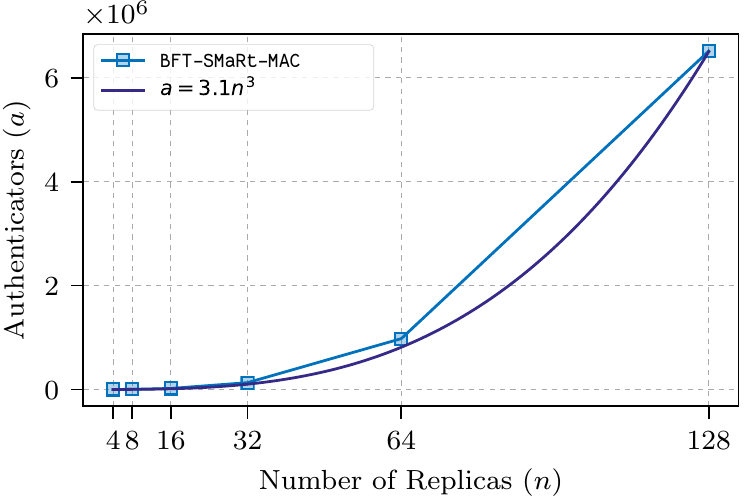}
    \caption{MACs}\label{fig:n-nauth}
    \end{subfigure}
    \begin{subfigure}{0.49\textwidth}
    \includegraphics[width=\linewidth]{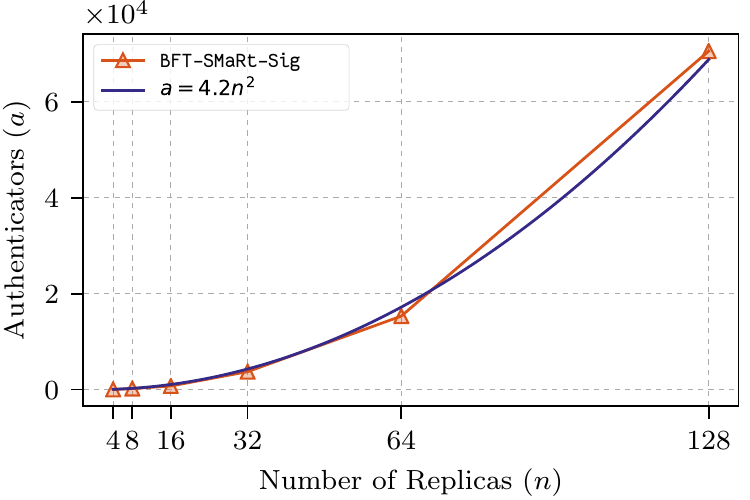}
    \caption{Signatures}\label{fig:n-nauth-sig}
    \end{subfigure}
    \captionof{figure}{Number of extra authenticators used for each BFT-SMaRt view change.}
\end{figure}

\section{Conclusion}

Since the introduction of PBFT, the first practical BFT replication solution in
the partial synchrony model,
numerous BFT solutions were built around its core two-phase paradigm. 
The first
phase guarantees proposal uniqueness through a \QC. The second phase guarantees
that a new leader can convince replicas to vote for a safe proposal. This
requires the leader to relay information from $(n-f)$ replicas, each reporting
its own highest \QC\ or vote. Generations of two-phase works thus suffer from a
quadratic communication bottleneck on leader replacement.

\HotStuff\ revolves around a three-phase core, allowing a new leader to simply
pick the highest \QC\ it knows of. This alleviates the above complexity
and at the same time considerably simplifies the leader replacement protocol.
Having (almost) canonized the phases, it is very easy to pipeline \HotStuff, and
to frequently rotate leaders.

\section*{Acknowledgements}
We are thankful to
Mathieu Baudet,
Avery Ching,
George Danezis,
Fran\c{c}ois Garillot,
Zekun Li,
Ben Maurer,
Kartik Nayak,
Dmitri Perelman,
and Ling Ren,
for many deep discussions of HotStuff,
and to Mathieu Baudet for exposing a subtle error in a previous
version posted to the ArXiv of this manuscript.

\begingroup
\bibliographystyle{plain}


\begin{thebibliography}{00}


\ifx \showCODEN    \undefined \def \showCODEN     #1{\unskip}     \fi
\ifx \showDOI      \undefined \def \showDOI       #1{#1}\fi
\ifx \showISBNx    \undefined \def \showISBNx     #1{\unskip}     \fi
\ifx \showISBNxiii \undefined \def \showISBNxiii  #1{\unskip}     \fi
\ifx \showISSN     \undefined \def \showISSN      #1{\unskip}     \fi
\ifx \showLCCN     \undefined \def \showLCCN      #1{\unskip}     \fi
\ifx \shownote     \undefined \def \shownote      #1{#1}          \fi
\ifx \showarticletitle \undefined \def \showarticletitle #1{#1}   \fi
\ifx \showURL      \undefined \def \showURL       {\relax}        \fi
\providecommand\bibfield[2]{#2}
\providecommand\bibinfo[2]{#2}
\providecommand\natexlab[1]{#1}
\providecommand\showeprint[2][]{arXiv:#2}

\bibitem[\protect\citeauthoryear{??}{Cas}{2017}]%
        {CasperOneMessage}
 \bibinfo{year}{2017}\natexlab{}.
\newblock \bibinfo{title}{Casper FFG with One Message Type, and Simpler Fork
  Choice Rule}.
\newblock
  \bibinfo{howpublished}{\url{https://ethresear.ch/t/casper-ffg-with-one-message-type-and-simpler-fork-choice-rule/103}}.
    (\bibinfo{year}{2017}).
\newblock


\bibitem[\protect\citeauthoryear{??}{IBF}{2018}]%
        {IBFTLivenessBug}
 \bibinfo{year}{2018}\natexlab{}.
\newblock \bibinfo{title}{Istanbul BFT's Design Cannot Successfully Tolerate
  Fail-Stop Failures}.
\newblock
  \bibinfo{howpublished}{\url{https://github.com/jpmorganchase/quorum/issues/305}}.
    (\bibinfo{year}{2018}).
\newblock


\bibitem[\protect\citeauthoryear{??}{TML}{2018}]%
        {TMLivenessBug}
 \bibinfo{year}{2018}\natexlab{}.
\newblock \bibinfo{title}{A Livelock Bug in the Presence of Byzantine
  Validator}.
\newblock
  \bibinfo{howpublished}{\url{https://github.com/tendermint/tendermint/issues/1047}}.
    (\bibinfo{year}{2018}).
\newblock


\bibitem[\protect\citeauthoryear{Abraham, Gueta, Malkhi, Alvisi, Kotla, and
  Martin}{Abraham et~al\mbox{.}}{2017}]%
        {RevisitBFT}
\bibfield{author}{\bibinfo{person}{Ittai Abraham}, \bibinfo{person}{Guy Gueta},
  \bibinfo{person}{Dahlia Malkhi}, \bibinfo{person}{Lorenzo Alvisi},
  \bibinfo{person}{Ramakrishna Kotla}, {and} \bibinfo{person}{Jean{-}Philippe
  Martin}.} \bibinfo{year}{2017}\natexlab{}.
\newblock \showarticletitle{Revisiting Fast Practical Byzantine Fault
  Tolerance}.
\newblock \bibinfo{journal}{{\em CoRR\/}}  \bibinfo{volume}{abs/1712.01367}
  (\bibinfo{year}{2017}).
\newblock
\showeprint[arxiv]{1712.01367}


\bibitem[\protect\citeauthoryear{Abraham, Gueta, Malkhi, and Martin}{Abraham
  et~al\mbox{.}}{2018}]%
        {Zelma18}
\bibfield{author}{\bibinfo{person}{Ittai Abraham}, \bibinfo{person}{Guy Gueta},
  \bibinfo{person}{Dahlia Malkhi}, {and} \bibinfo{person}{Jean{-}Philippe
  Martin}.} \bibinfo{year}{2018}\natexlab{}.
\newblock \showarticletitle{Revisiting Fast Practical {Byzantine} Fault
  Tolerance: Thelma, Velma, and Zelma}.
\newblock \bibinfo{journal}{{\em CoRR\/}}  \bibinfo{volume}{abs/1801.10022}
  (\bibinfo{year}{2018}).
\newblock
\showeprint[arxiv]{1801.10022}


\bibitem[\protect\citeauthoryear{Abraham and Malkhi}{Abraham and
  Malkhi}{2017}]%
        {AM17}
\bibfield{author}{\bibinfo{person}{Ittai Abraham} {and} \bibinfo{person}{Dahlia
  Malkhi}.} \bibinfo{year}{2017}\natexlab{}.
\newblock \showarticletitle{The Blockchain Consensus Layer and {BFT}}.
\newblock \bibinfo{journal}{{\em Bulletin of the {EATCS}\/}}
  \bibinfo{volume}{123} (\bibinfo{year}{2017}).
\newblock


\bibitem[\protect\citeauthoryear{Abraham, Malkhi, Nayak, Ren, and
  Spiegelman}{Abraham et~al\mbox{.}}{2017}]%
        {Solida}
\bibfield{author}{\bibinfo{person}{Ittai Abraham}, \bibinfo{person}{Dahlia
  Malkhi}, \bibinfo{person}{Kartik Nayak}, \bibinfo{person}{Ling Ren}, {and}
  \bibinfo{person}{Alexander Spiegelman}.} \bibinfo{year}{2017}\natexlab{}.
\newblock \showarticletitle{Solida: {A} Blockchain Protocol Based on
  Reconfigurable Byzantine Consensus}. In \bibinfo{booktitle}{{\em 21st
  International Conference on Principles of Distributed Systems, {OPODIS} 2017,
  Lisbon, Portugal, December 18-20, 2017}}. \bibinfo{pages}{25:1--25:19}.
\newblock
\showDOI{%
\url{https://doi.org/10.4230/LIPIcs.OPODIS.2017.25}}


\bibitem[\protect\citeauthoryear{Abraham, Malkhi, and Spiegelman}{Abraham
  et~al\mbox{.}}{2019}]%
        {AMS19}
\bibfield{author}{\bibinfo{person}{Ittai Abraham}, \bibinfo{person}{Dahlia
  Malkhi}, {and} \bibinfo{person}{Alexander Spiegelman}.}
  \bibinfo{year}{2019}\natexlab{}.
\newblock \showarticletitle{Validated Asynchronous Byzantine Agreement with
  Optimal Resilience and Asymptotically Optimal Time and Word Communication}.
  In \bibinfo{booktitle}{{\em Proceedings of the 2019 {ACM} Symposium on
  Principles of Distributed Computing, {PODC} 2019, Toronto, ON, Canada, July
  29-August 2, 2019}}.
\newblock


\bibitem[\protect\citeauthoryear{Amir, Coan, Kirsch, and Lane}{Amir
  et~al\mbox{.}}{2011}]%
        {Prime11}
\bibfield{author}{\bibinfo{person}{Yair Amir}, \bibinfo{person}{Brian~A. Coan},
  \bibinfo{person}{Jonathan Kirsch}, {and} \bibinfo{person}{John Lane}.}
  \bibinfo{year}{2011}\natexlab{}.
\newblock \showarticletitle{Prime: {Byzantine} Replication under Attack}.
\newblock \bibinfo{journal}{{\em {IEEE} Trans. Dependable Sec. Comput.\/}}
  \bibinfo{volume}{8}, \bibinfo{number}{4} (\bibinfo{year}{2011}),
  \bibinfo{pages}{564--577}.
\newblock
\showDOI{%
\url{https://doi.org/10.1109/TDSC.2010.70}}


\bibitem[\protect\citeauthoryear{Attiya, Dwork, Lynch, and Stockmeyer}{Attiya
  et~al\mbox{.}}{1994}]%
        {ADLS94}
\bibfield{author}{\bibinfo{person}{Hagit Attiya}, \bibinfo{person}{Cynthia
  Dwork}, \bibinfo{person}{Nancy~A. Lynch}, {and} \bibinfo{person}{Larry~J.
  Stockmeyer}.} \bibinfo{year}{1994}\natexlab{}.
\newblock \showarticletitle{Bounds on the Time to Reach Agreement in the
  Presence of Timing Uncertainty}.
\newblock \bibinfo{journal}{{\em J. {ACM}\/}} \bibinfo{volume}{41},
  \bibinfo{number}{1} (\bibinfo{year}{1994}), \bibinfo{pages}{122--152}.
\newblock
\showDOI{%
\url{https://doi.org/10.1145/174644.174649}}


\bibitem[\protect\citeauthoryear{Aublin, Guerraoui, Knezevic, Qu{\'{e}}ma, and
  Vukolic}{Aublin et~al\mbox{.}}{2015}]%
        {SevenBFT}
\bibfield{author}{\bibinfo{person}{Pierre{-}Louis Aublin},
  \bibinfo{person}{Rachid Guerraoui}, \bibinfo{person}{Nikola Knezevic},
  \bibinfo{person}{Vivien Qu{\'{e}}ma}, {and} \bibinfo{person}{Marko Vukolic}.}
  \bibinfo{year}{2015}\natexlab{}.
\newblock \showarticletitle{The Next 700 {BFT} Protocols}.
\newblock \bibinfo{journal}{{\em {ACM} Trans. Comput. Syst.\/}}
  \bibinfo{volume}{32}, \bibinfo{number}{4} (\bibinfo{year}{2015}),
  \bibinfo{pages}{12:1--12:45}.
\newblock
\showDOI{%
\url{https://doi.org/10.1145/2658994}}


\bibitem[\protect\citeauthoryear{Ben{-}Or}{Ben{-}Or}{1983}]%
        {Ben-Or83}
\bibfield{author}{\bibinfo{person}{Michael Ben{-}Or}.}
  \bibinfo{year}{1983}\natexlab{}.
\newblock \showarticletitle{Another Advantage of Free Choice: Completely
  Asynchronous Agreement Protocols (Extended Abstract)}. In
  \bibinfo{booktitle}{{\em Proceedings of the Second Annual {ACM}
  {SIGACT-SIGOPS} Symposium on Principles of Distributed Computing, Montreal,
  Quebec, Canada, August 17-19, 1983}}. \bibinfo{pages}{27--30}.
\newblock
\showDOI{%
\url{https://doi.org/10.1145/800221.806707}}


\bibitem[\protect\citeauthoryear{Bessani, Sousa, and Alchieri}{Bessani
  et~al\mbox{.}}{2014}]%
        {BFTSMaRt}
\bibfield{author}{\bibinfo{person}{Alysson~Neves Bessani},
  \bibinfo{person}{Jo{\~{a}}o Sousa}, {and} \bibinfo{person}{Eduardo
  Ad{\'{\i}}lio~Pelinson Alchieri}.} \bibinfo{year}{2014}\natexlab{}.
\newblock \showarticletitle{State Machine Replication for the Masses with
  {BFT-SMART}}. In \bibinfo{booktitle}{{\em 44th Annual {IEEE/IFIP}
  International Conference on Dependable Systems and Networks, {DSN} 2014,
  Atlanta, GA, USA, June 23-26, 2014}}. \bibinfo{pages}{355--362}.
\newblock
\showDOI{%
\url{https://doi.org/10.1109/DSN.2014.43}}


\bibitem[\protect\citeauthoryear{Boneh, Lynn, and Shacham}{Boneh
  et~al\mbox{.}}{2004}]%
        {BLS04}
\bibfield{author}{\bibinfo{person}{Dan Boneh}, \bibinfo{person}{Ben Lynn},
  {and} \bibinfo{person}{Hovav Shacham}.} \bibinfo{year}{2004}\natexlab{}.
\newblock \showarticletitle{Short Signatures from the Weil Pairing}.
\newblock \bibinfo{journal}{{\em J. Cryptology\/}} \bibinfo{volume}{17},
  \bibinfo{number}{4} (\bibinfo{year}{2004}), \bibinfo{pages}{297--319}.
\newblock
\showDOI{%
\url{https://doi.org/10.1007/s00145-004-0314-9}}


\bibitem[\protect\citeauthoryear{Buchman}{Buchman}{2016}]%
        {TendermintThesis}
\bibfield{author}{\bibinfo{person}{Ethan Buchman}.}
  \bibinfo{year}{2016}\natexlab{}.
\newblock {\em \bibinfo{title}{Tendermint: Byzantine fault tolerance in the age
  of blockchains}}.
\newblock \bibinfo{thesistype}{Ph.D. Dissertation}.
\newblock


\bibitem[\protect\citeauthoryear{Buchman, Kwon, and Milosevic}{Buchman
  et~al\mbox{.}}{2018}]%
        {TendermintGossip}
\bibfield{author}{\bibinfo{person}{Ethan Buchman}, \bibinfo{person}{Jae Kwon},
  {and} \bibinfo{person}{Zarko Milosevic}.} \bibinfo{year}{2018}\natexlab{}.
\newblock \showarticletitle{The Latest Gossip on {BFT} Consensus}.
\newblock \bibinfo{journal}{{\em CoRR\/}}  \bibinfo{volume}{abs/1807.04938}
  (\bibinfo{year}{2018}).
\newblock
\showeprint[arxiv]{1807.04938}


\bibitem[\protect\citeauthoryear{Buterin and Griffith}{Buterin and
  Griffith}{2017}]%
        {Casper}
\bibfield{author}{\bibinfo{person}{Vitalik Buterin} {and}
  \bibinfo{person}{Virgil Griffith}.} \bibinfo{year}{2017}\natexlab{}.
\newblock \showarticletitle{Casper the Friendly Finality Gadget}.
\newblock \bibinfo{journal}{{\em CoRR\/}}  \bibinfo{volume}{abs/1710.09437}
  (\bibinfo{year}{2017}).
\newblock
\showeprint[arxiv]{1710.09437}


\bibitem[\protect\citeauthoryear{Cachin, Kursawe, and Shoup}{Cachin
  et~al\mbox{.}}{2005}]%
        {CKS05}
\bibfield{author}{\bibinfo{person}{Christian Cachin}, \bibinfo{person}{Klaus
  Kursawe}, {and} \bibinfo{person}{Victor Shoup}.}
  \bibinfo{year}{2005}\natexlab{}.
\newblock \showarticletitle{Random Oracles in Constantinople: Practical
  Asynchronous Byzantine Agreement Using Cryptography}.
\newblock \bibinfo{journal}{{\em J. Cryptology\/}} \bibinfo{volume}{18},
  \bibinfo{number}{3} (\bibinfo{year}{2005}), \bibinfo{pages}{219--246}.
\newblock
\showDOI{%
\url{https://doi.org/10.1007/s00145-005-0318-0}}


\bibitem[\protect\citeauthoryear{Cachin and Vukolic}{Cachin and
  Vukolic}{2017}]%
        {CV17}
\bibfield{author}{\bibinfo{person}{Christian Cachin} {and}
  \bibinfo{person}{Marko Vukolic}.} \bibinfo{year}{2017}\natexlab{}.
\newblock \showarticletitle{Blockchain Consensus Protocols in the Wild}.
\newblock \bibinfo{journal}{{\em CoRR\/}}  \bibinfo{volume}{abs/1707.01873}
  (\bibinfo{year}{2017}).
\newblock
\showeprint[arxiv]{1707.01873}


\bibitem[\protect\citeauthoryear{Castro and Liskov}{Castro and Liskov}{1999}]%
        {CL99}
\bibfield{author}{\bibinfo{person}{Miguel Castro} {and}
  \bibinfo{person}{Barbara Liskov}.} \bibinfo{year}{1999}\natexlab{}.
\newblock \showarticletitle{Practical Byzantine Fault Tolerance}. In
  \bibinfo{booktitle}{{\em Proceedings of the Third {USENIX} Symposium on
  Operating Systems Design and Implementation (OSDI), New Orleans, Louisiana,
  USA, February 22-25, 1999}}. \bibinfo{pages}{173--186}.
\newblock
\showURL{%
\url{https://dl.acm.org/citation.cfm?id=296824}}


\bibitem[\protect\citeauthoryear{Castro and Liskov}{Castro and Liskov}{2002}]%
        {CL02}
\bibfield{author}{\bibinfo{person}{Miguel Castro} {and}
  \bibinfo{person}{Barbara Liskov}.} \bibinfo{year}{2002}\natexlab{}.
\newblock \showarticletitle{Practical Byzantine Fault Tolerance and Proactive
  Recovery}.
\newblock \bibinfo{journal}{{\em {ACM} Trans. Comput. Syst.\/}}
  \bibinfo{volume}{20}, \bibinfo{number}{4} (\bibinfo{year}{2002}),
  \bibinfo{pages}{398--461}.
\newblock
\showDOI{%
\url{https://doi.org/10.1145/571637.571640}}


\bibitem[\protect\citeauthoryear{Clement, Kapritsos, Lee, Wang, Alvisi, Dahlin,
  and Riche}{Clement et~al\mbox{.}}{2009}]%
        {Upright}
\bibfield{author}{\bibinfo{person}{Allen Clement}, \bibinfo{person}{Manos
  Kapritsos}, \bibinfo{person}{Sangmin Lee}, \bibinfo{person}{Yang Wang},
  \bibinfo{person}{Lorenzo Alvisi}, \bibinfo{person}{Michael Dahlin}, {and}
  \bibinfo{person}{Taylor Riche}.} \bibinfo{year}{2009}\natexlab{}.
\newblock \showarticletitle{Upright cluster services}. In
  \bibinfo{booktitle}{{\em Proceedings of the 22nd {ACM} Symposium on Operating
  Systems Principles 2009, {SOSP} 2009, Big Sky, Montana, USA, October 11-14,
  2009}}. \bibinfo{pages}{277--290}.
\newblock
\showDOI{%
\url{https://doi.org/10.1145/1629575.1629602}}


\bibitem[\protect\citeauthoryear{Dolev and Reischuk}{Dolev and
  Reischuk}{1985}]%
        {DR85}
\bibfield{author}{\bibinfo{person}{Danny Dolev} {and}
  \bibinfo{person}{R{\"{u}}diger Reischuk}.} \bibinfo{year}{1985}\natexlab{}.
\newblock \showarticletitle{Bounds on Information Exchange for Byzantine
  Agreement}.
\newblock \bibinfo{journal}{{\em J. {ACM}\/}} \bibinfo{volume}{32},
  \bibinfo{number}{1} (\bibinfo{year}{1985}), \bibinfo{pages}{191--204}.
\newblock
\showDOI{%
\url{https://doi.org/10.1145/2455.214112}}


\bibitem[\protect\citeauthoryear{Dolev and Strong}{Dolev and Strong}{1982}]%
        {DS82}
\bibfield{author}{\bibinfo{person}{Danny Dolev} {and}
  \bibinfo{person}{H.~Raymond Strong}.} \bibinfo{year}{1982}\natexlab{}.
\newblock \showarticletitle{Polynomial Algorithms for Multiple Processor
  Agreement}. In \bibinfo{booktitle}{{\em Proceedings of the 14th Annual {ACM}
  Symposium on Theory of Computing, May 5-7, 1982, San Francisco, California,
  {USA}}}. \bibinfo{pages}{401--407}.
\newblock
\showDOI{%
\url{https://doi.org/10.1145/800070.802215}}


\bibitem[\protect\citeauthoryear{Dwork, Lynch, and Stockmeyer}{Dwork
  et~al\mbox{.}}{1988}]%
        {DLS88}
\bibfield{author}{\bibinfo{person}{Cynthia Dwork}, \bibinfo{person}{Nancy~A.
  Lynch}, {and} \bibinfo{person}{Larry~J. Stockmeyer}.}
  \bibinfo{year}{1988}\natexlab{}.
\newblock \showarticletitle{Consensus in the Presence of Partial Synchrony}.
\newblock \bibinfo{journal}{{\em J. {ACM}\/}} \bibinfo{volume}{35},
  \bibinfo{number}{2} (\bibinfo{year}{1988}), \bibinfo{pages}{288--323}.
\newblock
\showDOI{%
\url{https://doi.org/10.1145/42282.42283}}


\bibitem[\protect\citeauthoryear{Eyal, Gencer, Sirer, and van Renesse}{Eyal
  et~al\mbox{.}}{2016}]%
        {BitcoinNG}
\bibfield{author}{\bibinfo{person}{Ittay Eyal}, \bibinfo{person}{Adem~Efe
  Gencer}, \bibinfo{person}{Emin~G{\"{u}}n Sirer}, {and}
  \bibinfo{person}{Robbert van Renesse}.} \bibinfo{year}{2016}\natexlab{}.
\newblock \showarticletitle{Bitcoin-NG: {A} Scalable Blockchain Protocol}. In
  \bibinfo{booktitle}{{\em 13th {USENIX} Symposium on Networked Systems Design
  and Implementation, {NSDI} 2016, Santa Clara, CA, USA, March 16-18, 2016}}.
  \bibinfo{pages}{45--59}.
\newblock


\bibitem[\protect\citeauthoryear{Fischer, Lynch, and Paterson}{Fischer
  et~al\mbox{.}}{1985}]%
        {FLP85}
\bibfield{author}{\bibinfo{person}{Michael~J. Fischer},
  \bibinfo{person}{Nancy~A. Lynch}, {and} \bibinfo{person}{Mike Paterson}.}
  \bibinfo{year}{1985}\natexlab{}.
\newblock \showarticletitle{Impossibility of Distributed Consensus with One
  Faulty Process}.
\newblock \bibinfo{journal}{{\em J. {ACM}\/}} \bibinfo{volume}{32},
  \bibinfo{number}{2} (\bibinfo{year}{1985}), \bibinfo{pages}{374--382}.
\newblock
\showDOI{%
\url{https://doi.org/10.1145/3149.214121}}


\bibitem[\protect\citeauthoryear{Garay, Kiayias, and Leonardos}{Garay
  et~al\mbox{.}}{2015}]%
        {GKL15}
\bibfield{author}{\bibinfo{person}{Juan~A. Garay}, \bibinfo{person}{Aggelos
  Kiayias}, {and} \bibinfo{person}{Nikos Leonardos}.}
  \bibinfo{year}{2015}\natexlab{}.
\newblock \showarticletitle{The Bitcoin Backbone Protocol: Analysis and
  Applications}. In \bibinfo{booktitle}{{\em Advances in Cryptology -
  {EUROCRYPT} 2015 - 34th Annual International Conference on the Theory and
  Applications of Cryptographic Techniques, Sofia, Bulgaria, April 26-30, 2015,
  Proceedings, Part {II}}}. \bibinfo{pages}{281--310}.
\newblock
\showDOI{%
\url{https://doi.org/10.1007/978-3-662-46803-6\_10}}


\bibitem[\protect\citeauthoryear{Gilad, Hemo, Micali, Vlachos, and
  Zeldovich}{Gilad et~al\mbox{.}}{2017}]%
        {Algorand}
\bibfield{author}{\bibinfo{person}{Yossi Gilad}, \bibinfo{person}{Rotem Hemo},
  \bibinfo{person}{Silvio Micali}, \bibinfo{person}{Georgios Vlachos}, {and}
  \bibinfo{person}{Nickolai Zeldovich}.} \bibinfo{year}{2017}\natexlab{}.
\newblock \showarticletitle{Algorand: Scaling Byzantine Agreements for
  Cryptocurrencies}. In \bibinfo{booktitle}{{\em Proceedings of the 26th
  Symposium on Operating Systems Principles, Shanghai, China, October 28-31,
  2017}}. \bibinfo{pages}{51--68}.
\newblock
\showDOI{%
\url{https://doi.org/10.1145/3132747.3132757}}


\bibitem[\protect\citeauthoryear{Golan{-}Gueta, Abraham, Grossman, Malkhi,
  Pinkas, Reiter, Seredinschi, Tamir, and Tomescu}{Golan{-}Gueta
  et~al\mbox{.}}{2018}]%
        {SBFT}
\bibfield{author}{\bibinfo{person}{Guy Golan{-}Gueta}, \bibinfo{person}{Ittai
  Abraham}, \bibinfo{person}{Shelly Grossman}, \bibinfo{person}{Dahlia Malkhi},
  \bibinfo{person}{Benny Pinkas}, \bibinfo{person}{Michael~K. Reiter},
  \bibinfo{person}{Dragos{-}Adrian Seredinschi}, \bibinfo{person}{Orr Tamir},
  {and} \bibinfo{person}{Alin Tomescu}.} \bibinfo{year}{2018}\natexlab{}.
\newblock \showarticletitle{{SBFT:} a Scalable Decentralized Trust
  Infrastructure for Blockchains}.
\newblock \bibinfo{journal}{{\em CoRR\/}}  \bibinfo{volume}{abs/1804.01626}
  (\bibinfo{year}{2018}).
\newblock
\showeprint[arxiv]{1804.01626}


\bibitem[\protect\citeauthoryear{Hanke, Movahedi, and Williams}{Hanke
  et~al\mbox{.}}{2018}]%
        {Dfinity}
\bibfield{author}{\bibinfo{person}{Timo Hanke}, \bibinfo{person}{Mahnush
  Movahedi}, {and} \bibinfo{person}{Dominic Williams}.}
  \bibinfo{year}{2018}\natexlab{}.
\newblock \showarticletitle{{DFINITY} Technology Overview Series, Consensus
  System}.
\newblock \bibinfo{journal}{{\em CoRR\/}}  \bibinfo{volume}{abs/1805.04548}
  (\bibinfo{year}{2018}).
\newblock
\showeprint[arxiv]{1805.04548}


\bibitem[\protect\citeauthoryear{Katz and Koo}{Katz and Koo}{2009}]%
        {KK09}
\bibfield{author}{\bibinfo{person}{Jonathan Katz} {and}
  \bibinfo{person}{Chiu{-}Yuen Koo}.} \bibinfo{year}{2009}\natexlab{}.
\newblock \showarticletitle{On Expected Constant-Round Protocols for Byzantine
  Agreement}.
\newblock \bibinfo{journal}{{\em J. Comput. Syst. Sci.\/}}
  \bibinfo{volume}{75}, \bibinfo{number}{2} (\bibinfo{year}{2009}),
  \bibinfo{pages}{91--112}.
\newblock
\showDOI{%
\url{https://doi.org/10.1016/j.jcss.2008.08.001}}


\bibitem[\protect\citeauthoryear{Kokoris{-}Kogias, Jovanovic, Gailly, Khoffi,
  Gasser, and Ford}{Kokoris{-}Kogias et~al\mbox{.}}{2016}]%
        {KJ16}
\bibfield{author}{\bibinfo{person}{Eleftherios Kokoris{-}Kogias},
  \bibinfo{person}{Philipp Jovanovic}, \bibinfo{person}{Nicolas Gailly},
  \bibinfo{person}{Ismail Khoffi}, \bibinfo{person}{Linus Gasser}, {and}
  \bibinfo{person}{Bryan Ford}.} \bibinfo{year}{2016}\natexlab{}.
\newblock \showarticletitle{Enhancing Bitcoin Security and Performance with
  Strong Consistency via Collective Signing}.
\newblock \bibinfo{journal}{{\em CoRR\/}}  \bibinfo{volume}{abs/1602.06997}
  (\bibinfo{year}{2016}).
\newblock


\bibitem[\protect\citeauthoryear{Kotla, Alvisi, Dahlin, Clement, and
  Wong}{Kotla et~al\mbox{.}}{2009}]%
        {KAD09}
\bibfield{author}{\bibinfo{person}{Ramakrishna Kotla}, \bibinfo{person}{Lorenzo
  Alvisi}, \bibinfo{person}{Michael Dahlin}, \bibinfo{person}{Allen Clement},
  {and} \bibinfo{person}{Edmund~L. Wong}.} \bibinfo{year}{2009}\natexlab{}.
\newblock \showarticletitle{Zyzzyva: Speculative Byzantine Fault Tolerance}.
\newblock \bibinfo{journal}{{\em {ACM} Trans. Comput. Syst.\/}}
  \bibinfo{volume}{27}, \bibinfo{number}{4} (\bibinfo{year}{2009}),
  \bibinfo{pages}{7:1--7:39}.
\newblock
\showDOI{%
\url{https://doi.org/10.1145/1658357.1658358}}


\bibitem[\protect\citeauthoryear{Lamport}{Lamport}{1978}]%
        {Lamport78}
\bibfield{author}{\bibinfo{person}{Leslie Lamport}.}
  \bibinfo{year}{1978}\natexlab{}.
\newblock \showarticletitle{Time, Clocks, and the Ordering of Events in a
  Distributed System}.
\newblock \bibinfo{journal}{{\em Commun. {ACM}\/}} \bibinfo{volume}{21},
  \bibinfo{number}{7} (\bibinfo{year}{1978}), \bibinfo{pages}{558--565}.
\newblock
\showDOI{%
\url{https://doi.org/10.1145/359545.359563}}


\bibitem[\protect\citeauthoryear{Lamport}{Lamport}{1998}]%
        {Paxos}
\bibfield{author}{\bibinfo{person}{Leslie Lamport}.}
  \bibinfo{year}{1998}\natexlab{}.
\newblock \showarticletitle{The Part-Time Parliament}.
\newblock \bibinfo{journal}{{\em {ACM} Trans. Comput. Syst.\/}}
  \bibinfo{volume}{16}, \bibinfo{number}{2} (\bibinfo{year}{1998}),
  \bibinfo{pages}{133--169}.
\newblock
\showDOI{%
\url{https://doi.org/10.1145/279227.279229}}


\bibitem[\protect\citeauthoryear{Lamport, Shostak, and Pease}{Lamport
  et~al\mbox{.}}{1982}]%
        {LSP82}
\bibfield{author}{\bibinfo{person}{Leslie Lamport}, \bibinfo{person}{Robert~E.
  Shostak}, {and} \bibinfo{person}{Marshall~C. Pease}.}
  \bibinfo{year}{1982}\natexlab{}.
\newblock \showarticletitle{The Byzantine Generals Problem}.
\newblock \bibinfo{journal}{{\em {ACM} Trans. Program. Lang. Syst.\/}}
  \bibinfo{volume}{4}, \bibinfo{number}{3} (\bibinfo{year}{1982}),
  \bibinfo{pages}{382--401}.
\newblock
\showDOI{%
\url{https://doi.org/10.1145/357172.357176}}


\bibitem[\protect\citeauthoryear{Mickens}{Mickens}{2014}]%
        {SaddestMoment}
\bibfield{author}{\bibinfo{person}{James Mickens}.}
  \bibinfo{year}{2014}\natexlab{}.
\newblock \showarticletitle{The Saddest Moment}.
\newblock \bibinfo{journal}{{\em ;login:\/}} \bibinfo{volume}{39},
  \bibinfo{number}{3} (\bibinfo{year}{2014}).
\newblock
\showURL{%
\url{https://www.usenix.org/publications/login/june14/mickens}}


\bibitem[\protect\citeauthoryear{Miller, Xia, Croman, Shi, and Song}{Miller
  et~al\mbox{.}}{2016}]%
        {HoneyBadger}
\bibfield{author}{\bibinfo{person}{Andrew Miller}, \bibinfo{person}{Yu Xia},
  \bibinfo{person}{Kyle Croman}, \bibinfo{person}{Elaine Shi}, {and}
  \bibinfo{person}{Dawn Song}.} \bibinfo{year}{2016}\natexlab{}.
\newblock \showarticletitle{The Honey Badger of {BFT} Protocols}. In
  \bibinfo{booktitle}{{\em Proceedings of the 2016 {ACM} {SIGSAC} Conference on
  Computer and Communications Security, Vienna, Austria, October 24-28, 2016}}.
  \bibinfo{pages}{31--42}.
\newblock
\showDOI{%
\url{https://doi.org/10.1145/2976749.2978399}}


\bibitem[\protect\citeauthoryear{Nakamoto}{Nakamoto}{2008}]%
        {Bitcoin}
\bibfield{author}{\bibinfo{person}{Satoshi Nakamoto}.}
  \bibinfo{year}{2008}\natexlab{}.
\newblock \bibinfo{title}{Bitcoin: A Peer-to-Peer Electronic Cash System}.
\newblock \bibinfo{howpublished}{\url{https://bitcoin.org/bitcoin.pdf}}.
  (\bibinfo{year}{2008}).
\newblock


\bibitem[\protect\citeauthoryear{Pass, Seeman, and Shelat}{Pass
  et~al\mbox{.}}{2017}]%
        {PSS17}
\bibfield{author}{\bibinfo{person}{Rafael Pass}, \bibinfo{person}{Lior Seeman},
  {and} \bibinfo{person}{Abhi Shelat}.} \bibinfo{year}{2017}\natexlab{}.
\newblock \showarticletitle{Analysis of the Blockchain Protocol in Asynchronous
  Networks}. In \bibinfo{booktitle}{{\em Advances in Cryptology - {EUROCRYPT}
  2017 - 36th Annual International Conference on the Theory and Applications of
  Cryptographic Techniques, Paris, France, April 30 - May 4, 2017, Proceedings,
  Part {II}}}. \bibinfo{pages}{643--673}.
\newblock
\showDOI{%
\url{https://doi.org/10.1007/978-3-319-56614-6\_22}}


\bibitem[\protect\citeauthoryear{Pass and Shi}{Pass and Shi}{2018}]%
        {Thunder}
\bibfield{author}{\bibinfo{person}{Rafael Pass} {and} \bibinfo{person}{Elaine
  Shi}.} \bibinfo{year}{2018}\natexlab{}.
\newblock \showarticletitle{Thunderella: Blockchains with Optimistic Instant
  Confirmation}. In \bibinfo{booktitle}{{\em Advances in Cryptology -
  {EUROCRYPT} 2018 - 37th Annual International Conference on the Theory and
  Applications of Cryptographic Techniques, Tel Aviv, Israel, April 29 - May 3,
  2018 Proceedings, Part {II}}}. \bibinfo{pages}{3--33}.
\newblock
\showDOI{%
\url{https://doi.org/10.1007/978-3-319-78375-8\_1}}


\bibitem[\protect\citeauthoryear{Pease, Shostak, and Lamport}{Pease
  et~al\mbox{.}}{1980}]%
        {PSL80}
\bibfield{author}{\bibinfo{person}{Marshall~C. Pease},
  \bibinfo{person}{Robert~E. Shostak}, {and} \bibinfo{person}{Leslie Lamport}.}
  \bibinfo{year}{1980}\natexlab{}.
\newblock \showarticletitle{Reaching Agreement in the Presence of Faults}.
\newblock \bibinfo{journal}{{\em J. {ACM}\/}} \bibinfo{volume}{27},
  \bibinfo{number}{2} (\bibinfo{year}{1980}), \bibinfo{pages}{228--234}.
\newblock
\showDOI{%
\url{https://doi.org/10.1145/322186.322188}}


\bibitem[\protect\citeauthoryear{Ramasamy and Cachin}{Ramasamy and
  Cachin}{2005}]%
        {RC05}
\bibfield{author}{\bibinfo{person}{HariGovind~V. Ramasamy} {and}
  \bibinfo{person}{Christian Cachin}.} \bibinfo{year}{2005}\natexlab{}.
\newblock \showarticletitle{Parsimonious Asynchronous Byzantine-Fault-Tolerant
  Atomic Broadcast}. In \bibinfo{booktitle}{{\em Principles of Distributed
  Systems, 9th International Conference, {OPODIS} 2005, Pisa, Italy, December
  12-14, 2005, Revised Selected Papers}}. \bibinfo{pages}{88--102}.
\newblock
\showDOI{%
\url{https://doi.org/10.1007/11795490\_9}}


\bibitem[\protect\citeauthoryear{Reiter}{Reiter}{1994}]%
        {R95}
\bibfield{author}{\bibinfo{person}{Michael~K. Reiter}.}
  \bibinfo{year}{1994}\natexlab{}.
\newblock \showarticletitle{The Rampart Toolkit for Building High-Integrity
  Services}. In \bibinfo{booktitle}{{\em Theory and Practice in Distributed
  Systems, International Workshop, Dagstuhl Castle, Germany, September 5-9,
  1994, Selected Papers}}. \bibinfo{pages}{99--110}.
\newblock
\showDOI{%
\url{https://doi.org/10.1007/3-540-60042-6\_7}}


\bibitem[\protect\citeauthoryear{Rogaway and Shrimpton}{Rogaway and
  Shrimpton}{2004}]%
        {RS04}
\bibfield{author}{\bibinfo{person}{Phillip Rogaway} {and}
  \bibinfo{person}{Thomas Shrimpton}.} \bibinfo{year}{2004}\natexlab{}.
\newblock \showarticletitle{Cryptographic Hash-Function Basics: Definitions,
  Implications, and Separations for Preimage Resistance, Second-Preimage
  Resistance, and Collision Resistance}. In \bibinfo{booktitle}{{\em Fast
  Software Encryption, 11th International Workshop, {FSE} 2004, Delhi, India,
  February 5-7, 2004, Revised Papers}} {\em (\bibinfo{series}{Lecture Notes in
  Computer Science})}, \bibfield{editor}{\bibinfo{person}{Bimal~K. Roy} {and}
  \bibinfo{person}{Willi Meier}} (Eds.), Vol.~\bibinfo{volume}{3017}.
  \bibinfo{publisher}{Springer}, \bibinfo{pages}{371--388}.
\newblock
\showDOI{%
\url{https://doi.org/10.1007/978-3-540-25937-4\_24}}


\bibitem[\protect\citeauthoryear{Schneider}{Schneider}{1990}]%
        {SMR}
\bibfield{author}{\bibinfo{person}{Fred~B. Schneider}.}
  \bibinfo{year}{1990}\natexlab{}.
\newblock \showarticletitle{Implementing Fault-Tolerant Services Using the
  State Machine Approach: {A} Tutorial}.
\newblock \bibinfo{journal}{{\em {ACM} Comput. Surv.\/}} \bibinfo{volume}{22},
  \bibinfo{number}{4} (\bibinfo{year}{1990}), \bibinfo{pages}{299--319}.
\newblock
\showDOI{%
\url{https://doi.org/10.1145/98163.98167}}


\bibitem[\protect\citeauthoryear{Shoup}{Shoup}{2000}]%
        {S00}
\bibfield{author}{\bibinfo{person}{Victor Shoup}.}
  \bibinfo{year}{2000}\natexlab{}.
\newblock \showarticletitle{Practical Threshold Signatures}. In
  \bibinfo{booktitle}{{\em Advances in Cryptology - {EUROCRYPT} 2000,
  International Conference on the Theory and Application of Cryptographic
  Techniques, Bruges, Belgium, May 14-18, 2000, Proceeding}}.
  \bibinfo{pages}{207--220}.
\newblock
\showDOI{%
\url{https://doi.org/10.1007/3-540-45539-6\_15}}


\bibitem[\protect\citeauthoryear{Song and van Renesse}{Song and van
  Renesse}{2008}]%
        {Bosco}
\bibfield{author}{\bibinfo{person}{Yee~Jiun Song} {and}
  \bibinfo{person}{Robbert van Renesse}.} \bibinfo{year}{2008}\natexlab{}.
\newblock \showarticletitle{Bosco: One-Step Byzantine Asynchronous Consensus}.
  In \bibinfo{booktitle}{{\em Distributed Computing, 22nd International
  Symposium, {DISC} 2008, Arcachon, France, September 22-24, 2008.
  Proceedings}}. \bibinfo{pages}{438--450}.
\newblock
\showDOI{%
\url{https://doi.org/10.1007/978-3-540-87779-0\_30}}


\bibitem[\protect\citeauthoryear{Yin, Malkhi, Reiter, Gueta, and Abraham}{Yin
  et~al\mbox{.}}{2018}]%
        {HotStuff}
\bibfield{author}{\bibinfo{person}{Maofan Yin}, \bibinfo{person}{Dahlia
  Malkhi}, \bibinfo{person}{Michael~K. Reiter}, \bibinfo{person}{Guy~Golan
  Gueta}, {and} \bibinfo{person}{Ittai Abraham}.}
  \bibinfo{year}{2018}\natexlab{}.
\newblock \showarticletitle{HotStuff: BFT Consensus in the Lens of Blockchain}.
\newblock \bibinfo{journal}{{\em CoRR\/}}  \bibinfo{volume}{abs/1803.05069}
  (\bibinfo{year}{2018}).
\newblock
\showeprint[arxiv]{1803.05069}


\end{thebibliography}


\begin{thebibliography}{10}

\bibitem{CasperOneMessage}
Casper ffg with one message type, and simpler fork choice rule.
\newblock
  \url{https://ethresear.ch/t/casper-ffg-with-one-message-type-and-simpler-fork-choice-rule/103},
  2017.

\bibitem{IBFTLivenessBug}
Istanbul bft's design cannot successfully tolerate fail-stop failures.
\newblock \url{https://github.com/jpmorganchase/quorum/issues/305}, 2018.

\bibitem{TMLivenessBug}
A livelock bug in the presence of byzantine validator.
\newblock \url{https://github.com/tendermint/tendermint/issues/1047}, 2018.

\bibitem{RevisitBFT}
Ittai Abraham, Guy Gueta, Dahlia Malkhi, Lorenzo Alvisi, Ramakrishna Kotla, and
  Jean{-}Philippe Martin.
\newblock Revisiting fast practical byzantine fault tolerance.
\newblock {\em CoRR}, abs/1712.01367, 2017.

\bibitem{Zelma18}
Ittai Abraham, Guy Gueta, Dahlia Malkhi, and Jean{-}Philippe Martin.
\newblock Revisiting fast practical {Byzantine} fault tolerance: Thelma, velma,
  and zelma.
\newblock {\em CoRR}, abs/1801.10022, 2018.

\bibitem{AM17}
Ittai Abraham and Dahlia Malkhi.
\newblock The blockchain consensus layer and {BFT}.
\newblock {\em Bulletin of the {EATCS}}, 123, 2017.

\bibitem{Solida}
Ittai Abraham, Dahlia Malkhi, Kartik Nayak, Ling Ren, and Alexander Spiegelman.
\newblock Solida: {A} blockchain protocol based on reconfigurable byzantine
  consensus.
\newblock In {\em 21st International Conference on Principles of Distributed
  Systems, {OPODIS} 2017, Lisbon, Portugal, December 18-20, 2017}, pages
  25:1--25:19, 2017.

\bibitem{AMS19}
Ittai Abraham, Dahlia Malkhi, and Alexander Spiegelman.
\newblock Validated asynchronous byzantine agreement with optimal resilience
  and asymptotically optimal time and word communication.
\newblock In {\em Proceedings of the 2019 {ACM} Symposium on Principles of
  Distributed Computing, {PODC} 2019, Toronto, ON, Canada, July 29-August 2,
  2019}, 2019.

\bibitem{Prime11}
Yair Amir, Brian~A. Coan, Jonathan Kirsch, and John Lane.
\newblock Prime: {Byzantine} replication under attack.
\newblock {\em {IEEE} Trans. Dependable Sec. Comput.}, 8(4):564--577, 2011.

\bibitem{ADLS94}
Hagit Attiya, Cynthia Dwork, Nancy~A. Lynch, and Larry~J. Stockmeyer.
\newblock Bounds on the time to reach agreement in the presence of timing
  uncertainty.
\newblock {\em J. {ACM}}, 41(1):122--152, 1994.

\bibitem{SevenBFT}
Pierre{-}Louis Aublin, Rachid Guerraoui, Nikola Knezevic, Vivien Qu{\'{e}}ma,
  and Marko Vukolic.
\newblock The next 700 {BFT} protocols.
\newblock {\em {ACM} Trans. Comput. Syst.}, 32(4):12:1--12:45, 2015.

\bibitem{Ben-Or83}
Michael Ben{-}Or.
\newblock Another advantage of free choice: Completely asynchronous agreement
  protocols (extended abstract).
\newblock In {\em Proceedings of the Second Annual {ACM} {SIGACT-SIGOPS}
  Symposium on Principles of Distributed Computing, Montreal, Quebec, Canada,
  August 17-19, 1983}, pages 27--30, 1983.

\bibitem{BFTSMaRt}
Alysson~Neves Bessani, Jo{\~{a}}o Sousa, and Eduardo Ad{\'{\i}}lio~Pelinson
  Alchieri.
\newblock State machine replication for the masses with {BFT-SMART}.
\newblock In {\em 44th Annual {IEEE/IFIP} International Conference on
  Dependable Systems and Networks, {DSN} 2014, Atlanta, GA, USA, June 23-26,
  2014}, pages 355--362, 2014.

\bibitem{BLS04}
Dan Boneh, Ben Lynn, and Hovav Shacham.
\newblock Short signatures from the weil pairing.
\newblock {\em J. Cryptology}, 17(4):297--319, 2004.

\bibitem{TendermintThesis}
Ethan Buchman.
\newblock {\em Tendermint: Byzantine fault tolerance in the age of
  blockchains}.
\newblock PhD thesis, 2016.

\bibitem{TendermintGossip}
Ethan Buchman, Jae Kwon, and Zarko Milosevic.
\newblock The latest gossip on {BFT} consensus.
\newblock {\em CoRR}, abs/1807.04938, 2018.

\bibitem{Casper}
Vitalik Buterin and Virgil Griffith.
\newblock Casper the friendly finality gadget.
\newblock {\em CoRR}, abs/1710.09437, 2017.

\bibitem{CKS05}
Christian Cachin, Klaus Kursawe, and Victor Shoup.
\newblock Random oracles in constantinople: Practical asynchronous byzantine
  agreement using cryptography.
\newblock {\em J. Cryptology}, 18(3):219--246, 2005.

\bibitem{CV17}
Christian Cachin and Marko Vukolic.
\newblock Blockchain consensus protocols in the wild.
\newblock {\em CoRR}, abs/1707.01873, 2017.

\bibitem{CL99}
Miguel Castro and Barbara Liskov.
\newblock Practical byzantine fault tolerance.
\newblock In {\em Proceedings of the Third {USENIX} Symposium on Operating
  Systems Design and Implementation (OSDI), New Orleans, Louisiana, USA,
  February 22-25, 1999}, pages 173--186, 1999.

\bibitem{CL02}
Miguel Castro and Barbara Liskov.
\newblock Practical byzantine fault tolerance and proactive recovery.
\newblock {\em {ACM} Trans. Comput. Syst.}, 20(4):398--461, 2002.

\bibitem{Upright}
Allen Clement, Manos Kapritsos, Sangmin Lee, Yang Wang, Lorenzo Alvisi, Michael
  Dahlin, and Taylor Riche.
\newblock Upright cluster services.
\newblock In {\em Proceedings of the 22nd {ACM} Symposium on Operating Systems
  Principles 2009, {SOSP} 2009, Big Sky, Montana, USA, October 11-14, 2009},
  pages 277--290, 2009.

\bibitem{DR85}
Danny Dolev and R{\"{u}}diger Reischuk.
\newblock Bounds on information exchange for byzantine agreement.
\newblock {\em J. {ACM}}, 32(1):191--204, 1985.

\bibitem{DS82}
Danny Dolev and H.~Raymond Strong.
\newblock Polynomial algorithms for multiple processor agreement.
\newblock In {\em Proceedings of the 14th Annual {ACM} Symposium on Theory of
  Computing, May 5-7, 1982, San Francisco, California, {USA}}, pages 401--407,
  1982.

\bibitem{DLS88}
Cynthia Dwork, Nancy~A. Lynch, and Larry~J. Stockmeyer.
\newblock Consensus in the presence of partial synchrony.
\newblock {\em J. {ACM}}, 35(2):288--323, 1988.

\bibitem{BitcoinNG}
Ittay Eyal, Adem~Efe Gencer, Emin~G{\"{u}}n Sirer, and Robbert van Renesse.
\newblock Bitcoin-ng: {A} scalable blockchain protocol.
\newblock In {\em 13th {USENIX} Symposium on Networked Systems Design and
  Implementation, {NSDI} 2016, Santa Clara, CA, USA, March 16-18, 2016}, pages
  45--59, 2016.

\bibitem{FLP85}
Michael~J. Fischer, Nancy~A. Lynch, and Mike Paterson.
\newblock Impossibility of distributed consensus with one faulty process.
\newblock {\em J. {ACM}}, 32(2):374--382, 1985.

\bibitem{GKL15}
Juan~A. Garay, Aggelos Kiayias, and Nikos Leonardos.
\newblock The bitcoin backbone protocol: Analysis and applications.
\newblock In {\em Advances in Cryptology - {EUROCRYPT} 2015 - 34th Annual
  International Conference on the Theory and Applications of Cryptographic
  Techniques, Sofia, Bulgaria, April 26-30, 2015, Proceedings, Part {II}},
  pages 281--310, 2015.

\bibitem{Algorand}
Yossi Gilad, Rotem Hemo, Silvio Micali, Georgios Vlachos, and Nickolai
  Zeldovich.
\newblock Algorand: Scaling byzantine agreements for cryptocurrencies.
\newblock In {\em Proceedings of the 26th Symposium on Operating Systems
  Principles, Shanghai, China, October 28-31, 2017}, pages 51--68, 2017.

\bibitem{SBFT}
Guy Golan{-}Gueta, Ittai Abraham, Shelly Grossman, Dahlia Malkhi, Benny Pinkas,
  Michael~K. Reiter, Dragos{-}Adrian Seredinschi, Orr Tamir, and Alin Tomescu.
\newblock {SBFT:} a scalable decentralized trust infrastructure for
  blockchains.
\newblock {\em CoRR}, abs/1804.01626, 2018.

\bibitem{Dfinity}
Timo Hanke, Mahnush Movahedi, and Dominic Williams.
\newblock {DFINITY} technology overview series, consensus system.
\newblock {\em CoRR}, abs/1805.04548, 2018.

\bibitem{KK09}
Jonathan Katz and Chiu{-}Yuen Koo.
\newblock On expected constant-round protocols for byzantine agreement.
\newblock {\em J. Comput. Syst. Sci.}, 75(2):91--112, 2009.

\bibitem{KJ16}
Eleftherios Kokoris{-}Kogias, Philipp Jovanovic, Nicolas Gailly, Ismail Khoffi,
  Linus Gasser, and Bryan Ford.
\newblock Enhancing bitcoin security and performance with strong consistency
  via collective signing.
\newblock {\em CoRR}, abs/1602.06997, 2016.

\bibitem{KAD09}
Ramakrishna Kotla, Lorenzo Alvisi, Michael Dahlin, Allen Clement, and Edmund~L.
  Wong.
\newblock Zyzzyva: Speculative byzantine fault tolerance.
\newblock {\em {ACM} Trans. Comput. Syst.}, 27(4):7:1--7:39, 2009.

\bibitem{Lamport78}
Leslie Lamport.
\newblock Time, clocks, and the ordering of events in a distributed system.
\newblock {\em Commun. {ACM}}, 21(7):558--565, 1978.

\bibitem{Paxos}
Leslie Lamport.
\newblock The part-time parliament.
\newblock {\em {ACM} Trans. Comput. Syst.}, 16(2):133--169, 1998.

\bibitem{LSP82}
Leslie Lamport, Robert~E. Shostak, and Marshall~C. Pease.
\newblock The byzantine generals problem.
\newblock {\em {ACM} Trans. Program. Lang. Syst.}, 4(3):382--401, 1982.

\bibitem{SaddestMoment}
James Mickens.
\newblock The saddest moment.
\newblock {\em ;login:}, 39(3), 2014.

\bibitem{HoneyBadger}
Andrew Miller, Yu~Xia, Kyle Croman, Elaine Shi, and Dawn Song.
\newblock The honey badger of {BFT} protocols.
\newblock In {\em Proceedings of the 2016 {ACM} {SIGSAC} Conference on Computer
  and Communications Security, Vienna, Austria, October 24-28, 2016}, pages
  31--42, 2016.

\bibitem{Bitcoin}
Satoshi Nakamoto.
\newblock Bitcoin: A peer-to-peer electronic cash system.
\newblock \url{https://bitcoin.org/bitcoin.pdf}, 2008.

\bibitem{PSS17}
Rafael Pass, Lior Seeman, and Abhi Shelat.
\newblock Analysis of the blockchain protocol in asynchronous networks.
\newblock In {\em Advances in Cryptology - {EUROCRYPT} 2017 - 36th Annual
  International Conference on the Theory and Applications of Cryptographic
  Techniques, Paris, France, April 30 - May 4, 2017, Proceedings, Part {II}},
  pages 643--673, 2017.

\bibitem{Thunder}
Rafael Pass and Elaine Shi.
\newblock Thunderella: Blockchains with optimistic instant confirmation.
\newblock In {\em Advances in Cryptology - {EUROCRYPT} 2018 - 37th Annual
  International Conference on the Theory and Applications of Cryptographic
  Techniques, Tel Aviv, Israel, April 29 - May 3, 2018 Proceedings, Part {II}},
  pages 3--33, 2018.

\bibitem{PSL80}
Marshall~C. Pease, Robert~E. Shostak, and Leslie Lamport.
\newblock Reaching agreement in the presence of faults.
\newblock {\em J. {ACM}}, 27(2):228--234, 1980.

\bibitem{RC05}
HariGovind~V. Ramasamy and Christian Cachin.
\newblock Parsimonious asynchronous byzantine-fault-tolerant atomic broadcast.
\newblock In {\em Principles of Distributed Systems, 9th International
  Conference, {OPODIS} 2005, Pisa, Italy, December 12-14, 2005, Revised
  Selected Papers}, pages 88--102, 2005.

\bibitem{R95}
Michael~K. Reiter.
\newblock The rampart toolkit for building high-integrity services.
\newblock In {\em Theory and Practice in Distributed Systems, International
  Workshop, Dagstuhl Castle, Germany, September 5-9, 1994, Selected Papers},
  pages 99--110, 1994.

\bibitem{RS04}
Phillip Rogaway and Thomas Shrimpton.
\newblock Cryptographic hash-function basics: Definitions, implications, and
  separations for preimage resistance, second-preimage resistance, and
  collision resistance.
\newblock In Bimal~K. Roy and Willi Meier, editors, {\em Fast Software
  Encryption, 11th International Workshop, {FSE} 2004, Delhi, India, February
  5-7, 2004, Revised Papers}, volume 3017 of {\em Lecture Notes in Computer
  Science}, pages 371--388. Springer, 2004.

\bibitem{SMR}
Fred~B. Schneider.
\newblock Implementing fault-tolerant services using the state machine
  approach: {A} tutorial.
\newblock {\em {ACM} Comput. Surv.}, 22(4):299--319, 1990.

\bibitem{S00}
Victor Shoup.
\newblock Practical threshold signatures.
\newblock In {\em Advances in Cryptology - {EUROCRYPT} 2000, International
  Conference on the Theory and Application of Cryptographic Techniques, Bruges,
  Belgium, May 14-18, 2000, Proceeding}, pages 207--220, 2000.

\bibitem{Bosco}
Yee~Jiun Song and Robbert van Renesse.
\newblock Bosco: One-step byzantine asynchronous consensus.
\newblock In {\em Distributed Computing, 22nd International Symposium, {DISC}
  2008, Arcachon, France, September 22-24, 2008. Proceedings}, pages 438--450,
  2008.

\bibitem{HotStuff}
Maofan Yin, Dahlia Malkhi, Michael~K. Reiter, Guy~Golan Gueta, and Ittai
  Abraham.
\newblock Hotstuff: Bft consensus in the lens of blockchain.
\newblock {\em CoRR}, abs/1803.05069, 2018.

\end{thebibliography}

\endgroup
\begin{appendices}
\section{Proof of Safety for \HotStuffPro}\label{app:hs-proof}
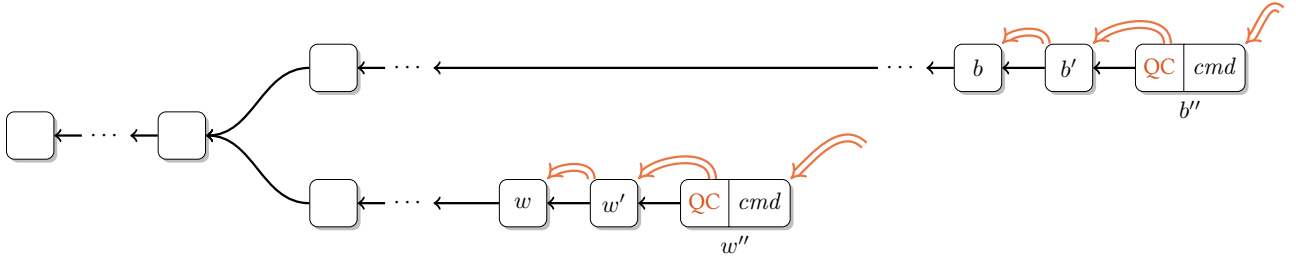
\begin{figure*}
\centering
\begin{tikzpicture}[x=1.12cm, scale=0.9, every node/.append style={transform shape}]
    \tikzstyle{block}=[drop shadow={shadow xshift=0.3ex,shadow yshift=-0.3ex}, rounded corners=0.7ex]
    \tikzstyle{data}=[rectangle split,rectangle split horizontal,rectangle split parts=2,rectangle split part align=base,draw,text centered,block]
    \definecolor{darkorange}{HTML}{d95319}
    \definecolor{orange}{HTML}{eb7645}
    \begin{scope}[all/.style={draw, minimum height=0.7cm, minimum width=0.7cm},line width=0.08ex]
        \node[all, block, draw,fill=white] (b1) at (1, 0) {}; 
        \node[all, block, draw,fill=white] (b3) at (3, 0) {};
        \node[all, block, draw,fill=white] (b4) at (5, -1) {};
        \node[all, block, draw,fill=white] (b5) at (5, 1) {};
        \node[all, draw=none, fill=none] (b2) at (2, 0) {$\cdots$};
        \node[all, draw=none,fill=none] (c5) at (6, -1) {$\cdots$};
        \node[all, draw=none,fill=none] (c6) at (6, 1) {$\cdots$};
        \node[all, draw=none,fill=none] (c7) at (12.5, 1.5) {};
        \node[all, draw=none,fill=none] (c8) at (12, -0.5) {};
        \node[all, draw=none,fill=none] (c9) at (17.5, 1.5) {};
        \node[all, draw=none,fill=none] (c62) at (12.5, 1) {$\cdots$};
        \node[all, block, draw,fill=white] (b6a) at (7.5, -1) {$w$};
        \node[all, block, draw,fill=white] (b6) at (8.7, -1) {$w'$};
        \node[all, block, draw,fill=white] (b7a) at (13.5, 1) {$b$};
        \node[all, block, draw,fill=white] (b7) at (14.7, 1) {$b'$};
        \node[all, data, draw,fill=white] (b8) at (10.3, -1) {{\color{darkorange}\QC}\nodepart{two}$\cmd$};
        \node[all, data, draw,fill=white] (b9) at (16.3, 1) {{\color{darkorange}\QC}\nodepart{two}$\cmd$};
        \node[anchor=north] (b8l) at (b8.south) {$w''$};
        \node[anchor=north] (b9l) at (b9.south) {$b''$};
        \begin{scope}[line width=0.2ex]
        \path[<-] (b1) edge[out=0, in=180] node[sloped,above] {} (b2) ;
        \path[<-] (b2) edge[out=0, in=180] node[sloped,above] {} (b3) ;
        \path[<-] (b3) edge[out=0, in=180] node[sloped,above] {} (b4) ;
        \path[<-] (b3) edge[out=0, in=180] node[sloped,above] {} (b5) ;
        \path[<-] (b5) edge[out=0, in=180] node[sloped,above] {} (c6) ;
        \path[<-] (b4) edge[out=0, in=180] node[sloped,above] {} (c5) ;
        \path[<-] (c5) edge[out=0, in=180] node[sloped,above] {} (b6a) ;
        \path[<-] (b6a) edge[out=0, in=180] node[sloped,above] {} (b6) ;
        \path[<-] (c6) edge[out=0, in=180] node[sloped,above] {} (c62) ;
        \path[<-] (b7a) edge[out=0, in=180] node[sloped,above] {} (b7) ;
        \path[<-] (c62) edge[out=0, in=180] node[sloped,above] {} (b7a) ;
        \path[<-] (b6) edge[out=0, in=180] node[sloped,above] {} (b8) ;
        \path[<-] (b7) edge[out=0, in=180] node[sloped,above] {} (b9) ;
        \end{scope}
        \begin{scope}[color=orange,line width=0.2ex]
            \path[-implies] (b8.north) ++ (-0.3, 0) edge[double distance=0.4ex, out=90, in=north east] node[sloped,above] {} (b6);
            \path[-implies] (b6.north) ++ (-0.3, 0) edge[double distance=0.4ex, out=90, in=north east] node[sloped,above] {} (b6a);
            \path[-implies] (b9.north) ++ (-0.3, 0) edge[double distance=0.4ex, out=90, in=north east] node[sloped,above] {} (b7);
            \path[-implies] (b7.north) ++ (-0.3, 0) edge[double distance=0.4ex, out=90, in=north east] node[sloped,above] {} (b7a);
            \path[-implies] (c8.north) edge[double distance=0.4ex, in=north east, out=135] node[sloped,above] {} (b8);
            \path[-implies] (c9.north) edge[double distance=0.4ex, in=north east, out=135] node[sloped,above] {} (b9);
        \end{scope}
    \end{scope}
\end{tikzpicture}
\captionof{figure}{$w$ and $b$ both getting committed (impossible). Nodes horizontally arranged by view numbers.}
\label{fig:proof}
\end{figure*}
\begin{theorem}\label{thm:hotstuff-safety}
    Let $b$ and $w$ be two conflicting nodes. Then they cannot both become committed, each by an honest replica.
\end{theorem}
\begin{proof}
    We prove this theorem by contradiction.
    By an argument similar to Lemma~\ref{lm:basic}, $b$ and $w$ must be in different views.
    Assume during an exectuion $b$ becomes committed at some honest replica via the QC Three-Chain $b, b', b'', b^{*}$,
    likewise, $w$ becomes committed at some honest replica via the QC Three-Chain $w, w', w'', w^{*}$.
    Since each of $b, b', b'', w, w', w''$ get its QC, then w.l.o.g., we assume $b$ is created in a view higher than $w''$, namely, $\fViewNumber{\fJustify{b'}} > \fViewNumber{\fJustify{w^{*}}}$, as shown in Figure~\ref{fig:proof}.

    We now denote by $v_s$ the lowest view higher than $v_{w''} = \fViewNumber{\fJustify{w^{*}}}$ in which there is a $\qc_s$ such that $\fBranch{\qc_s}$ conflicts with $w$. Let $v_b = \fViewNumber{\fJustify{b'}}$.
    Formally, we define the following predicate for any $\qc$:

    \begin{align*}
        E(\qc) \coloneqq & (v_{w''} < \fViewNumber{\qc} \le v_b)
        \land (\fBranch{\qc}\textrm{ conflicts with }w).
    \end{align*}
    We can now set the \emph{first} switching point $\qc_s$:
    \[ \qc_s \coloneqq \argmin_{\qc}\{\fViewNumber{\qc} \mid \qc\textrm{ is valid} \land E(\qc)\}. \]

    By assumption, such $\qc_s$ exists, for example, $\qc_s$ could be $\fJustify{b'}$.
    Let $r$ denote a correct replica in the intersection of $\fJustify{w^{*}}$ and $\qc_s$.
    By assumption on the minimality of $\qc_s$, the lock that $r$ has on $w$ is not changed before $\qc_s$ is formed.
    Now consider the invocation of $\Call{safeNode}{}$ in view $v_s$ by $r$, with
    a message $m$ carrying a conflicting node $\msgBranch{m} = \fBranch{\qc_s}$.
    By assumption, the condition on the lock (Line~\ref{line:safeNode:conflict} in Algorithm~\ref{alg:util}) is false.
    On the other hand, the protocol requires $t = \fBranch{\fJustify{\msgBranch{m}}}$ to be an ancestor of $\fBranch{\qc_s}$. By minimality of $\qc_s$, $\fViewNumber{\fJustify{\msgBranch{m}}} \le v_{w''}$. Since $\fBranch{\qc_s}$ conflicts with $w$, $t$ cannot be $w, w'$ or $w''$. Then, $\fViewNumber{\fJustify{\msgBranch{m}}} < \fViewNumber{\fJustify{w'}}$, so the other half of the disjunct is also false. Therefore, $r$ will not vote for $\fBranch{\qc_s}$, contradicting the assumption of $r$.

\end{proof}
The liveness argument is almost identical to \HotStuffBasic, except that we
have to assume after GST, two consecutive leaders are correct, to guarantee a
decision. It is omitted for brevity.

\section{Proof of Safety for Implementation Pseudocode}
\begin{lemma}\label{lm:incumbent-basic}
    Let $\Lvl$ and $\lvl$ be two conflicting nodes such that
    $\fHeight{\Lvl} = \fHeight{\lvl}$, then they cannot both have valid quorum certificates.
\end{lemma}

\begin{proof}
    Suppose they can, so both $\Lvl$ and $\lvl$ receive $2f + 1$ votes, among which there
    are at least $f + 1$ honest replicas voting for each node,
    then there must be an honest replica that votes for both,
    which is impossible because $\Lvl$ and $\lvl$ are of the same height.
\end{proof}

\begin{notation}
    For any node $b$, let ``$\gets$'' denote parent relation, i.e. $\fParent{b} \gets b$. Let ``$\pathfrom$'' denote ancestry, that is, the reflexive transitive closure of the parent relation. Then two nodes $b, w$ are conflicting iff. $b \nopathfrom w \land w \nopathfrom b$. Let ``$\qcedge$'' denote the node a \QC{} refers to, i.e. $\fBranch{\fJustify{b}} \qcedge b$.
\end{notation}

\begin{lemma}\label{lm:incumbent-safety}
    Let $\Lvl$ and $\lvl$ be two conflicting nodes. Then they cannot both become 
committed, each by an honest replica.
\end{lemma}

\begin{proof}
    We prove this important lemma by contradiction.
    Let $\Lvlx$ and $\lvly$ be two conflicting nodes at different heights.
    Assume during an execution, $\Lvlx$ becomes committed at some honest
replica via the QC Three-Chain 
$\Lvlx \QCedge \Lvlxx \QCedge \Lvlxxx \qcedge \Lvl^{*}$; 
    likewise, $\lvly$ becomes committed at some honest
replica via the QC Three-Chain 
    $\lvly \QCedge \lvlyy \QCedge \lvlyyy \qcedge \lvl^{*}$. 
By Lemma~\ref{lm:basic}, since each of the nodes $\Lvlx, \Lvlxx, \Lvlxxx,
\lvly, \lvlyy, \lvlyyy$ have QC's, then
    w.l.o.g., we assume $\fHeight{\Lvlx} > \fHeight{\lvlyyy}$, as shown in Figure~\ref{fig:proof}.


We now denote by $\qc_s$ the QC for a node with the lowest height larger than $\fHeight{\lvlyyy}$, that conflicts with $w$.
Formally, we define the following predicate for any $\qc$:
\begin{align*}
    E(\qc) \coloneqq (\fHeight{\lvlyyy} < \fHeight{\fBranch{\qc}} \le \fHeight{\Lvl} )
    \land (\fBranch{\qc}\textrm{ conflicts with }\lvly)
\end{align*}

We can now set the first switching point $\qc_s$:
\[
    \qc_s \coloneqq \argmin_{\qc} \{ \fHeight{\fBranch{\qc}} \mid \qc\textrm{ is valid} \land E(\qc) \}.
\]

    By assumption, such $\qc_s$ exists, for example, $\qc_s$ could be $\fJustify{\Lvlxx}$.
    Let $r$ denote a correct replica in the intersection of $\fJustify{w^{*}}$ and $\qc_s$. By assumption of minimality of $\qc_s$, the lock that $r$ has on $w$ is not changed before $\qc_s$ is formed. Now consider the invocation of \Call{onReceiveProposal}{}, with a message carrying a conflicting node $\Lnew$ such that $\Lnew = \fBranch{\qc_s}$. By assumption, the condition on the lock (Line~\ref{line:hs-incumbent-conflict} in Algorithm~\ref{alg:hotstuff-incumbent-code}) is false. On the other hand, the protocol requires $t = \fBranch{\fJustify{\Lnew}}$ to be an ancestor of $\Lnew$. By minimality of $\qc_s$, $\fHeight{t} \le \fHeight{w''}$. Since $\fBranch{\qc_s}$ conflicts with $w$, $t$ cannot be $w, w'$ or $w''$. Then, $\fHeight{t} < \fHeight{w}$, so the other half of the disjunct is also false. Therefore, $r$ will not vote for $\Lnew$, contradicting the assumption of $r$.
\end{proof}

\begin{theorem}
    Let $\Cmd_1$ and $\Cmd_2$ be any two commands where $\Cmd_1$ is executed before $\Cmd_2$ by some honest replica, then any honest replica that executes $\Cmd_2$ must executes $\Cmd_1$ before $\Cmd_2$.
\end{theorem}

\begin{proof}
	Denote by $\lvl$ the node that carries $\Cmd_1$, $\Lvl$ carries $\Cmd_2$.
    From Lemma~\ref{lm:incumbent-basic}, it is clear the committed nodes are at distinct heights.
    Without loss of generality, assume $\fHeight{\lvl} < \fHeight{\Lvl}$. The commit of $\lvl$ are $\Lvl$ are
    triggered by some $\Call{onCommit}{\lvl'}$ and $\Call{onCommit}{\Lvl'}$ in $\Call{update}{}$,
    where $\lvl \pathfrom \lvl'$ and $\Lvl \pathfrom \Lvl'$. According to Lemma~\ref{lm:incumbent-safety},
    $\lvl'$ must not conflict with $\Lvl'$, so $\lvl$ does not conflict with $\Lvl$.
    Then $\lvl \pathfrom \Lvl$, and when any honest replica executes $\Lvl$, it must
    first executes $\lvl$ by the recursive logic in \Call{onCommit}{}.
\end{proof}

\subsection{Remarks}

In order to shed insight on the tradeoffs taken in the \HotStuff design, we explain why
certain constraints are necessary for safety.

\paragraph{Why monotonic vheight?}
Suppose we change the voting rule such that a replica does not need to vote
monotonically, as long as it does not vote more than once for each height. The weakened constraint
will break safety. For example, a replica can first vote for $\Lvlx$ and then $\lvly$.
Before learning about $\Lvlxx, \Lvlxxx$, it first delivers $\lvlyy, \lvlyyy$,
assuming the lock is on $\lvly$, and vote for $\lvlyyy$.
When it eventually delivers $\Lvlxxx$, it will flip to the branch led by $\Lvlx$ because it is
eligible for locking, and $\Lvlx$ is higher than $\lvly$. Finally, the replica will also vote for $\Lvlxxx$, causing
the commit of both $\lvly$ and $\Lvlx$.

\paragraph{Why direct parent?}
The direct parent constraint is used to ensure the equality $\fHeight{\Lvlx} >
\fHeight{\lvlyyy}$ used in the proof, with the help of Lemma~\ref{lm:incumbent-basic}. Suppose we do not enforce the rule for commit, so the
commit constraint is weakened to $\lvly \pathfrom \lvlyy \pathfrom \lvlyyy$ instead of $\lvly \gets \lvlyy \gets \lvlyyy$ (same for $\Lvl$).
Consider the case
where $\fHeight{\lvlyy} < \fHeight{\Lvlx} < \fHeight{\Lvlxx} < \fHeight{\lvlyyy} < \fHeight{\Lvlxxx}$. Chances are, a replica can first vote for $\lvlyyy$, and then discover $\Lvlxxx$ to switch to the branch by $\Lvlx$, but it is too late since $\lvly$ could be committed.

\end{appendices}
\end{document}